\documentclass[12pt]{article}
\usepackage[margin=1in]{geometry}

\usepackage{setspace}
\onehalfspace

\usepackage{amssymb,amsmath,amsfonts,graphicx, mathrsfs,bm,amsthm,bbm}
\usepackage{float}
\usepackage{subcaption}
\usepackage[round]{natbib}
\usepackage[title, titletoc]{appendix}
\usepackage{multirow}

\usepackage{hyperref}
\hypersetup{
    colorlinks,
    linkcolor={blue},
    citecolor={blue},
    urlcolor={blue}
}

\usepackage{fancyvrb}
\VerbatimFootnotes 

\newcommand{\Cov}{\mathrm{Cov}}

\DeclareMathOperator*{\argmin}{arg\,min}
\DeclareMathOperator{\tr}{tr}

\theoremstyle{plain}

\newtheorem{observation}{Observation}

\title{P-splines with an $\ell_1$ penalty for repeated measures}
\author{Brian D. Segal$^*$, Michael R. Elliott, Thomas Braun, and Hui Jiang \\
\normalsize Department of Biostatistics \\
\normalsize University of Michigan \\
\normalsize Ann Arbor, MI \\
\normalsize $^*$\href{mailto:bdsegal@umich.edu}{bdsegal@umich.edu}}
    
\begin{document}

\begin{singlespace}
\maketitle
\end{singlespace}

\begin{quotation}
P-splines are penalized B-splines, in which finite order differences in coefficients are typically penalized with an $\ell_2$ norm.  P-splines can be used for semiparametric regression and can include random effects to account for within-subject variability. In addition to $\ell_2$ penalties, $\ell_1$-type penalties have been used in nonparametric and semiparametric regression to achieve greater flexibility, such as in locally adaptive regression splines, $\ell_1$ trend filtering, and the fused lasso additive model. However, there has been less focus on using $\ell_1$ penalties in P-splines, particularly for estimating conditional means.

In this paper, we demonstrate the potential benefits of using an $\ell_1$ penalty in P-splines with an emphasis on fitting non-smooth functions. We propose an estimation procedure using the alternating direction method of multipliers and cross validation, and provide degrees of freedom and approximate confidence bands based on a ridge approximation to the $\ell_1$ penalized fit. We also demonstrate potential uses through simulations and an application to electrodermal activity data collected as part of a stress study.
\end{quotation}



\tableofcontents

\section{Introduction}

Many nonparametric regression methods, including smoothing splines and regression splines, obtain point estimates by minimizing a penalized negative log-likelihood function of the form $l_{\text{pen}} = -l(\bm{\beta}) + \lambda P(\bm{\beta})$, where $l$ is a log-likelihood, $P$ is a penalty term, $\lambda > 0$ is a smoothing parameter, and $\bm{\beta}$ are the coefficients to be estimated. Typically, quadratic ($\ell_2$ norm) penalties are used, which lead to straightforward computation and inference. In particular, $\ell_2$ penalties typically lead to ridge estimators, which have both closed form solutions and are linear smoothers. The $\ell_2$ penalty also has connections to mixed models, which allows the smoothing parameters to be estimated as variance components \citep{green1987, speed1991blup, wang1998mixed, zhang1998}.

However, nonparametric regression methods that use an $\ell_1$-type penalty, such as $\ell_1$ trend filtering 
\citep{kim2009ell1} and locally adaptive regression splines \citep{mammen1997locally}, are better able to adapt to local differences in smoothness and achieve the minimax rate of convergence for weakly differentiable functions of bounded variation \citep{tibshirani2014adaptive}, whereas $\ell_2$ penalized methods do not \citep{donoho1998}. The trade-off is that $\ell_1$ penalties generally lead to more difficult computation and inference because the objective function is convex but non-differentiable, and the fit is no longer a linear smoother.

In this article, we propose P-splines with an $\ell_1$ penalty as a framework for generalizing $\ell_1$ trend filtering within the context of repeated measures data and semiparametric (additive) models \citep{hastie1986generalized}. In Section \ref{sec_pspline_l1}, we discuss the connection between P-splines and $\ell_1$ trend filtering which motivates the methodological development. In Section \ref{sec_model}, we present our proposed model, and in Section \ref{sec_related}, we discuss related work. In Section \ref{sec_est}, we propose an estimation procedure using the alternating direction method of multipliers (ADMM) \citep[see][]{boyd2011distributed} and cross validation (CV). In Section \ref{sec_edf}, we derive the degrees of freedom and propose computationally stable and fast approximations, and in Section \ref{sec_inf}, we develop approximate confidence bands based on a ridge approximation to the $\ell_1$ fit. In Section \ref{sec_sim}, we study our method through simulations and evaluate its performance in fitting non-smooth functions. In section \ref{sec_app}, we demonstrate our method in an application to electrodermal activity data collected as part of a stress study. We close with a discussion in Section \ref{discussSec}.

\section{P-splines and $\ell_1$ trend filtering \label{sec_pspline_l1}}

In this section, we give brief background on P-splines and $\ell_1$ trend filtering, and show the relation between them when the data are independent and identically distributed (i.i.d.) normal.

P-splines \citep{eilers1996flexible} are penalized B-splines \citep[see][]{de1978practical}. B-splines are flexible bases that are notable in part because they have compact support, which leads to banded design matrices and faster computation. This compact support can be seen in Figure \ref{BsplineWide}, which shows eight evenly spaced first degree and third degree B-spline bases on $[0,1]$. We can define an order $M$ (degree $M-1$) B-spline basis with $j=1,\ldots,p$ basis functions recursively as \citep{de1978practical}
\begin{align*}
\phi_j^m(x) &= \frac{x - t_j}{t_{j+m-1} - t_j} \phi_j^{m-1}(x) + \frac{t_{j+m} - x}{t_{j+m} - t_{j+1}} \phi_{j+1}^{m-1}(x), && j = 1,\ldots, 2M + c -m, \\ & && 1 < m \le M \\
\phi_j^1(x) &= 
\begin{cases}
1 & t_j \le x < t_{j+1} \\
0 & \text{otherwise}
\end{cases}, && j = 1,\ldots,2M +c -1
\end{align*}
where $t_j$ are the knots, division by zero is taken to be zero, and $c$ is the number of internal knots. For order $M$ B-splines defined on the interval $[x_{\min}, x_{\max}]$, in order to obtain $j=1,\ldots,p$ basis functions, we set $2M$ boundary knots ($M$ knots on each side) and $c=p-M$ interior knots. In general, one can set $t_1 \le t_2 \le \cdots \le t_M = x_{\min} < t_{M+1} < \cdots < t_{M+c} < x_{\max} = t_{M+c+1} \le t_{M+c+2} \le \cdots \le t_{2M + c}$. In order to ensure continuity at the boundaries, we set $t_1 < t_2 < \cdots < t_{M-1} < t_M = x_{\min}$ and $x_{\max} = t_{M+c+1} < t_{M+c+2} < \cdots < t_{2M + c}$. We also use equally spaced interior knots, which is important for the P-spline penalty, and drop the superscript on $\phi$ designating order when the order does not matter or is stated in the text.

\begin{figure}[H]
\centering
\includegraphics[scale=0.55]{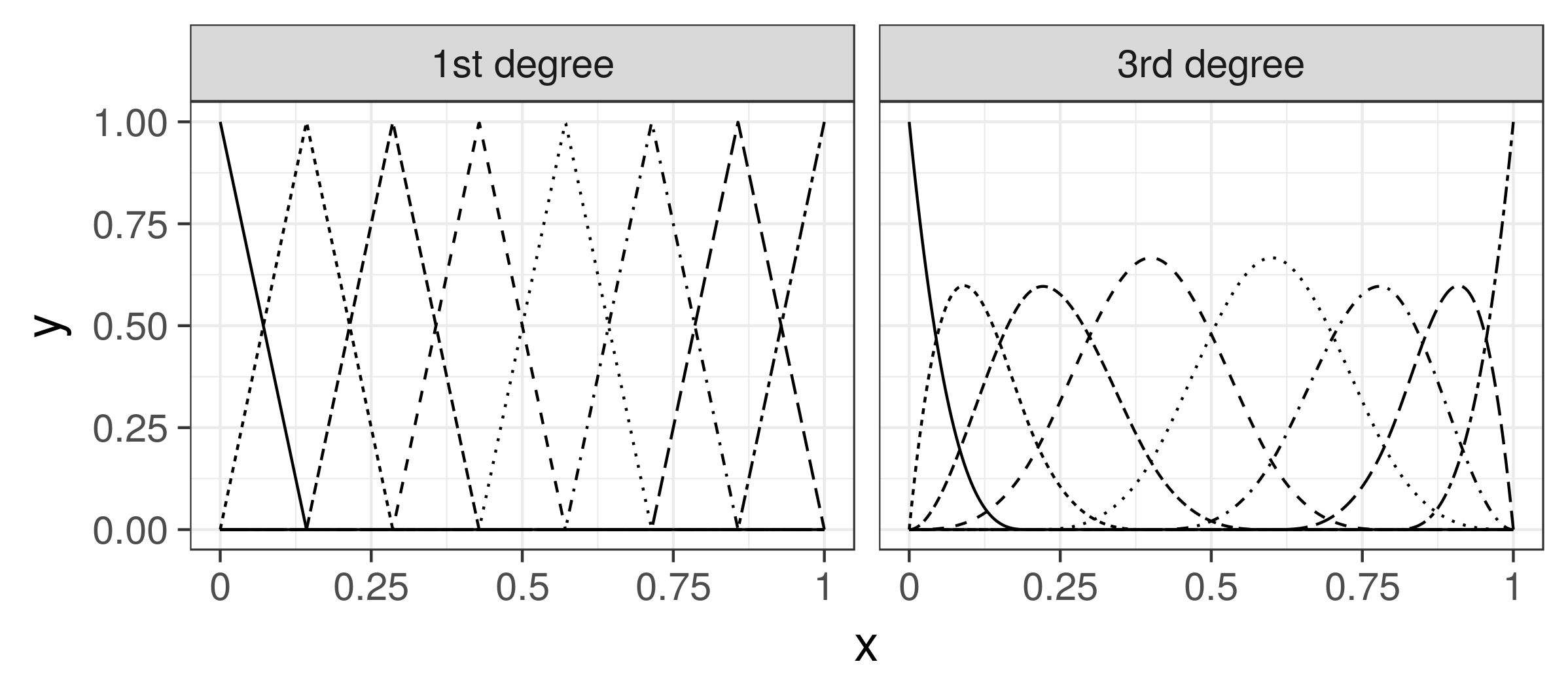}
\caption{Eight evenly spaced B-spline bases on $[0,1]$}
\label{BsplineWide}
\end{figure}

B-spline bases can be used to fit nonparametric models of the form $y(x) = f(x) + \epsilon(x)$, where $y(x)$ is the outcome $y$ at point $x$, $f(x)$ is the mean response function at $x$, and $\epsilon(x)$ is the error at $x$. To that end, let $\bm{y}=(y_1,\ldots,y_n)^T$ be an $n \times 1$ vector of outcomes and $\bm{x} = (x_1 \ldots, x_n)^T$ be a corresponding $n \times 1$ vector of covariates. Also, let $\phi_1,\ldots,\phi_p$ be B-spline basis functions and let $F$ be an $n \times p$ design matrix such that $F_{ij} = \phi_j(x_i)$, i.e., the $j^{th}$ column of $F$ is the $j^{th}$ basis function evaluated at $x_1,\ldots,x_n$. Equivalently, the $i^{th}$ row of $F$ is the $i^{th}$ data point evaluated by $\phi_1,\ldots,\phi_p$. For i.i.d. normal $\bm{y}$, a simple linear P-spline model with the standard $\ell_2$ penalty can be written as
\begin{equation}
\hat{\beta}_0, \bm{\hat{\beta}} = \argmin_{\beta_0 \in \mathbb{R}, \bm{\beta} \in \mathbb{R}^p} \frac{1}{2}\| \bm{y} - \beta_0 \bm{1} - F\bm{\beta} \|_2^2 + \frac{\lambda}{2} \|D^{(k+1)} \bm{\beta} \|_2^2,
\label{pspline}
\end{equation}
where $\beta_0$ is the intercept, $\bm{\beta}$ is a $p \times 1$ vector of parameter estimates, $\bm{1}$ is an $n \times 1$ vector with each element equal to 1, $\lambda>0$ is a smoothing parameter, and $D^{(k + 1)} \in \mathbb{R}^{(p-k-1) \times p}$ is the $k+1$ order finite difference matrix. For example, for $k=1$
\begin{equation}
D^{(2)}  =
\begin{bmatrix}
1 & -2 & 1 &  \\
 & 1 & -2 & 1 \\
  & &   \ddots & \ddots & \ddots \\
  & &  & 1 & -2 & 1 \\
&   & &  & 1 & -2 & 1
\end{bmatrix} \in \mathbb{R}^{(p-2) \times p}
\label{Dk}
\end{equation}
In general, as described by \citet{tibshirani2014adaptive}, $D^{(k+1)} = D^{(1)} D^{(k)}$ where $D^{(1)}$ is the $(p-k-1) \times (p-k)$ upper left matrix of:
\begin{equation}
D^{(1)} = 
\begin{bmatrix}
-1 & 1 &  \\
 & -1  & 1 &  \\
 & & \ddots & \ddots \\
 &  &  & -1 & 1
\end{bmatrix} \in \mathbb{R}^{(p-1) \times p}.
\label{D}
\end{equation}
Our proposed model builds on one in which the $\ell_2$ penalty in (\ref{pspline}) is replaced with an $\ell_1$ penalty:
\begin{equation}
\hat{\beta}_0, \bm{\hat{\beta}} = \argmin_{\beta_0 \in \mathbb{R}, \bm{\beta} \in \mathbb{R}^p} \frac{1}{2}\| \bm{y} - \beta_0 \bm{1} - F\bm{\beta} \|_2^2 + \lambda \|D^{(k+1)} \bm{\beta} \|_1.
\label{psplinel1}
\end{equation}

Letting $f(x) = \sum_{j=1}^p \beta_j \phi_j^M(x)$, for order $M=4$ B-splines, \citet{eilers1996flexible} show that 
\begin{equation*}
\int_{x_{\min}}^{x_{\max}} \left( \frac{d^{2}}{dx^{2}} f(x) \right)^2 dx = c_1 \|D^{(2)} \bm{\beta} \|_2^2 + c_2 \sum_{j=4}^p \nabla^2 \beta_j \nabla^2 \beta_{j-1}
\end{equation*}
where $\nabla^2$ is the second-order backwards difference and $c_1$ and $c_2$ are constants. As shown in Appendix \ref{controlSmooth}, a similar result holds for P-splines with an $\ell_1$ penalty. In particular, for $0 \le k < M-1$,
\begin{equation*}
\int_{x_{\min}}^{x_{\max}} \left| \frac{d^{k+1}}{dx^{k+1}} f(x) \right| dx \le C_{M,k+1} \|D^{(k+1)} \bm{\beta} \|_1
\end{equation*}
where $C_{M, k+1}$ is a constant given in Appendix \ref{controlSmooth} that depends on the order $M$ of the B-splines and order $k+1$ of the finite difference. In other words, controlling the $\ell_1$ norm of the $(k+1)^{th}$ order finite differences in coefficients also controls the total variation of the $k^{th}$ derivative of the function.

$\ell_1$ trend filtering is similar to (\ref{psplinel1}). In the case where $x_1 < x_2 < \cdots <x_n$ are unique and equally spaced, $
\ell_1$ trend filtering solves the following problem (the intercept is handled implicitly):
\begin{equation}
\bm{\hat{\beta}} = \argmin_{\bm{\beta} \in \mathbb{R}^n} \frac{1}{2}\| \bm{y} - \bm{\beta} \|_2^2 + \lambda \|D^{(k+1)} \bm{\beta} \|_1.
\label{ell1t}
\end{equation}
Problem (\ref{ell1t}) differs from (\ref{psplinel1}) in that (\ref{ell1t}) has one parameter per data point, and the design matrix is the identity matrix. $D^{(k+1)}$ is also resized appropriately by replacing $p$ with $n$ in the dimensions of (\ref{Dk}) and (\ref{D}). However, under certain conditions noted in Observation \ref{equiv}, (\ref{psplinel1}) and (\ref{ell1t}) are identical.

\begin{observation}[Continuous representation]
For second order (first degree) B-splines with $n$ basis functions, equally spaced data $x_1 < x_2 < \cdots < x_n$ with knots at $t_1 < x_1, t_2 = x_1, t_3=x_2,\ldots,t_{n} = x_{n-1}, t_{n+1} = x_n, t_{n+2} > x_n$, and centered outcomes such that $y(0) = 0$, P-splines with an $\ell_1$ penalty are a continuous analogue to $\ell_1$ trend filtering.
  \label{equiv}
\end{observation}

\begin{proof}[Proof of Observation \ref{equiv}]
Under these conditions, for $i=1,\ldots,n$
$$
\phi^2_j(x_i) = 
\begin{cases}
1 & i = j \\
0 & \text{otherwise}
\end{cases}.
$$
To see this, note that
\begin{align}
\phi_j^2(x_i) &= \frac{x_i - t_j}{t_{j+1} - t_j} \phi_j^{1}(x_i) + \frac{t_{j+2} - x_i}{t_{j+2} - t_{j+1}} \phi_{j+1}^{1}(x_i) \nonumber \\
&= \frac{t_{i+1} - t_j}{t_{j+1} - t_j} \phi_j^{1}(t_{i+1}) + \frac{t_{j+2} - t_{i+1}}{t_{j+2} - t_{j+1}} \phi_{j+1}^{1}(t_{i+1}). \label{eval}
\end{align}
Now, 
\begin{align*}
\phi_j^1(t_{i+1}) = 
\begin{cases}
1 & t_j \le t_{i+1} < t_{j+1}\\
0 & \text{otherwise}
\end{cases} \quad \text{and}
\quad
\phi_{j+1}^1(t_{i+1}) = 
\begin{cases}
1 & t_{j+1} \le t_{i+1} < t_{j+2}\\
0 & \text{otherwise}
\end{cases}.
\end{align*}
We have $\phi_j^1(t_{i+1}) = 1$ for $i=j-1$ and $0$ otherwise, but for $i=j-1$, we have $t_{i+1} - t_j = t_j - t_j = 0$. We also have $\phi_{j+1}^1(t_{i+1})=1$ for $i=j$ and 0 otherwise, and for $i=j$, we have $t_{j+2} - t_{i+1} = t_{j+2} - t_{j+1} > 0$. It follows that for $i=1\ldots,n$, (\ref{eval}) evaluates to 1 if $i=j$ and 0 otherwise.

Let $F$ be the design matrix in (\ref{psplinel1}), where $F_{ij} = \phi^2_j(x_i)$. Then from the previous result, we have $F=I_{n}$, where $I_{n}$ is the $n \times n$ identity matrix. This, together with the assumption that $\beta_0 = y(0) = 0$, implies that the objective functions (\ref{psplinel1}) and (\ref{ell1t}) are identical, which proves Observation \ref{equiv}.
\end{proof}

We note that \citet{tibshirani2014adaptive} shows that $\ell_1$ trend filtering has a continuous representation when expressed in the standard lasso form, and Observation \ref{equiv} gives a continuous representation of $\ell_1$ trend filtering when expressed in generalized lasso form.

$\ell_1$ trend filtering can be applied to irregularly spaced data, such as with the algorithm developed by \citet{ramdas2016fast}. It might also be possible to extend $\ell_1$ trend filtering to repeated measures data to account for within-subject correlations. However, due to Observation \ref{equiv}, we think it is beneficial to view $\ell_1$ trend filtering as a special case of P-splines with an $\ell_1$ penalty. We think this approach has the potential to be a general framework, because higher order B-splines could be used in combination with different order difference matrices just as can be done with P-splines that use the standard $\ell_2$ penalty. Furthermore, expressing $\ell_1$ trend filtering as P-splines with an $\ell_1$ penalty may facilitate the development of confidence bands (see Section \ref{sec_inf}), which could help to fill a gap in the $\ell_1$ penalized regression literature.

In addition, there are connections between P-splines with an $\ell_1$ penalty and locally adaptive regression splines. In particular, as \citet{tibshirani2014adaptive} shows, the continuous analogue of $\ell_1$ trend filtering is identical to locally adaptive regression splines \citep{mammen1997locally} for $k=0,1$, and asymptotically equivalent for $k \ge 2$.

\section{Proposed model: additive mixed model using P-splines with an $\ell_1$ penalty \label{sec_model}}

To introduce our model, let $\bm{y}_i = (y_{i1}, \ldots, y_{i n_i})^T$ be an $n_i \times 1$ vector of responses for subject $i=1,\ldots, N$, and let $\bm{y} = (\bm{y}_1^T, \ldots, \bm{y}_N^T)^T$ be the stacked $n \times 1$ vector of responses for all $N$ subjects, where $n= \sum_{i=1}^N n_i$. Let $\bm{x}_i = (x_{i1}, \ldots, x_{i n_i})^T$ be a corresponding $n_i \times 1$ vector of covariates for subject $i$, and $\bm{x} = (\bm{x}_1^T,\ldots,\bm{x}_N^T)^T$ be the $n \times 1$ stacked vector of all covariate values. In many contexts, $x$ is time. To account for the within-subject correlations of $\bm{y}_i$, we can incorporate random effects into the P-spline model. To that end, let $Z_i$ be an $n_i \times q_i$ design matrix for the random effects for subject $i$ (possibly including a B-spline basis), and let $\bm{b}_i = (b_{i1}, \ldots, b_{iq_i})^T$ be the corresponding $q_i \times 1$ vector of random effect coefficients for subject $i$. Also, let
$$
Z = \begin{bmatrix}
Z_1 & & \\
& \ddots \\
& & Z_N
\end{bmatrix}
$$ be the $n \times q$ block diagonal random effects design matrix for all subjects, where $q=\sum_{i=1}^N q_i$, and let $\bm{b} = (\bm{b}_1^T,\ldots,\bm{b}_N^T)^T$ be the $q \times 1$ stacked vector of random effects for all subjects. We propose an additive mixed model with $j=1,\ldots,J$ smooths (tildes denote quantities that will be subject to additional constraints, as described below):
\begin{align}
\underset{\beta_0 \in \mathbb{R}, \bm{b} \in \mathbb{R}^{q}, \bm{\tilde{\beta}}_j \in \mathbb{R}^{p_j}, j = 1,\ldots,J}{\text{minimize}} \; \; &\frac{1}{2}\| \bm{y} - \beta_0 \bm{1} - \sum_{j=1}^J \tilde{F}_j \bm{\tilde{\beta}}_j  - Z \bm{b} \|_2^2 + \sum_{j=1}^J \lambda_j \|\tilde{D}_j^{(k_j+1)} \bm{\tilde{\beta}}_j \|_1 \nonumber \\
&+ \tau \frac{1}{2}\bm{b}^T S \bm{b}
\label{l1PsplineAMM0}
\end{align}
where $\tilde{F}_j$ is a $n \times p_j$ design matrix of B-spline bases for smooth $j$, $\tilde{D}_j^{(k_j+1)}$ is the $k_j + 1$ finite difference matrix, and $\sigma^2_b S$ is the covariance matrix of the random effects $\bm{b}$. For example, if $\bm{b}$ are random intercepts, then $S=I_{N}$ and $Z$ would be an $n \times N$ matrix such that $Z_{il} = 1$ if observation $i$ belonged to subject $l$ and zero otherwise. Alternatively, to obtain random curves using smoothing splines and a B-spline basis, we could set
$$
S = \begin{bmatrix}
S_1 & & \\
& \ddots \\
& & S_N
\end{bmatrix}
$$
where $S_{j, il} = \int \phi''_{ji}(t) \phi''_{jl}(t) dt$, and $\phi''_{j1},\ldots,\phi''_{jp_j}$ are the second derivatives of the B-spline basis functions for the $j^{th}$ smooth. We would then set $Z$ to be the corresponding B-splines evaluated at the input points.

We note that (\ref{l1PsplineAMM0}) includes varying-coefficient models \citep{hastie1993varying}. For example, as pointed out by \citet[][p. 169]{wood2006generalized}, if $\tilde{F}_1$ are B-splines evaluated at $\bm{x}$, we could have $\tilde{F}_2 = \text{diag}(\bm{x'}) \tilde{F}_1$, where $\bm{x'} \ne \bm{x}$ is another covariate vector and $\text{diag}(\bm{x'})$ is a diagonal matrix with $x'_{i}$ at the $i^{th}$ leading diagonal position.

As written, (\ref{l1PsplineAMM0}) is not generally identifiable. To see this, suppose $\hat{y}(x) = \hat{\beta}_0 + \hat{f}_1(x) + \hat{f}_2(x)$, where neither $f_1$ nor $f_2$ are varying-coefficient terms. Then letting $\hat{f}'_1(x) = \hat{f}_1(x) + \delta$ and $\hat{f}'_2(x) = \hat{f}_2(x) - \delta$ for $\delta \in \mathbbm{R}$, we also have $\hat{y}(x) = \hat{\beta}_0 + \hat{f}'_1(x) + \hat{f}'_2(x)$. To make (\ref{l1PsplineAMM0}) identifiable, we follow \citet[][Section 4.2]{wood2006generalized} and introduce a centering constraint on each non-varying coefficient smooth, i.e. $\int \hat{f}_j(x) dx = 0$ for all smooths $j=1,\ldots, J$ such that $\tilde{F}_j \ne \text{diag}({\bm{x}'})\tilde{F}_l$ for some $\bm{x}'$ and $l \ne j$. To this end, let $\mathcal{E} = \{j \in \{1,\ldots, J\} : \tilde{F}_j \ne \text{diag}({\bm{x}'})\tilde{F}_l \text{ for some } \bm{x}', l \ne j\}$ be the indices of the non-varying coefficient smooths, and let $\bar{\mathcal{E}} = \{j \in \{1,\ldots, J\} : j \not \in \mathcal{E} \}$ be its complement. We constrain $\bm{1}^T \tilde{F}_j \bm{\tilde{\beta}}_j = 0$ for $j \in \mathcal{E}$. We accomplish this by defining new $p_j \times (p_j - 1)$ orthonormal matrices $Q_j$, $j=1, \ldots, J$, such that $\bm{1}^T \tilde{F}_j Q_j = \bm{0}$. If desired, one can also define a $q \times (q-1)$ matrix $Q_{J+1}$ such that $\bm{1}^T Z Q_{J+1} = \bm{0}$. However, this last centering constraint is not necessary, because the penalty on the random effect terms pulls the coefficients themselves towards zero, as opposed to the finite order differences in coefficients.

As \citet[][Section 1.8.1]{wood2006generalized} shows, $Q$ can be obtained by taking the QR decomposition of $\tilde{F}^T_j \bm{1}$ and retaining the last $p_j -1$ columns of the left orthonormal matrix.\footnote{The matrices $\bm{1}^T \tilde{F}_j$, $j=1,\ldots, J$ are of rank 1, so the remaining $p_j-2$ columns are arbitrary orthonormal vectors. In \textsf{R} \citep{R}, when taking the QR decomposition of $\tilde{F}^T \bm{1}$, an appropriate matrix $Q$ can be obtained as \verb|Q <- qr.Q(qr(colSums(F_tilde)), complete = TRUE)[, -1]|.} We can then re-parameterize the $p_j$ constrained parameters $\bm{\tilde{\beta}}_j$ in terms of the $p_j-1$ unconstrained parameters $\bm{\beta}_j$, such that $\tilde{\bm{\beta}}_j = Q_j \bm{\beta}_j$. For $j \in \mathcal{E}$, let $F_j = \tilde{F}_j Q_j$ and $D_j = \tilde{D}_j^{(k_j+1)}Q_j$. For $j \in \bar{\mathcal{E}}$, let $F_j = \tilde{F}_j$ and $D_j = \tilde{D}^{(k_j+1)}_j$. If centering the random effects, then we redefine $S := Q_{J+1}^T S Q_{J+1}$ and $Z := Z Q_{J+1}$. Then we can re-write (\ref{l1PsplineAMM0}) in the identifiable form
\begin{align}
\underset{\beta_0 \in \mathbb{R}, \bm{b} \in \mathbb{R}^{q}, \bm{\beta}_j \in \mathbb{R}^{p'_j}, j = 1,\ldots,J}{\text{minimize}} \; \; &\frac{1}{2}\| \bm{y} - \beta_0 \bm{1} - \sum_{j=1}^J F_j \bm{\beta}_j  - Z \bm{b} \|_2^2 + \sum_{j=1}^J \lambda_j \|D_j \bm{\beta}_j \|_1 \nonumber \\
&+ \tau \frac{1}{2}\bm{b}^T S \bm{b}
\label{l1PsplineAMM}
\end{align}
where $p_j' = p_j - 1$ for $j \in \mathcal{E}$ and $p_j' = p_j$ for $j \in \bar{\mathcal{E}}$.

We note that the penalty matrix $S$ given above for random subject-specific splines defines non-zero correlation between nearby within-subject random effect coefficients. This is in contrast to the approach of \citet{ruppert2003semiparametric} for estimating subject-specific random curves, which focuses on the case in which nearby within-subject coefficients are not correlated. To see this, let $\hat{d}_i(x) = \sum_{j=1}^{q_i} \hat{b}_{ij} \phi_{ij}(x)$ be the estimated difference between the $i^{th}$ subject-specific curve and the marginal mean at point $x$. The smoothing spline approach above constrains $\int (\hat{d}'')^2(x)dx = \bm{b}_i^T S_i \bm{b}_i <C$ for some constant $C>0$, whereas the approach of \citet{ruppert2003semiparametric} constrains $\bm{b}_i^T I_{q_i} \bm{b}_i = \sum_{j=1}^{q_i} \hat{b}_j^2 <C$. Whereas the non-diagonal penalty matrix $S$ implies correlations between nearby coefficients, the identity matrix in the approach of \citet{ruppert2003semiparametric} implies zero correlation.

Similar to the equivalence between Bayesian models and $\ell_2$ penalized smoothing splines \citep{wahba1990spline}, there is an equivalence between Bayesian models and $\ell_1$ penalized splines. In particular, (\ref{l1PsplineAMM}) is equivalent to the following distributional assumptions, which we can use to obtain Bayesian estimates:
\begin{align}
\bm{y}|\bm{b} &= \beta_0 \bm{1} + \sum_{j=1}^J F_j \bm{\beta}_j + Z \bm{b} + \bm{\epsilon} \nonumber \\
\bm{\epsilon} &\sim N \left( \bm{0}, \sigma^2_\epsilon I_n \right) \nonumber \\
\bm{b} &\sim N(\bm{0}, \sigma^2_b S^{-1}) \text{ for } \sigma^2_b = \sigma^2_\epsilon / \tau  \label{tauEquiv} \\
\bm{\epsilon} &\perp \bm{b} \nonumber \\
\left( D_j \bm{\beta}_j \right)_l &\sim \text{Laplace}(0, a_j) \text{ for } a_j = \sigma^2_\epsilon / (2 \lambda_j), l=1,\ldots, p_j - k_j - 1, j=1,\ldots,J \nonumber
\end{align}
The last distributional assumption is an element-wise Laplace prior on the $k_j+1$ order differences in coefficients.

In some cases, the random effects penalty matrix $S$ may be positive semidefinite but not invertible. For example, the smoothing spline random curves outlined above lead to a penalty matrix $S$ that is not strictly positive definite, but that is still positive semidefinite. This does not cause problems for the ADMM algorithm, but some changes are required for other algorithms as well as for Bayesian estimation. Following \citet[][Section 6.6.1]{wood2006generalized}, let $S=U \Lambda U^T$ be the eigendecomposition of a positive semidefinite matrix $S$, where $U U^T = I_{q}$ and $\Lambda$ is a diagonal matrix with eigenvalues in descending order in the diagonal positions. Let $\bm{\breve{b}} = U^T \bm{b}$ and $\breve{Z} = Z U$, so that $\bm{b}^T S \bm{b} = \bm{\breve{b}}^T \Lambda \bm{\breve{b}}$ and $\breve{Z} \bm{\breve{b}} = Z \bm{b}$. Let $q_r$ be the number of strictly positive eigenvalues of $S$, where $0 < q_r < q$, and let $\Lambda_r$ be the $q_r \times q_r$ upper left portion of $\Lambda$. We can partition $\bm{\breve{b}}$ as $\bm{\breve{b}} = (\bm{\breve{b}}_r^T, \bm{\breve{b}}_f^T)^T$, where $\bm{\breve{b}}_r^T$ is a $q_r \times 1$ vector of penalized coefficients and $\bm{\breve{b}}_f^T$ is a $q_f \times 1$ vector of unpenalized coefficients, where $q_r + q_f = q$. Then $\bm{\breve{b}}^T \Lambda \bm{\breve{b}} = \bm{\breve{b}}^T_r \Lambda_r \bm{\breve{b}}_r$, and it follows that $\bm{\breve{b}}_r \sim N(\bm{0}, \sigma^2_{b} \Lambda_r^{-1})$ and $\bm{\breve{b}}_f \propto \bm{1}$.

However, allowing for unconstrained random effect parameters leads to identifiability issues. Therefore, in practice if $q_f > 0$, we recommend using a normal or Cauchy prior on $\bm{\breve{b}}_f$. In particular, $\breve{b}_{f,l} \sim N(0, \sigma_f)$ or $\breve{b}_{f,l} \sim \text{Cauchy}(0, \sigma_f)$, $l=1,\ldots, q_f$ with either a diffuse prior on $\sigma_f$ and constraints to ensure $\sigma_f > 0$, or a diffuse prior on $\log(\sigma_f)$ without constraints. The Cauchy prior may be a preferable first choice, as it provides a weaker penalty and is similar to the recommendations of \citet{gelman2008weakly} for logistic regression. However, in some cases, such as in Section \ref{sec_app}, it is necessary to use a normal prior.

To further improve the computational efficiency of Monte Carlo sampling methods, we can partition $\breve{Z}$ into $\breve{Z} = [\breve{Z}_r, \breve{Z}_f]$ where $\breve{Z}_r$ contains the first $q_r$ columns of $\breve{Z}$ and $\breve{Z}_f$ contains the remaining $q_f$ columns. We then set $\bm{\check{b}}_r = \Lambda_r^{-1/2} \bm{\breve{b}}_r$ and $\check{Z}_r = \breve{Z}_r \Lambda_r^{1/2}$, so that $\check{Z}_r \bm{\check{b}}_r = \breve{Z}_r \bm{\breve{b}}_r$ and $\bm{\check{b}}_r \sim N (\bm{0}, \sigma^2_{b} I)$, which allows for more efficient sampling \citep{wood2006generalized}.

\section{Related work \label{sec_related}}

There are many nonparametric and semiparametric methods for analyzing repeated measures data. For an overview, please see \citet[][Part III]{fitzmaurice2008longitudinal}. However, most existing methods use an $\ell_2$ penalty \citep[e.g.][]{rice2001nonparametric, guo2002functional, chen2011penalized, scheipl2015functional}. 

Focusing on the optimization problem, our method puts a generalized lasso penalty \citep{tibshirani1996} on the fixed effects and a quadratic penalty on the random effects. Unlike the elastic net \citep{zou2005}, we do not mix the $\ell_1$ and $\ell_2$ penalties on the same parameters, though this could be done in the future.

The additive model with trend filtering developed by \citet{sadhanala2017} is similar to our approach. \citet{sadhanala2017} optimize 
\begin{align}
&\underset{\bm{\theta}_1,\ldots, \bm{\theta}_J \in \mathbb{R}^n}{\text{minimize}} \quad \frac{1}{2} \| \bm{y} - \bar{y} \bm{1} - \sum_{j=1}^J \bm{\theta}_j \|_2^2 + \lambda \sum_{j=1}^J \|D^{(k+1)} \bm{\theta}_j \|_1 \label{addtf} \\
&\text{subject to} \quad \bm{1}^T \bm{\theta}_j = \bm{0}, j=1,\ldots,J. \nonumber
\end{align}
In contrast to (\ref{l1PsplineAMM}), (\ref{addtf}) has one smoothing parameter and constrains all smooths to be zero-centered. From Observation \ref{equiv}, we see that (\ref{addtf}) is equivalent to (\ref{l1PsplineAMM}) when there are is $J=1$ smooth and no random effects, in which case there would be only one smoothing parameter $\lambda$ and no varying-coefficient smooths.

\citet{sadhanala2017} develop the theoretical and computational aspects of additive models with trend filtering, including the extension of the falling factorial basis to additive models. Similar to the B-spline basis, the falling factorial basis allows for linear time multiplication and inversion, which leads to fast computation \citep{wang2014}.

When smooths $j=1,\ldots,J$ are expected to have similar degrees of freedom and $n$ is not large enough to require dimension reduction, then (\ref{addtf}) with the addition of random effects and the relaxation of the zero-constraints for varying-coefficient smooths may be a viable alternative to (\ref{l1PsplineAMM}) that could potentially adapt better to local differences in smoothness because it would have one knot per data point.

While not developed for analyzing repeated measures, the fused lasso additive model (FLAM) \citep{petersen2016fused} is also similar to (\ref{l1PsplineAMM}). FLAM optimizes the following problem:
\begin{equation}
\underset{\theta_0 \in \mathbb{R}, \bm{\theta}_j \in \mathbb{R}^n, 1 \le j \le J}{\text{minimize}} \quad \frac{1}{2} \| \bm{y} - \theta_0 \bm{1} - \sum_{j=1}^J \bm{\theta}_j \|_2^2 + \alpha \lambda \sum_{j=1}^J \|D^{(1)} \bm{\theta}_j \|_1 + (1-\alpha) \lambda \sum_{j=1}^J \| \bm{\theta}_j \|_2
\label{flam}
\end{equation}
where $0 \le \alpha \le 1$ specifies the balance between fitting piecewise constant functions ($\alpha = 1$) and inducing sparsity on the selected smooths ($\alpha = 0$). From Observation \ref{equiv}, we see that (\ref{flam}) is equivalent to our model (\ref{l1PsplineAMM}) when: $\alpha = 1$, there is $J=1$ smooth, our design matrix has $p=n$ columns, our B-spline bases have appropriately chosen knots, and our model has no random effects. As \citet{petersen2016fused} show, FLAM can be a very useful method for modeling additive phenomenon, and as with the fused lasso \citep{tibshirani2005sparsity}, jumps in the piecewise linear fits have the advantage of being interpretable.

We also mention the sparse additive model (SpAM) \citep{ravikumar2009sparseadd} and sparse partially linear additive model (SPLAM) \citep{lou2016}. SpAM fits an additive model and uses a group lasso penalty \citep{yuan2006} to induce sparsity on the number of active smooths. SPLAM fits a partially linear additive model and uses a hierarchical group lasso penalty \citep{zhao2009grouplasso} to induce sparsity in the selected predictors and to control the number of nonlinear features.

One notable difference between our model and that of 
\citet{sadhanala2017}, as well as FLAM, SpAM, and SPLAM, is that we allow for multiple smoothing parameters. In our applied experience with additive models and standard $\ell_2$ penalties, we have found that in practice it can be important to allow for multiple smoothing parameters, particularly when the quantities of interest are the individual smooths as opposed to the overall prediction. This is equivalent to allowing each smooth to have different variance. However, this flexibility comes at a cost: estimating multiple smoothing parameters is currently the greatest challenge in fitting our proposed model. Perhaps due in part to these computational difficulties, several other authors also assume a single smoothing parameter in high-dimensional additive models \citep[e.g.][]{lin2006, meier2009}.

There are fast and stable methods for fitting multiple smoothing parameters for $\ell_2$ penalties paired with exponential family and quasilikelihood loss functions, notably the work of \citet{wood2004stable} using generalized cross validation (GCV) and \citet{wood2011} using restricted maximum likelihood. Furthermore, \citet{wood2015generalized} extend these methods to larger datasets, and \citet{wood2016smooth} extend these methods to likelihoods outside the exponential family and quasilikelihood form. However, similarly computationally efficient methods do not yet exist for fitting multiple smoothing parameters for $\ell_1$ penalties.

In addition to allowing for multiple smoothing parameters, we also propose approximate inferential methods, which is not typically provided for $\ell_1$ penalized models. \citet{yuan2006}, \citet{ravikumar2009sparseadd}, \citet{lou2016}, and \citet{petersen2016fused} focus on prediction and provide bounds on the prediction risk and related quantities. These are important results, and we think that distributional results for individual parameters and smooths will also be useful to practitioners.

We also note that \citet{eilers2000robust} and \citet{bollaerts2006quantile} discuss a variant of P-splines for quantile regression, in which the $\ell_1$ norm is used in both the loss and penalty function. However, we are not aware of existing P-spline methods that combine an $\ell_1$ penalty with an $\ell_2$ loss function.

\section{Point estimation \label{sec_est}}

\subsection{Regression parameters and random effects \label{admm}}

To fit (\ref{l1PsplineAMM}), we use the alternating direction method of multipliers (ADMM) \citep[see][]{boyd2011distributed}. ADMM has the advantage of being scalable to large datasets. To formulate (\ref{l1PsplineAMM}) for ADMM, we introduce constraint terms $\bm{w}_j$ and re-write the optimization problem as
\begin{align}
\text{minimize } \quad & \frac{1}{2} \| \bm{y} - \beta_0 \bm{1} - \sum_{j=1}^J F_j \bm{\beta}_j -Z\bm{b} \|_2^2 + \sum_{j=1}^J \lambda_j \|\bm{w}_j \|_1 + \frac{\tau}{2} \bm{b}^T S \bm{b} \label{ADMMobj} \\
\text{subject to } \quad & D_j \bm{\beta}_j - \bm{w}_j = \bm{0}, \; j = 1,\ldots, J \nonumber
\end{align}

The augmented Lagrangian in scaled form (using $\bm{u}$ to denote the scaled dual variable) is
\begin{align*}
L_{\rho}(\bm{\beta}, \bm{b}, \bm{w}, \bm{u}) &\propto \frac{1}{2} \| \bm{y} - \beta_0 \bm{1} - \sum_j F_j \bm{\beta}_j - Z \bm{b} \|_2^2 + \sum_j \lambda_j \left\|\bm{w}_j \right\|_1 \\
& \quad + \frac{\rho}{2} \sum_j \left\| D_j \bm{\beta}_j - \bm{w}_j + \bm{u}_j \right\|_2^2 + \frac{\tau}{2} \bm{b}^T S \bm{b}
\end{align*}
where $\rho > 0$ is the penalty parameter. The dimensions are 
$\bm{y} \in \mathbb{R}^{n \times 1}$,
$\beta_0 \in \mathbb{R}$,
$F_j \in \mathbb{R}^{n \times p'_j}$,
$\bm{\beta}_j \in \mathbb{R}^{p_j' \times 1}$,
$Z \in \mathbb{R}^{n \times q}$,
$\bm{b} \in \mathbb{R}^{q \times 1}$,
$D_j \in \mathbb{R}^{(p_j-k_j-1) \times p'_j}$,
$\bm{w}_j \in \mathbb{R}^{(p_j-k_j-1) \times 1}$,
$\bm{u}_j \in \mathbb{R}^{(p_j-k_j-1) \times 1}$, and 
$S \in \mathbb{R}^{q \times q}$,
where $p'_j = p_j - 1$ if $j \in \mathcal{E}$ (non-varying coefficient smooths) and $p'_j = p_j$ if $j \in \bar{\mathcal{E}}$ (varying coefficient smooths).

ADMM is an iterative algorithm, and we re-estimate the parameters for updates $m=1,2,\ldots$ until convergence.\footnote{We use $m$ to denote the iteration of the ADMM algorithm. This is unrelated to our use of $m$ in Section \ref{sec_pspline_l1} to denote the order of the B-spline basis.} It is straightforward to derive the $m+1$ updates \citep[see][Section 6.4.1]{boyd2011distributed}:
\begin{align}
\beta_0^{m+1} &= \frac{1}{n} \bm{1}^T \left( \bm{y} - \sum_j F_j \bm{\beta}_j^m - Z \bm{b}^m \right) \nonumber \\
\bm{\beta}^{m+1}_j &:= \argmin_{\bm{\beta}_j} L_{\rho}(\beta_0^{m+1}, \bm{\beta}_j, \bm{\beta}^{m+1}_{l<j}, \bm{\beta}^{m}_{l>j}, \bm{b}^m, \bm{w}^m, \bm{u}^m) \nonumber \\
  &= \left(F_j^T F_j + \rho D_j^T D_j \right)^{-1} \left(F_j^T \bm{y}^{(j,m)} + \rho D_j^T(\bm{w}^m_j - \bm{u}^m_j) \right) \nonumber \\
\bm{b}^{m+1} &:= \argmin_{\bm{b}} L_{\rho}(\bm{\beta}_{j=1,\ldots, J}^{m+1}, \bm{b}, \bm{w}^m, \bm{u}^m) \nonumber \\
  &= (Z^T Z + \tau S)^{-1}Z^T(\bm{y} - \beta_0^{m+1} \bm{1} - \sum_j F_j \bm{\beta}^{m+1}_j) \label{bUpdate} \\
\bm{w}_j^{m+1} &:= \argmin_{\bm{w}_j} L_{\rho}(\bm{\beta}^{m+1}_{j=1,\ldots,J}, \bm{b}^{m+1}, \bm{w}_j, \bm{u}^m) \nonumber \\
  &= \psi_{\lambda_j/\rho}(D_j \bm{\beta}_j^{m+1} + \bm{u}_j^m) \nonumber \\
\bm{u}_j^{m+1} &:= \bm{u}_j^m + D_j \bm{\beta}_j^{m+1} - \bm{w}_j^{m+1} \nonumber
\end{align}
where $\bm{y}^{(j,m)} = \bm{y} - \beta_0^{m+1} \bm{1} - \sum_{l < j} F_l \bm{\beta}^{m+1}_l - \sum_{l > j} F_j \bm{\beta}^m_l - Z\bm{b}^m$ and $\psi_{\lambda/\rho}$ is the element-wise soft thresholding operator, where for a single scalar element $x$
$$
\psi_{\lambda/\rho}(x) = 
\begin{cases}
x - \lambda/\rho & x > \lambda/\rho \\
0 & |x| \le \lambda / \rho \\
x + \lambda / \rho & x < -\lambda / \rho
\end{cases}
$$
To initialize the algorithm, we set $\beta_0 := \bar{y}$, $\bm{b} := \bm{0}$, and $\bm{\beta}_j := \bm{0}$, $\bm{w}_j := \bm{0}$, and $\bm{u}_j := \bm{0}$, for $j = 1,\ldots, J.$

As an alternative to the closed-form update (\ref{bUpdate}) for the random effects, it is also possible to update the random effects via a linear mixed effects (LME) model that is embedded into the ADMM algorithm. In particular, an LME model is fit to the residuals $\bm{y} - \beta_0^{m+1} \bm{1} - \sum_j F_j \bm{\beta}^{m+1}_j$, and $\bm{b}^{m+1}$ are updated as the best linear unbiased predictors (BLUPs). This update occurs at each step of the ADMM algorithm and replaces the update given by (\ref{bUpdate}). The LME update has the additional benefit of simultaneously estimating the variance of the random effects $\sigma^2_b$. In simulations, we have found that using an LME update leads to more accurate estimates of $\sigma^2_b$, which is important for subsequent estimates of degrees of freedom and confidence intervals.

For stopping criteria, we use the primal and dual residuals ($r^m$ and $s^m$, respectively):
\begin{align*}
r^m &= \begin{bmatrix}
D_1 \bm{\beta}_1^m - \bm{w}_1^m \\
\vdots \\
D_J \bm{\beta}_J^m - \bm{w}_J^m
\end{bmatrix}
\in \mathbb{R}^{\left(p-k - J \right) \times 1}\\
s^m &= -\rho
\begin{bmatrix}
D_1^T \left(\bm{w}_1^m - \bm{w}_1^{m-1} \right) \\
\vdots \\
D_J^T \left(\bm{w}_J^m - \bm{w}_J^{m-1} \right)
\end{bmatrix}
 \in \mathbb{R}^{p \times 1}
\end{align*}
where $k=\sum_{j=1}^J k_j$, $p = \sum_{j=1}^J p_j - |\mathcal{E}|$, and $|\mathcal{E}|$ is the cardinality of $\mathcal{E}$.

Following the guidance of \cite{boyd2011distributed}, we stop when $\|r^m\|_2 \le \epsilon^{\text{pri}}$ and $\|s^m\|_2 \le \epsilon^{\text{dual}}$, where
\begin{align*}
\epsilon^{\text{pri}} &= \epsilon^{\text{abs}} \sqrt{p - k - J} + \epsilon^{\text{rel}} \max \left\{ 
\left\|
\begin{matrix}
D_1 \bm{\beta}_1^m \\
\vdots \\
D_J \bm{\beta}_J^m
\end{matrix}
\right
\|_2,
\left\|
\begin{matrix}
\bm{w}^m_1 \\ 
\vdots \\
\bm{w}^m_J
\end{matrix}
\right\|_2
\right\} \\
\epsilon^{\text{dual}} &= \epsilon^{\text{abs}} \sqrt{p} + \epsilon^{\text{rel}} \rho
\left\|
\begin{matrix}
D_1^T \bm{u}_1^m \\
\vdots \\
D_J^T \bm{u}_J^m
\end{matrix}
\right\|_2.
\end{align*}
By default, we set $\epsilon^{\text{rel}} = \epsilon^{\text{abs}} = 10^{-4}$ and the maximum number of iterations at $1,000$.

\subsection{Smoothing parameters \label{smoothParams}}

To estimate $\lambda_1,\ldots,\lambda_J$ we compute cross validation (CV) error for a path of values one smoothing parameter at a time. In the CV, we split the sample at the subject level, as opposed to individual observations, and ensure that there are at least two subjects in each fold per unique combination of factor covariates. First, we estimate a path for $\tau$ with $\lambda_1,\ldots, \lambda_J$ set to 0. Then we fix $\tau$ at the value that minimizes CV error and compute a path for $\lambda_1$, setting it to the value that minimizes CV error, and so on.

We fit a path for each $\lambda_j$ from $\lambda_{j}^{\text{max}}$ to $10^{-5} \lambda_{j}^{\text{max}}$ evenly spaced on the log scale, where $\lambda_{j}^{\text{max}}$ is the smallest value at which $D_j \bm{\beta}_j = \bm{0}$. As shown in Appendix \ref{lambdaMaxDetails},\\
$\lambda_{j}^{\text{max}} = \|(D_j D_j^T)^{-1} D_j (F_j^T F_j)^{-1} F_j^T \bm{r}_j\|_\infty$, where $\bm{r}_j = \bm{y} - \beta_0 \bm{1} - \sum_{\ell \ne j} F_{\ell} \bm{\beta}_{\ell}  - Z \bm{b}$ are the $j^{th}$ partial residuals and for a vector $\bm{a}$, $\|\bm{a}\|_\infty = \max_j |a_j|$.

We also use warm starts, passing starting values separately for each fold, though warm starts appear to be minimally beneficial with ADMM. We set $\rho = \min( \max ( \lambda_1, \ldots, \lambda_J), c)$ at each iteration for some constant $c>0$ (e.g. $c=5$). When the number of smooths $J$ is small (e.g. $J \le 2$) a grid search is also feasible.

To estimate $\tau$, we can either use CV and the close-form update given by (\ref{bUpdate}), or an LME update that is embedded in the ADMM algorithm, as described in Section \ref{admm}. In simulations, we have found that the overall computation time to estimate the smoothing parameters is greater when using the LME update, and that the estimates of $\lambda_1, \ldots, \lambda_J$ do not appear sensitive to updates for $\bm{b}$. However, the final estimates of $\sigma^2_b$, and consequently the width of confidence intervals can be improved by using the LME update. Consequently, we recommend using cross validation to estimate $\tau$ for the purposes of then estimating $\lambda_1,\ldots,\lambda_J$, but using an LME update when estimating the final model.

With both the closed-form and LME update, we cannot use the training sample to estimate the random effect parameters $\bm{b}$ for the test sample, because these parameters are subject-specific and the test subjects are not included in the training sample. Instead, we use the training sample to obtain estimates for the fixed effect parameters $\beta_0$, $\bm{\beta}_j$, $j=1,\ldots,J$ and then use the test sample to estimate the random effects.

To make our approach clear, we first fix notation. Let $\mathcal{T}^r \subseteq \{1,\ldots,n\}$ be the row indices for the observations in the test sample for both the fixed and random effect design matrices $F_j$, $j=1,\ldots,J$, and $Z$. Also, let $\mathcal{T}^c \subseteq \{1,\ldots,q\}$ be the column indices of $Z$ for observations in the test sample, and let $\mathcal{T} = (\mathcal{T}^r, \mathcal{T}^c)$ be the tuple of row and column indices designating the test sample. Let matrices $F_{j,\mathcal{T}}$ and $F_{j,-\mathcal{T}}$ be matrix $F_j$ with only rows indexed by $\mathcal{T}^r$ retained and removed, respectively. Similarly, let matrices $Z_{\mathcal{T}}$ and $Z_{-\mathcal{T}}$ be matrix $Z$ with only rows and columns indexed by $\mathcal{T}^r$ and $\mathcal{T}^c$, respectively, retained and removed, respectively. Let matrices $S_{\mathcal{T}}$ and $S_{-\mathcal{T}}$ be matrix $S$ with only rows and columns indexed by $\mathcal{T}^c$ retained and removed, respectively. Also, let $\bm{y}_{\mathcal{T}}$ and $\bm{y}_{-\mathcal{T}}$ be vector $\bm{y}$ with elements indexed by $\mathcal{T}^r$ retained and removed, respectively.

We obtain out-of-sample marginal estimates as $\bm{\hat{\mu}}_{\mathcal{T}} = \hat{\beta}_0 \bm{1} + \sum_{j=1}^J F_{j,\mathcal{T}} \bm{\hat{\beta}}_j$, where $\hat{\beta}_0$ and $\bm{\hat{\beta}}_j$, $j=1,\ldots,J$ are estimated with $\bm{y}_{-\mathcal{T}}$, $F_{j, -\mathcal{T}}$, and $Z_{-\mathcal{T}}$. If using the closed-form update (\ref{bUpdate}), we estimate subject-specific random effects as $\bm{\hat{b}}_{\mathcal{T}} = \left( Z_{\mathcal{T}}^T Z_{\mathcal{T}}^T + \tau S_{\mathcal{T}} \right)^{-1} Z_{\mathcal{T}}^T (\bm{y}_{\mathcal{T}} - \bm{\hat{\mu}}_{\mathcal{T}})$ and obtain the out-of-sample prediction residuals as $\bm{r}_{\mathcal{T}} = \bm{y}_{\mathcal{T}} - \bm{\hat{\mu}}_{\mathcal{T}} - Z_{\mathcal{T}} \bm{\hat{b}}_{\mathcal{T}}$. Letting $\mathcal{T}_k$ be the tuple of indices for test sample (fold) $k=1,\ldots,K$, we obtain the CV error as $\sum_{k=1}^K \|\bm{r}_{\mathcal{T}_k} \|_2^2$.

\section{Degrees of freedom \label{sec_edf}}

In this section, we obtain the degrees of freedom, with the primary goal of estimating variance (see Section \ref{sec_var}). However, we note that degrees of freedom does not always align with a model's complexity in terms of its tendency to overfit the data \citep{janson2015effective}.

In each of the approaches described in this section, the degrees of freedom (df) is a function of the smoothing parameters $\lambda_1,\ldots \lambda_J$ and $\tau$. We always obtain the fixed effects smoothing parameters $\lambda_1,\ldots,\lambda_J$ from CV, but when using an LME update for the random effects $\bm{b}$ as described in Sections \ref{admm} and \ref{smoothParams}, we do not directly obtain $\tau$. Consequently, we cannot directly apply the results in this section to estimate df. However, from (\ref{tauEquiv}), we have that $\tau = \sigma^2_b / \sigma^2_\epsilon$. Writing $\text{df} = \text{df}(\tau)$, and letting $\bm{r} = \bm{y} - \sum_{j=1}^J F_j \bm{\hat{\beta}}_j - Z \bm{\hat{b}}$ be an $n \times 1$ vector of residuals and $\hat{\sigma}^2_{\epsilon} = \|\bm{r}\|_2^2 / (n - df(\tau))$ be an estimate of variance, we have that
\[
\hat{\tau} = \frac{\hat{\sigma}^2_b}{\hat{\sigma}^2_\epsilon} = \frac{\hat{\sigma}^2_b}{\|\bm{r}\|_2^2} \left(n - \text{df}(\hat{\tau}) \right).
\]
Therefore, letting
\[
\psi(\tau) = \tau - \frac{\hat{\sigma}^2_b}{\|\bm{r}\|_2^2} \left(n - \text{df}(\tau) \right),
\]
we numerically solve for $\hat{\tau}$ such that $\psi(\hat{\tau}) = 0$ and set $\text{df} = \text{df}(\hat{\tau})$.

\subsection{Stein's method \label{sec_dfDirect}}
 
Let $g(\bm{y}) = \hat{\bm{y}}$, where $g : \mathbb{R}^n \rightarrow \mathbb{R}^n$ is the model fitting procedure. For $\bm{y} \sim N(\mu, \sigma^2I)$, the degrees of freedom is defined as \citep[see][]{efron1986, hastie1990GAM}
\begin{equation}
\text{df} = \frac{1}{\sigma^2} \sum_{i=1}^n \Cov (g_i(y), y_i).
\label{dfDef}
\end{equation}
As \citet{tibshirani2014adaptive} notes, (\ref{dfDef}) is motivated by the fact that the risk $\text{Risk}(g) = \mathbb{E} \|g(\bm{y}) - \bm{\mu} \|_2^2$ can be decomposed as
$$
\text{Risk}(g) = \mathbb{E} \| g(\bm{y}) - \bm{y} \|_2^2 - n \sigma^2 + 2 \sum_{i=1}^n \Cov(g_i(\bm{y}), y_i).
$$
Therefore, the degrees of freedom (\ref{dfDef}) corresponds to the difference between risk and expected training error. Furthermore, if $g$ is continuous and weakly differentiable, then $ \text{df} = \mathbb{E}[\nabla \cdot g(y)]$ \citep{stein1981} where $\nabla \cdot g = \sum_{i=1}^n \partial g_i / \partial y_i$ is the divergence of $g$. Therefore, an unbiased estimate of $\text{df}$ (also used in Stein's unbiased risk estimate \citep{stein1981}) is 
\begin{equation}
\hat{\text{df}} = \sum_{i=1}^n \partial g_i / \partial y_i.
\label{dfStein}
\end{equation}

To obtain an estimate of degrees of freedom, we transform the generalized lasso component of our model to standard form, similar to the approach of \citet{petersen2016fused}. To do so, we use the following matrices described by \citet{tibshirani2014adaptiveSupp}. Let
$$
\tilde{D}^*_j = \begin{bmatrix}
\tilde{D}^{(0)}_{j,1} \\
\vdots \\
\tilde{D}^{(k_j)}_{j,1} \\
\tilde{D}^{(k_j+1)}_{j}
\end{bmatrix}
\in \mathbb{R}^{p_j \times p_j}
$$
be an augmented finite difference matrix, where $\tilde{D}^{(i)}_{j,1}$ is the first row of the finite difference matrix $\tilde{D}^{(i)}_j$, and $\tilde{D}^{(0)}_j = I_{p_j}$ is the identity matrix.
As shown by \citet{tibshirani2014adaptiveSupp}, the inverse of $\tilde{D}^*_j$ is given by $M_j = M_j^{(0)} M_j^{(1)} \cdots M_j^{(k)}$ where\footnote{We denote the inverse matrix as $M_j$. This is unrelated to our use of $M$ in Section \ref{sec_pspline_l1} to denote the order of the B-spline basis.}
$$
M_j^{(i)} = 
\begin{bmatrix}
I_{i} \\
 & L_{(p_j - i) \times (p_j - i)}
\end{bmatrix}
\in \mathbb{R}^{p_j \times p_j},
$$
where $L_{(p_j - i) \times (p_j - i)}$ is the $(p_j - i) \times (p_j - i)$ lower diagonal matrix of 1s.

Assuming our outcome $\bm{y}$ is centered, so that $\beta_0 = y(0) = 0$, and letting $V_j = \tilde{F}_j M_j$, $D_j^* = \tilde{D}_j^* Q_j$ for $j \in \mathcal{E}$ and $D_j^* = \tilde{D}_j^*$ for $j \in \bar{\mathcal{E}}$, and $\bm{\alpha}_j = D^*_j \bm{\beta}_j$, we can write the penalized log likelihood (\ref{l1PsplineAMM}) as 
\begin{equation}
l_{\text{pen}} = \frac{1}{2} \| \bm{y} - \sum_j V_j \bm{\alpha}_j - Z \bm{b} \|_2^2 + \sum_{j=1}^J \lambda_j \sum_{l=k_j+2}^{p_j} |\alpha_{jl} | + \frac{1}{2} \tau \bm{b}^T S \bm{b}.
\label{l1PsplineAMMV}
\end{equation}

To avoid difficulties later differentiating with respect to the $\ell_1$ norm, we remove the non-active $\ell_1$ penalized coefficients from (\ref{l1PsplineAMMV}). We also form the concatenated design matrix $V = [V_1,\ldots,V_J]$ and will need to index the active set of $V$. To these ends, let $\mathcal{A}_j = \{l \in \{k_j+2, \ldots, p'_j\} : \hat{\alpha}_{j,l} \ne 0 \}$ be the active set of the penalized coefficients for smooth $j$, and let $\mathcal{A}_j^* = \{1,\ldots,k_j + 1 \} \cup \mathcal{A}_j$ be the active set for smooth $j$ augmented with the unpenalized coefficients. Also, for a set $\mathcal{A}_j$ and constant $c \in \mathbb{R}$, let $\mathcal{A}_j + c = \{i + c : i \in \mathcal{A}_j \}$ be the set of elements in $\mathcal{A}_j$ shifted by $c$. Now let $\mathcal{A}^* = \bigcup_{j=1}^J (\mathcal{A}^*_j + \sum_{l =0}^{j-1} p'_l)$ be the augmented active set of $V$, where $p'_0 = 0$ and $p'_j, j=1,\ldots,J$ are the number of columns in $V_j$ (equivalently $F_j$). Finally, let $V_{\mathcal{A}^*}$ be matrix $V$ subset to retain only those columns indexed by $\mathcal{A}^*$. Similarly, let $\bm{\hat{\alpha}} = (\bm{\hat{\alpha}}_1^T, \ldots, \bm{\hat{\alpha}}_J^T)^T$ be the concatenated vector of estimated coefficients, and let $\bm{\hat{\alpha}}_{\mathcal{A}^*}$ be vector $\bm{\hat{\alpha}}$ subset to retain only elements indexed by $\mathcal{A}^*$. Then we can write the estimated penalized loss (\ref{l1PsplineAMMV}) as
\begin{equation}
\hat{l}_{\text{pen}} = \frac{1}{2} \left\| \bm{y} - [V_{\mathcal{A}^*}, Z] \begin{pmatrix} \bm{\hat{\alpha}}_{\mathcal{A}^*} \\ \bm{\hat{b}} \end{pmatrix} \right\|_2^2 + \sum_{j=1}^J \lambda_j \sum_{l=k_j+2}^{p_j} |\hat{\alpha}_{jl} | + \frac{1}{2} \tau \bm{\hat{b}}^T S \bm{\hat{b}}
\label{l1PsplineAMMVest}
\end{equation}

Taking the derivative of (\ref{l1PsplineAMMVest}) and keeping in mind that the first $k_j + 1$ elements of each $\bm{\hat{\alpha}}_j$ are unpenalized and $|\hat{\alpha}_{jl}| > 0$ for all $l \in \mathcal{A}_j$, we have
\begin{equation}
\bm{0}_{(|\mathcal{A}^*| + q) \times 1} = \frac{\partial l_{\text{pen}}}{\partial (\bm{\hat{\alpha}}_{\mathcal{A}^*}^T, \bm{\hat{b}}^T)^T} = \begin{bmatrix} V_{\mathcal{A}^*}^T \\  Z^T \end{bmatrix} \left([V_{\mathcal{A}^*}, Z] \begin{pmatrix} \bm{\hat{\alpha}}_{\mathcal{A}^*} \\ \bm{\hat{b}} \end{pmatrix} - \bm{y}\right)
+ 
\begin{pmatrix}
\bm{\eta} \\
\tau S \bm{\hat{b}}
\end{pmatrix}
\label{lpenDer}
\end{equation}
where 
$$
\bm{\eta} = 
\begin{bmatrix}
\bm{0}_{k_1+1} \\
\lambda_1 \; \text{sign}(\bm{\hat{\alpha}}_{\mathcal{A}_1}) \\
\bm{0}_{k_2+1} \\
\lambda_2 \; \text{sign}(\bm{\hat{\alpha}}_{\mathcal{A}_2 + p_1}) \\
\vdots \\
\bm{0}_{k_J+1} \\
\lambda_J \; \text{sign}(\bm{\hat{\alpha}}_{\mathcal{A}_J + \sum_{j=1}^{J-1} p_j})
\end{bmatrix},
$$
$\bm{0}_{k_j + 1}$ is a $(k_j + 1) \times 1$ vector of zeros, and the sign operator is taken element-wise. 

From \citet[][Lemmas 6 and 9]{tibshirani2012degrees}, we know that within a small neighborhood of $\bm{y}$, the active set $\mathcal{A}$ and the sign of the fitted terms $\hat{\alpha}_{\mathcal{A}}$ are constant with respect to $\bm{y}$ except for $\bm{y}$ in a set of measure zero. Therefore, $\partial \bm{\eta} / \partial \bm{y} = 0_{|\mathcal{A}^*| \times n}$, where $0_{|\mathcal{A}^*| \times n}$ is an $|\mathcal{A}^*| \times n$ matrix of zeros and $|\mathcal{A}^*|$ is the cardinality of $\mathcal{A}^*$. Then taking the derivative of (\ref{lpenDer}) with respect to $\bm{y}$, we have 
\begin{align*}
0_{(|\mathcal{A}^*| + q) \times n} = \frac{\partial^2 l_{\text{pen}}}{\partial (\bm{\hat{\alpha}}_{\mathcal{A}^*}^T, \bm{\hat{b}}^T)^T \partial \bm{y}} &= 
\begin{bmatrix} 
  V_{\mathcal{A}^*}^T \\  
  Z^T
\end{bmatrix} 
[V_{\mathcal{A}^*}, Z] 
\begin{bmatrix} 
  \partial \bm{\hat{\alpha}}_{\mathcal{A}^*} / \partial \bm{y} \\ 
  \partial \bm{\hat{b}} / \partial \bm{y}
\end{bmatrix} - 
\begin{bmatrix} 
  V_{\mathcal{A}^*}^T \\  
  Z^T
\end{bmatrix} \\
& \, + 
\begin{bmatrix}
0_{|\mathcal{A}^*| \times n} \\
\tau S (\partial \bm{\hat{b} / \partial \bm{y})}
\end{bmatrix}.
\end{align*}
Solving for the derivatives of the estimated coefficients, we have
$$
\begin{bmatrix} 
  \partial \bm{\hat{\alpha}}_{\mathcal{A}^*} / \partial \bm{y} \\ 
  \partial \bm{\hat{b}} / \partial \bm{y}
\end{bmatrix}
=
\left(
\begin{bmatrix} 
  V_{\mathcal{A}^*}^T \\  
  Z^T
\end{bmatrix} 
[V_{\mathcal{A}^*}, Z] 
+
\begin{bmatrix}
  0_{|\mathcal{A}^*| \times |\mathcal{A}^*|} & 0_{|\mathcal{A}^*| \times q} \\
  0_{q \times |\mathcal{A}^*|} & \tau S
\end{bmatrix}
\right)^{-1}
\begin{bmatrix} 
  V_{\mathcal{A}^*}^T \\  
  Z^T
\end{bmatrix}.
$$
Now let $A = [V_{\mathcal{A}^*}, Z]$ and 
\begin{equation*}
\Omega = \begin{bmatrix}
  0_{|\mathcal{A}^*| \times |\mathcal{A}^*|} & 0_{|\mathcal{A}^*| \times q} \\
  0_{q \times |\mathcal{A}^*|} & \tau S
\end{bmatrix}.
\end{equation*}
Then since
$\bm{\hat{y}} = A (\hat{\bm{\alpha}}_{\mathcal{A}^*}^T, \hat{\bm{b}}^T)^T$
we have
\begin{align*}
\frac{\partial \bm{\hat{y}}}{\partial \bm{y}} &= 
\frac{\partial \bm{\hat{y}}}{\partial (\bm{\hat{\alpha}}_{\mathcal{A}^*}^T, \bm{\hat{b}}^T)^T}
\frac{\partial (\bm{\hat{\alpha}}_{\mathcal{A}^*}^T, \bm{\hat{b}}^T)^T}{\partial \bm{y}} \\
 &=
A \left( A^T A + \Omega \right)^{-1} A^T.
\end{align*}

From \citet[][Lemmas 1 and 8]{tibshirani2012degrees}, we know that $g(\bm{y}) = \bm{\hat{y}}$ is continuous and weakly differentiable. Also, $\nabla g = \tr(\partial \bm{\hat{y}} / \partial \bm{y})$. Therefore, we can use Stein's formula (\ref{dfStein}) to estimate the degrees of freedom as
\begin{equation}
\hat{\text{df}} = 1 + \tr \left(A ( A^T A + \Omega \right)^{-1} A^T) = 1 + \tr \left( (A^T A + \Omega)^{-1} A^T A \right),
\label{dfAll}
\end{equation}
where we add 1 for the intercept. We note that this result is similar to the degrees of freedom for the elastic net \citep[see the remark on page 18 of][]{tibshirani2012degrees} as well as for FLAM \citep{petersen2016fused}.

To obtain degrees of freedom for individual smooths $j=1,\ldots, J$, let $E_j$ be an $(|\mathcal{A}^*|+q) \times (|\mathcal{A}^*|+q)$ matrix with 1s on the diagonal positions indexed by $\mathcal{A}^*_j + \sum_{l=0}^{j-1} |\mathcal{A}^*_l|$ and zero elsewhere, where $|\mathcal{A}^*_j|$ is the cardinality of $\mathcal{A}^*_j$ and $\mathcal{A}^*_0 = \emptyset$. Also, let $\hat{f}_j = V_j \bm{\hat{\alpha}}_j$ be the estimate of the $j^{th}$ smooth. Then as \citet{ruppert2003semiparametric} note, $\hat{f}_j = A E_j(A^TA + \Omega)^{-1} A^T \bm{y}$. Therefore, 
\begin{equation}
  \hat{\text{df}}_j = \tr \left(A E_j(A^T A + \Omega)^{-1} A^T \right) = \tr \left(E_j(A^T A + \Omega)^{-1} A^T A \right).
  \label{dfj}
\end{equation}
In other words, the degrees of freedom for smooth $j$ is the sum of the diagonal elements of $(A^T A + \Omega)^{-1} A^T A$ indexed by $\mathcal{A}_j^* + \sum_{l=0}^{j-1} |\mathcal{A}^*_l|$.

We note that when using the ADMM algorithm, or most likely any proximal algorithm, the fitted $D_j \bm{\hat{\beta}}_j$, or equivalently $\bm{\hat{\alpha}}_j$, will typically have several very small non-zero values, but will not typically be sparse. However, the vector $\bm{\hat{w}}_j$ is sparse, where in the ADMM algorithm we constrain $\bm{w}_j = D_j \bm{\beta}_j$. Therefore, in practice we use $\bm{w}_j$ to obtain the active set $\mathcal{A}_j$.

\subsection{Stable and fast approximations \label{stab_fast_df}}

In some cases, such as the application in Section \ref{sec_app}, the estimates based on Stein's method (\ref{dfAll}) and (\ref{dfj}) cannot be computed due to numerical instability. In this section, we propose alternatives that are more numerically stable and which are also more computationally efficient.

\subsubsection{Based on restricted derivatives}

In this approach, we take derivatives of the fitted values restricted to individual smooths. In particular, from Section \ref{sec_dfDirect}, we see that
\begin{align*}
\frac{\partial \bm{\hat{y}}}{\partial \bm{\hat{\alpha}}_{\mathcal{A}^*_j}}
\frac{\partial \bm{\hat{\alpha}}_{\mathcal{A}^*_j}}{\partial \bm{y}}
&= V_{\mathcal{A}_j^*} (V_{\mathcal{A}_j^*}^T V_{\mathcal{A}_j^*})^{-1} V_{\mathcal{A}_j^*}^T \\
\frac{\partial \bm{\hat{y}}}{\partial \bm{\hat{b}}}
\frac{\partial \bm{\hat{b}}}{\partial \bm{y}}
&= Z (Z^T Z + \tau S)^{-1} Z^T.
\end{align*}
We can then approximate the degrees of freedom for each individual smooth and the random effects by
\begin{equation}
\tilde{\text{df}}_j = 
\begin{cases}
\tr \left((V_{\mathcal{A}_j^*}^T V_{\mathcal{A}_j^*})^{-1} V_{\mathcal{A}_j^*}^T V_{\mathcal{A}_j^*} \right) & j=1,\ldots,J \\
\tr \left((Z^T Z + \tau S)^{-1} Z^T Z \right) & j=J+1
\end{cases}
\label{dfResj}
\end{equation}
We estimate the overall degrees of freedom as
\begin{equation}
\tilde{\text{df}} = 1 + \sum_{j=1}^{J+1} \tilde{\text{df}}_j
\label{dfRes}
\end{equation}
where we add 1 for the intercept.

This approach is similar to one described by \citet[][p. 176]{ruppert2003semiparametric}, though in a different context and for a different purpose. In particular, whereas we use this approach to approximate the degrees of freedom after fitting the model, \citet{ruppert2003semiparametric} use it to set the degrees of freedom before fitting the model in the context of $\ell_2$ penalized loss functions.

\subsubsection{Based on ADMM constraint parameters}

In this approach, we propose estimates of degrees of freedom specific to the ADMM algorithm. As in the previous section, this approach is based on estimates for the individual smooths. Consider the model with $J=1$ smooth, no random effects, and centered $\bm{y}$:
$$
\|\bm{y} - F \bm{\beta} \|_2^2 + \lambda \| D \bm{\beta} \|_1.
$$
Suppose we make the centering constraints described Section \ref{sec_model}, i.e. we set $F = \tilde{F} Q$ and $D = \tilde{D}^{(k+1)} Q$ for an $n \times p$ design matrix $\tilde{F}$, a $k+1$ order finite difference matrix $D^{(k+1)}$, and an orthonormal $p \times (p-1)$ matrix $Q$. Let $\mathcal{A} = \{l \in \{1,\ldots,p - k-1 \}: (D \bm{\hat{\beta}})_l \ne 0\}$ be the active set, and let $|\mathcal{A}|$ be its cardinality. In our context, we expect the design matrices $F$ to be full rank, in which case Theorem 3 of \citet{tibshirani2012degrees} (see the first Remark) states that the degrees of freedom is given by $\text{df} = \mathbb{E}[\text{nullity}(D_{-\mathcal{A}})]$. Here, $\text{nullity}(D)$ is the dimension of the null space of matrix $D$, and $D_{-\mathcal{A}}$ is matrix $D$ with rows indexed by $\mathcal{A}$ removed. Now, $D$ has dimensions $(p-k-1) \times (p-1)$, and we can see by inspection that for all $k<p-1$ the columns of $D$ are linearly independent. Therefore, the rank of $D_{-\mathcal{A}}$ is equal to the number of rows $p - k - 1 - |\mathcal{A}|$, and the nullity is equal to the number of columns $p-1$ minus the number of rows. This gives $\hat{\text{df}} = \text{nullity}(D_{-\mathcal{A}}) = k + |\mathcal{A}|$ for centered smooths, i.e. the number of non-zero elements of $D \bm{\hat{\beta}}$ plus one less than the order of the difference penalty. This is similar to the result for $\ell_1$ trend filtering, but we have lost one degree of freedom due to the constraint that $\bm{1}^T \tilde{F} \bm{\tilde{\beta}} = \bm{0}$. For uncentered smooths, $D$ has dimensions $(p-k-1) \times p$, which gives $\hat{\text{df}} = \text{nullity}(D_{-\mathcal{A}})) = k + 1 + |\mathcal{A}|$.

As before, we note that in the ADMM algorithm, $D \bm{\hat{\beta}}$ will not generally be sparse, as ADMM is a proximal algorithm. However, the corresponding $\bm{w}$ is sparse, where in the optimization problem we constrain $D \bm{\beta} = \bm{w}$. Now considering a model with smooths $j=1,\ldots,J$, a numerically stable and fast alternative to (\ref{dfj}) is given by
\begin{align}
  \tilde{\text{df}}_{j}^{\text{ADMM}}
   &= \mathbbm{1}\mathcal[j \in \bar{\mathcal{E}}] + k_j + \sum_{l=1}^{p - k - 1} \mathbbm{1} \left[w_{jl} \ne 0\right].
   \label{dfADMMj}
\end{align}
where $\bar{\mathcal{E}}$ indexes the un-centered smooths and $\mathbbm{1}$ is an indicator variable. We then combine (\ref{dfADMMj}) with the restricted derivative approximation for the degrees of freedom of the random effects given above to obtain the overall degrees of freedom
\begin{equation}
\tilde{\text{df}}^{\text{ADMM}} = 1 + \sum_{j=1}^J \tilde{\text{df}}^{\text{ADMM}}_j + \tr \left((Z^T Z + \tau S)^{-1} Z^T Z \right),
\label{dfADMM}
\end{equation}
where we add 1 for the intercept.

\subsection{Ridge approximation}

Let $U = [F_1, \ldots, F_J, Z]$ be the concatenated design matrix of fixed and random effects and
$$
\Omega^{\text{ridge}} = \begin{bmatrix}
\lambda_1 D_1^T D_1  \\
& \ddots \\
&& \lambda_J D_J^T D_J \\
&& & \tau S
\end{bmatrix}
$$
be the penalty matrix. Then the hat matrix from the linear smoother approximation (see Section \ref{sec_inf}) is given by $H = U (U^T U + \Omega^{\text{ridge}})^{-1} U^T$. Similar to before, we can get the overall degrees of freedom as
\begin{equation}
\hat{\text{df}}^{\text{ridge}} = 1 + \tr \left((U^T U + \Omega^{\text{ridge}})^{-1} U^T U \right),
\label{dfAllRidge}
\end{equation}
where we add 1 for the intercept. To obtain degrees of freedom for individual smooths $j=1,\ldots, J$, let $E_j$ be a $(p+q) \times (p+q)$ matrix with 1s on the diagonal positions indexed by the columns of $F_j$ and zero elsewhere. Also, let $\hat{f}_j = F_j \bm{\hat{\beta}}_j$ be the estimate of the $j^{th}$ smooth. Then the ridge approximation for smooth $j$ is given by $\hat{f}_j \approx U E_j(U^T U + \Omega^{\text{ridge}})^{-1} U^T \bm{y}$. Therefore, 
\begin{equation}
  \hat{\text{df}}^{\text{ridge}}_j = \tr \left( E_j(U^T U + \Omega^{\text{ridge}})^{-1} U^T U \right)
  \label{dfRidgej}
\end{equation}

Similar to before, we also propose stable and fast approximations to the ridge estimate of degrees of freedom based on restricted derivatives. In particular, let
\begin{equation}
\tilde{\text{df}}_j^{\text{ridge}} = 
\begin{cases}
\tr \left(({F_j}^T F_j + \lambda_j D_j^T D_j)^{-1} F_j^T F \right) & j=1,\ldots,J \\
\tr \left((Z^T Z + \tau S)^{-1} Z^T Z \right) & j=J+1
\end{cases}
\label{dfRidgeResj}
\end{equation}
Then we can estimate the overall degrees of freedom as
\begin{equation}
\tilde{\text{df}}^{\text{ridge}} = 1 + \sum_{j=1}^{J+1} \tilde{\text{df}}_j^{\text{ridge}}
\label{dfRidgeRes}
\end{equation}
where we add 1 for the intercept.

As noted above, this approach is similar to one described by \citet[][p. 176]{ruppert2003semiparametric}, though for a different purpose. Whereas we use this approach to obtain the degrees of freedom after fitting the model, \citet{ruppert2003semiparametric} use it to set the degrees of freedom before fitting the model.

\section{Approximate inference \label{sec_inf}}

In this section, we discuss approximate inferential methods based on ridge approximations to the $\ell_1$ penalized fit and conditional on the smoothing parameters $\lambda_j, j=1,\ldots,J$ and $\tau$. We use the ADMM algorithm to analyze the approximation. In particular, we note that we can write the ADMM update for $\bm{\beta}_j$ as
\begin{equation}
\bm{\beta}_j^{m+1} = \left(F_j^T F_j + \rho D_j^T D_j \right)^{-1} F_j^T \bm{y}^{(j,m)} + \bm{\delta}_j^m
\label{betaUpdate}
\end{equation}
where $\bm{\delta}_j^m = \rho(F_j^T F_j + \rho D_j^T D_j)^{-1}F^T_j D^T_j(\bm{w}^m_j - \bm{u}^m_j)$ and $\bm{y}^{(j,m)} = \bm{y} - \beta_0^{m+1} - \sum_{l<j} F_l \bm{\beta}_l^{m+1} - \sum_{l>j} F_l \bm{\beta}_l^m - Z \bm{b}^m$. As we note in Observation \ref{delta}, $\bm{\delta}_j$ loosely represents the difference in the estimate of $\bm{\beta}_j$ obtained with the $\ell_1$ and $\ell_2$ penalties.

\begin{observation}
With the $\ell_1$ penalty, i.e. $\|D_j \bm{\beta}_j\|_1$, in general $\bm{\delta}_j^{m} \ne \bm{0}$. However, with the $\ell_2$ penalty, i.e. $\|D_j \bm{\beta}_j\|_2^2$, and $\lambda_j = \rho$, we have $\bm{\delta}_j^{m} = \bm{0}$.
\label{delta}
\end{observation}

\begin{proof}[Proof of Observation \ref{delta}]
Similar to the ridge update for $\bm{b}$, if we changed $\lambda_j \|D_j \bm{\beta}_j\|_1$ to $(\lambda_j / 2) \|D_j \bm{\beta}_j\|_2^2$ in (\ref{l1PsplineAMM}) we could remove the $\bm{w}_j$ term and the constraint that $D_j \bm{\beta}_j^m = \bm{w}_j$ from (\ref{ADMMobj}) to obtain the ridge update $\bm{\beta}_j^{m+1} = \left(F_j^T F_j + \lambda_j D_j^T D_j \right)^{-1} F_j^T \bm{y}^{(j,m)}$. Then since we assumed $\lambda_j = \rho$, we have $\bm{\beta}_j^{m+1} = \left(F_j^T F_j + \rho D_j^T D_j \right)^{-1} F_j^T \bm{y}^{(j,m)}$. By comparison with (\ref{betaUpdate}), we see that $\bm{\delta}_j^m = \bm{0}$.
\end{proof}

Observation \ref{delta} motivates our approximate inferential strategy. Letting $\bm{\hat{f}}_j$ be the $j^{th}$ fitted smooth, and letting $\bm{y}^{(j)} = \bm{y} - \hat{\beta}_0 - \sum_{l \ne j} F_l \bm{\hat{\beta}}_l - Z \bm{\hat{b}}$, we have
\begin{align}
\bm{\hat{f}}_j = F_j \bm{\hat{\beta}}_j &= F_j(F_j^T F_j + \rho D_j^T D_j)^{-1} F_j^T \bm{y}^{(j)} + F_j \bm{\hat{\delta}}_j \label{l1fit} \\
&\approx F_j(F_j^T F_j + \rho D_j^T D_j)^{-1} F_j^T \bm{y}^{(j)} && (\text{assuming } F_j \bm{\hat{\delta}}_j \approx \bm{0}) \nonumber \\
&\approx F_j(F_j^T F_j + \lambda_j D_j^T D_j)^{-1} F_j^T \bm{y}^{(j)} && (\text{assuming } \lambda_j \approx \rho) \nonumber \\
&=H_j \bm{y}^{(j)} \label{linSmoothFit}
\end{align}
where $H_j = F_j(F_j^T F_j + \lambda_j D_j^T D_j)^{-1} F_j^T$. We obtain confidence intervals for the linear smoother (\ref{linSmoothFit}) centered around the estimated fit (\ref{l1fit}), ignore $F_j \bm{\delta}_j$ when estimating variance, and assume $\lambda_j \approx \rho$. We also condition on the smoothing parameters $\lambda_1,\ldots,\lambda_J$ and $\tau$.

Figure \ref{Hy_breakdown} gives a visual demonstration of the approximation for the simulation presented in Section \ref{sec_sim} and the application shown in Section \ref{sec_app}. As seen in Figure \ref{Hy_breakdown}, in these examples the $\ell_1$ fit and ridge approximation are very similar. If this holds in general, then this would suggest that 1) the approximate inferential procedures we propose might have reliable coverage probabilities, and 2) there may be minimal practical advantage to using an $\ell_1$ penalty instead of the standard $\ell_2$ penalty. However, as shown in Section \ref{sec_change_point}, the $\ell_1$ penalty appears to perform noticeably better in certain situations, including the detection of change points.

\begin{figure}[H]
  \begin{subfigure}{0.48\textwidth}
    \includegraphics[width = \linewidth]{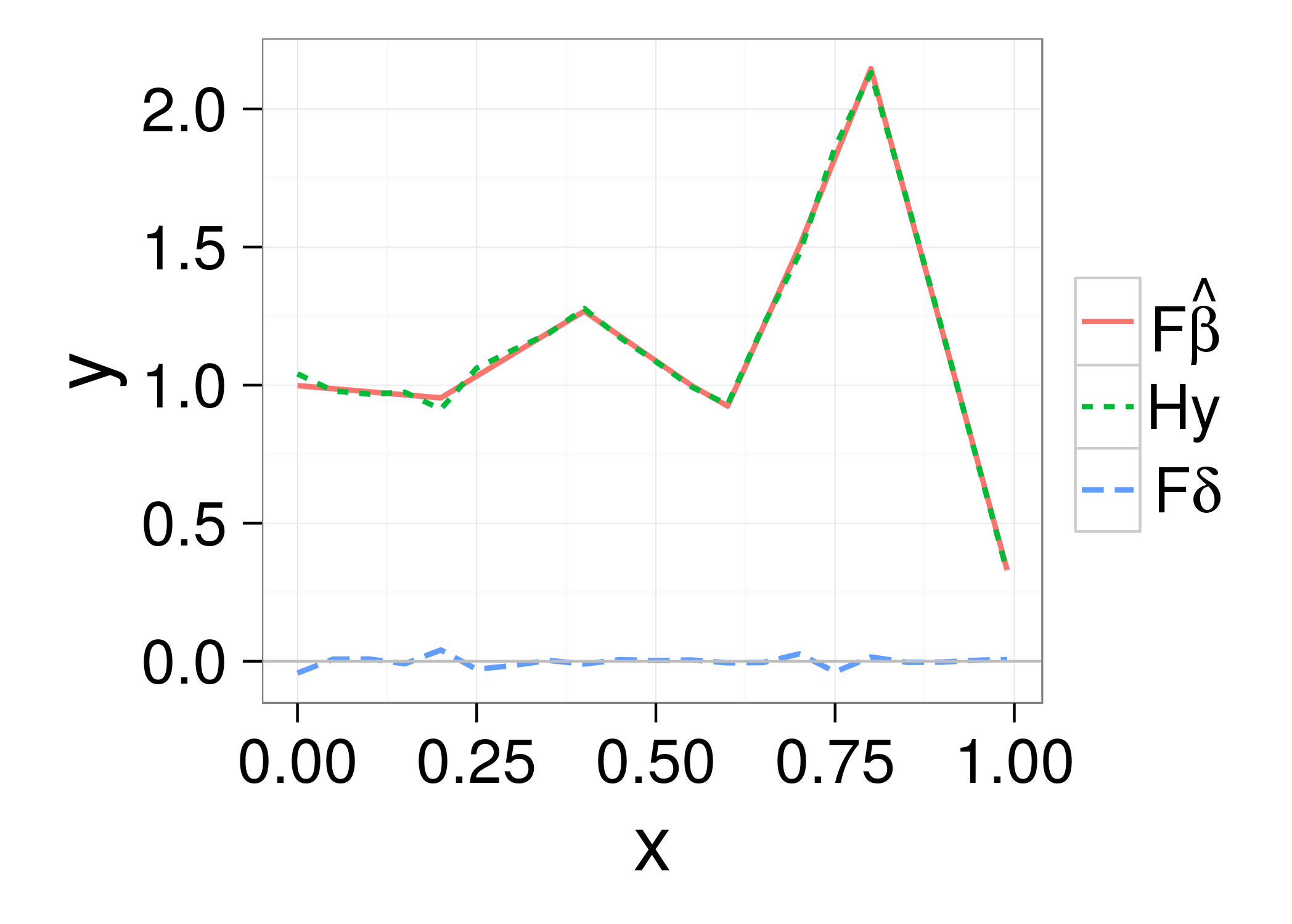}
    \caption{Simulation (Section \ref{sec_sim})}
  \end{subfigure}
  \begin{subfigure}{0.48\textwidth}
    \includegraphics[width = \linewidth]{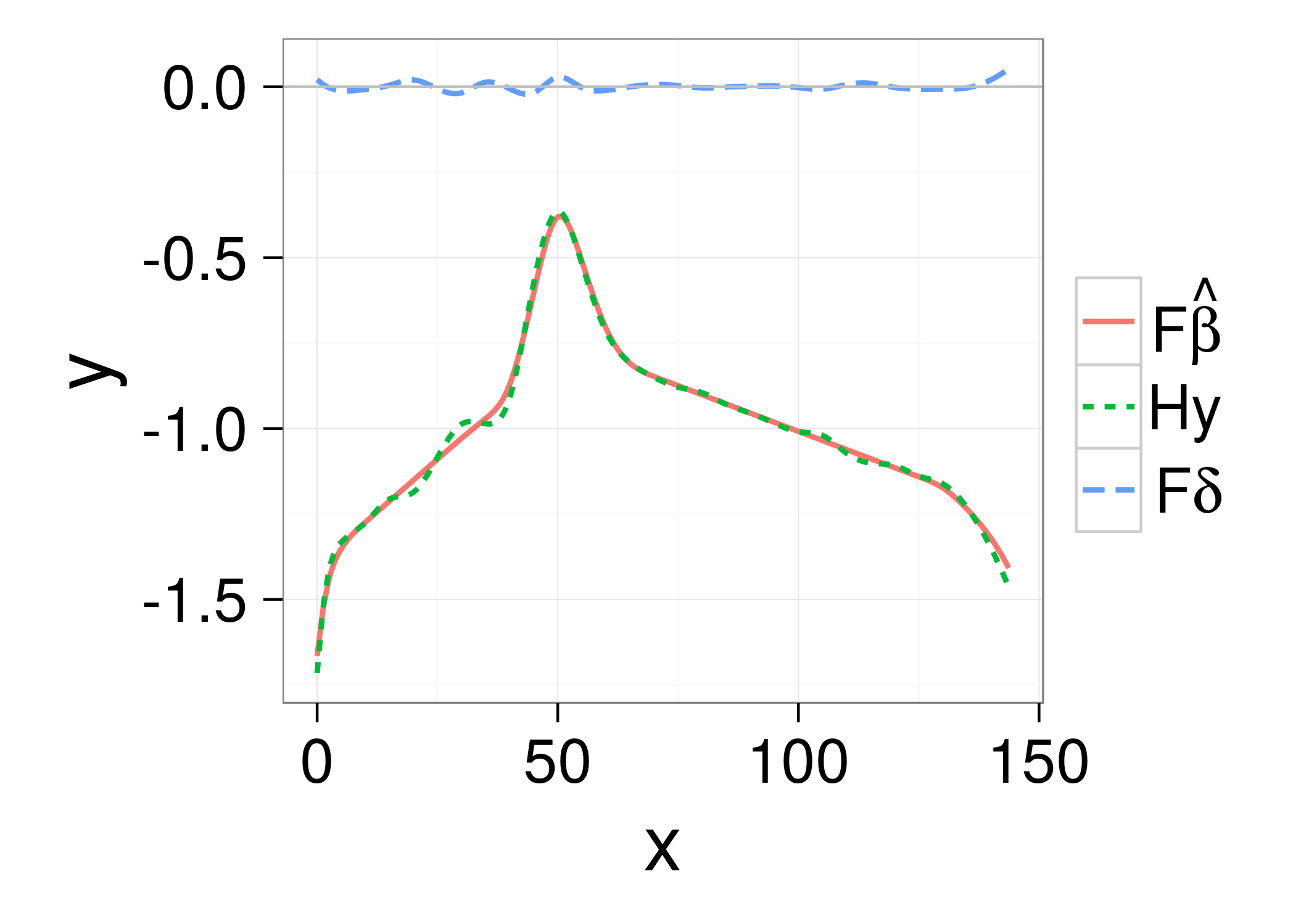}
    \caption{Application (Section \ref{sec_app})}
  \end{subfigure}
  \caption[Linear smoother approximation to the $\ell_1$ penalized fit]{Linear smoother approximation to the $\ell_1$ penalized fit in the simulation (see Section \ref{sec_sim}) and application (see Section \ref{sec_app}). The solid red line is the $\ell_1$ penalized fit, the dotted green line is the linear smoother approximation, and the dashed blue line is the difference between the two.}
  \label{Hy_breakdown}
\end{figure}

Before presenting the confidence bands in greater detail, we discuss our approach for estimating variance in Section \ref{sec_var}, which we then use to form confidence bands in Section \ref{sec_CI}.

\subsection{Variance \label{sec_var}}

Let $\bm{r} = \bm{y} - \sum_{j=1}^J F_j \bm{\hat{\beta}}_j - Z \bm{\hat{b}}$ be an $n \times 1$ vector of residuals. We estimate the overall variance as $\hat{\sigma}^2_\epsilon = \|\bm{r}\|^2_2 / \hat{\text{df}}_{\text{resid}}$, where $\hat{\text{df}}_{\text{resid}}$ is the residual degrees of freedom. When possible, we use the estimate based on Stein's method (\ref{dfAll}) and set $\hat{\text{df}}_{\text{resid}} = n - \hat{\text{df}}$. If Stein's method is not numerically stable, then we use the restricted derivatives approximation (\ref{dfRes}) and set $\hat{\text{df}}_{\text{resid}} = n - \tilde{\text{df}}$. As another alternative, we could also use the ADMM approximation and set $\hat{\text{df}}_{\text{resid}} = n - \tilde{\text{df}}^{\text{ADMM}}$.

\subsection{Confidence bands \label{sec_CI}}

In this section, we obtain confidence bands for typical subjects, i.e. for subjects for whom $\bm{b}_i = \bm{0}$. Since we assume a normal outcome, this is equivalent to the marginal population level response.

\subsubsection{Frequentist confidence bands}

Ignoring the distribution on $D_j \bm{\beta}_j$ and treating $\bm{\beta}_l$, $l \ne j$ as fixed, $\bm{y}^{(j)}$ is normal with variance $\text{Var}(\bm{y}^{(j)}) = \sigma^2_\epsilon I_n + \sigma^2_b Z S^{+} Z^T$, where $S^+$ is the Moore-Penrose generalized inverse of matrix $S$ (as noted in Section \ref{sec_model}, $S$ may not be positive definite). Therefore, $\widehat{\text{Var}}(\bm{\hat{f}}_j) \approx H_j \widehat{\text{Var}}(\bm{y}^{(j)}) H_j^T$ where $\widehat{\text{Var}}(\bm{y}^{(j)})$ is an $n \times n$ estimate of $\text{Var}(\bm{y}^{(j)})$ with $\hat{\sigma}^2_{\epsilon}$ and $\hat{\sigma}^2_b$ plugged in for $\sigma^2_{\epsilon}$ and $\sigma^2_b$ respectively, and $\bm{\hat{f}}_j \overset{\cdot}{\sim} N(\bm{\hat{f}}_j, H_j \widehat{\text{Var}}(\bm{y}^{(j)}) H_j^T)$. The estimated variance of the fit at a single point $x$, which we denote as $\widehat{\text{Var}}(\hat{f}_j(x))$, is the corresponding diagonal element of $H_j \widehat{\text{Var}}(\bm{y}^{(j)}) H_j^T$. Therefore, asymptotic pointwise $1-\alpha$ confidence bands take the form $\hat{f}_j(x) \pm z_{1-\alpha/2} \sqrt{\widehat{\text{Var}}(\hat{f}_j(x))}$ where $\Phi(z_a) = a$ and $\Phi$ is the standard normal CDF, e.g. $z_{1-\alpha/2} = 1.96$ for $\alpha = 0.05$.

For the purposes of interpretation, we include the intercept term in the confidence band for the $j=1$ smooth, but not for the remaining smooths.

\subsubsection{Bayesian credible bands \label{BayesCB}}

Many authors, including \citet{wood2006generalized}, recommend using Bayesian confidence bands for nonparametric and semiparametric models, because the point estimates are themselves biased. While Bayesian credible bands do not remedy the bias, they are self consistent.

To this end, we replace the element-wise Laplace prior with the (generally improper) joint normal prior that is equivalent to the standard $\ell_2$ penalty: $\bm{\beta}_j \sim N \left(\bm{0}, ( \lambda_j D_j^T D_j)^{-1} \right)$. This leads to the posterior
\begin{equation}
\bm{\beta}_j | \bm{y} \overset{\cdot}{\sim} N \left(\bm{\hat{\beta}}_j, (\underset{W_j}{\underbrace{F^T_j \widehat{\text{Var}}(\bm{y}^{(j)})^{-1} F_j + \lambda_j D^T_j D_j}})^{-1} \right).
\label{bayesPost}
\end{equation}
We can then form simultaneous Bayesian credible bands for $\bm{f}_j | \bm{y}$ by simulating from the posterior (\ref{bayesPost}) and taking quantiles from  $F_j \bm{\beta}_j^b, b = 1,\ldots,B$. Alternatively, for a faster approximation we use frequentist confidence bands with $F_j W_j^{-1} F^T_j$ in place of $H_j \widehat{\text{Var}}(\bm{y}^{(j)}) H_j^T$. In practice, we have found the simultaneous credible bands and the faster approximation to be nearly indistinguishable.\footnote{It appears that the latter (faster) method is the default in the \verb|mgcv| package \citep{wood2006generalized}. As in \verb|mgcv|, we only need to compute the diagonal elements of $F_j W_j^{-1} F^T_j$ as \verb|rowSums|$((F_j W_j^{-1}) \circ F_j)$, where $\circ$ is the Hadamard (element-wise) product.}

As before, for the purposes of interpretation, we include the intercept term in the credible band for the $j=1$ smooth, but not for the remaining smooths.

\section{Simulation \label{sec_sim}}

We simulated data from a piecewise linear mean curve as shown in Figure \ref{nonDiff_smallN}. Each subject had a random intercept and is observed over only a portion of the domain. There are 50 subjects, each with between 4 and 14 measurements (450 total observations). The random intercepts were normally distributed with variance $\sigma_b^2 = 1$, and the overall noise was normally distributed with variance $\sigma_\epsilon^2 = 0.01$.

\begin{figure}[H]
\centering
\includegraphics[scale = 0.4]{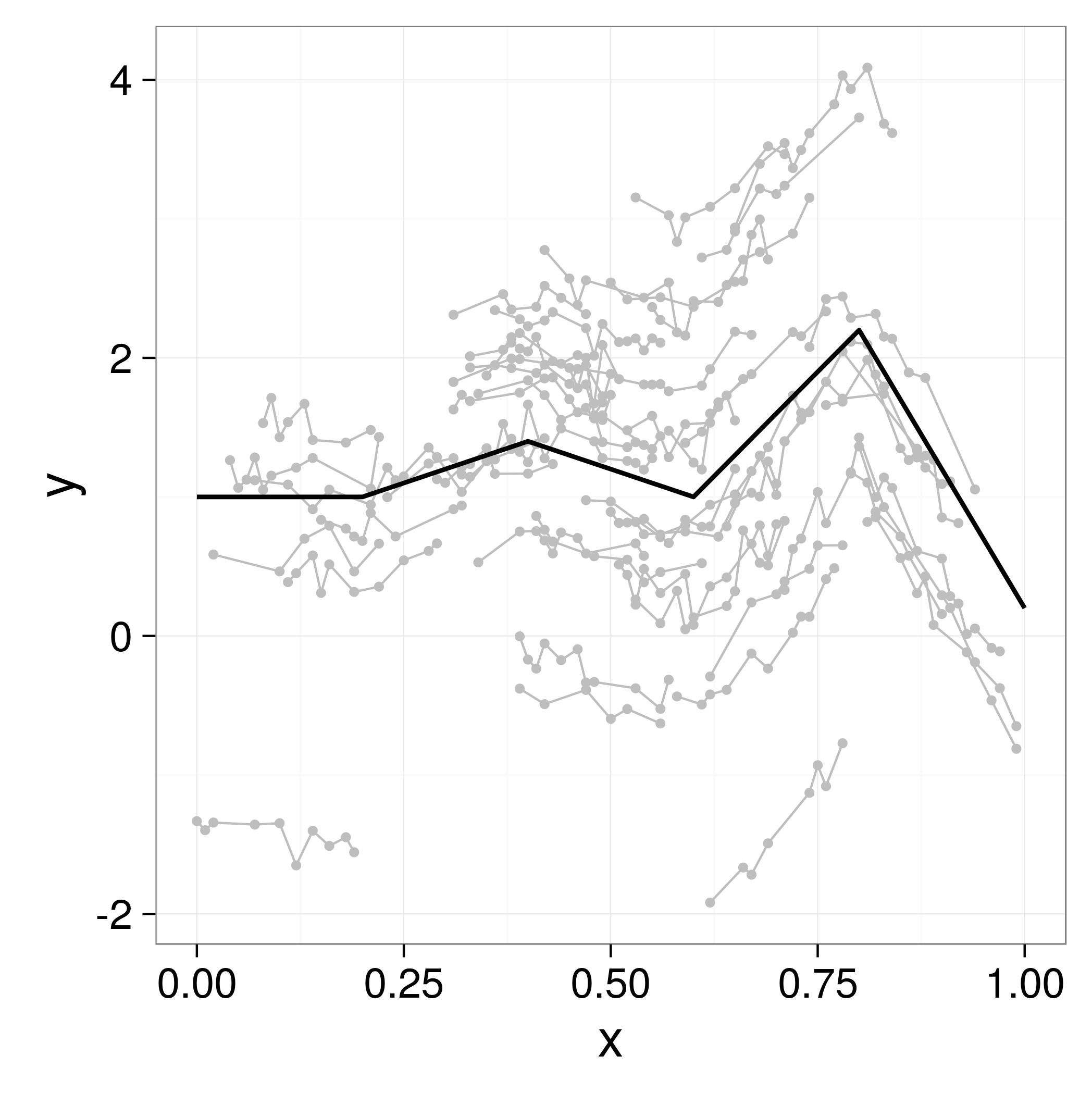}
\caption[Simulated data]{Simulated data: true marginal curve in black, observed (simulated) data in gray.}
\label{nonDiff_smallN}
\end{figure}

In all models, we used order 2 (degree 1) B-splines with $p=21$ basis functions.

\subsection{Frequentist estimation}

We fit models with $J=1$ smooth term and random intercepts. To obtain estimates for the $\ell_1$ penalized model, we used ADMM and 5-fold CV to minimize
\begin{equation}
\underset{\beta_0 \in \mathbb{R}, \bm{\beta} \in \mathbb{R}^{p-1}, \bm{b} \in \mathbb{R}^{N}}{\text{minimize}} \frac{1}{2} \| \bm{y} - \beta_0 \bm{1} - F \bm{\beta} - Z \bm{b} \|_2^2 + \lambda \|D^{(2)} \bm{\beta} \|_1 + \tau \bm{b}^T \bm{b}.
\label{sim_obj}
\end{equation}
where $Z_{il} = 1$ if observation $i$ belongs to subject $l$ and zero otherwise. As noted above, we used order 2 (degree 1) B-splines with $p=21$ basis functions, i.e. $F \in \mathbb{R}^{n \times (p-1)}$ where $n=450$ and $p=21$. After estimating $\lambda$ and $\tau$ via CV, we used LME updates to estimate $\sigma^2_b$ and $\bm{b}$ in the final model. We also fit an equivalent model with an $\ell_2$ penalty using the \verb|mgcv| package \citep{wood2006generalized}, i.e. with $(\lambda / 2) \|D^{(2)} \bm{\beta} \|_2^2$ in place of $\lambda \|D^{(2)} \bm{\beta} \|_1$ in (\ref{sim_obj}). Figure \ref{sim_point} shows the marginal mean with 95\% credible intervals, and Figure \ref{sim_subj} shows the subject-specific predicted curves.

\begin{figure}[H]
\centering
  \begin{subfigure}{0.48\linewidth}
    \includegraphics[width=\textwidth]{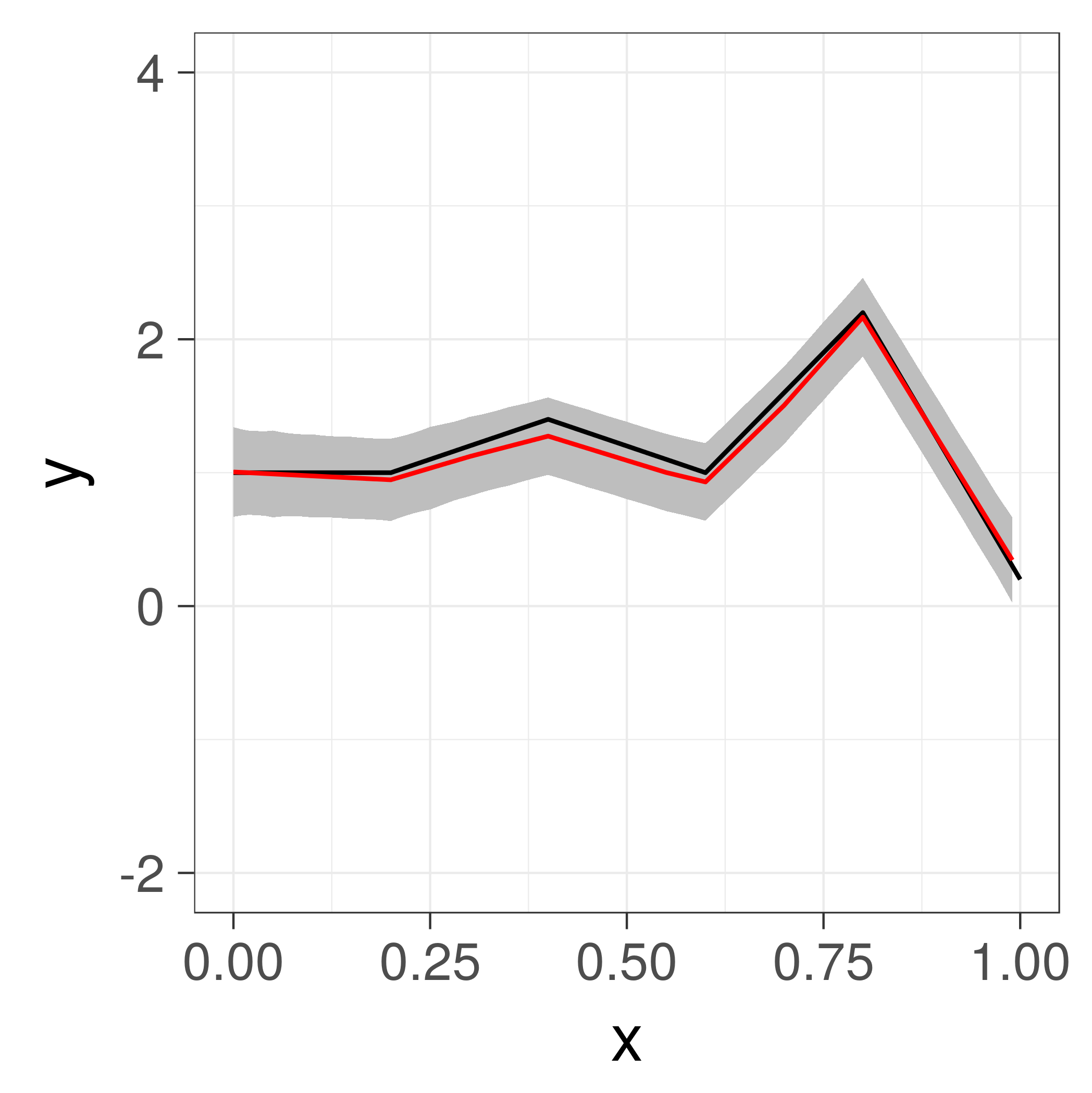}
    \caption{$\ell_1$ fit with ADMM and CV}
  \end{subfigure}
  \begin{subfigure}{0.48\linewidth}
    \includegraphics[width=\textwidth]{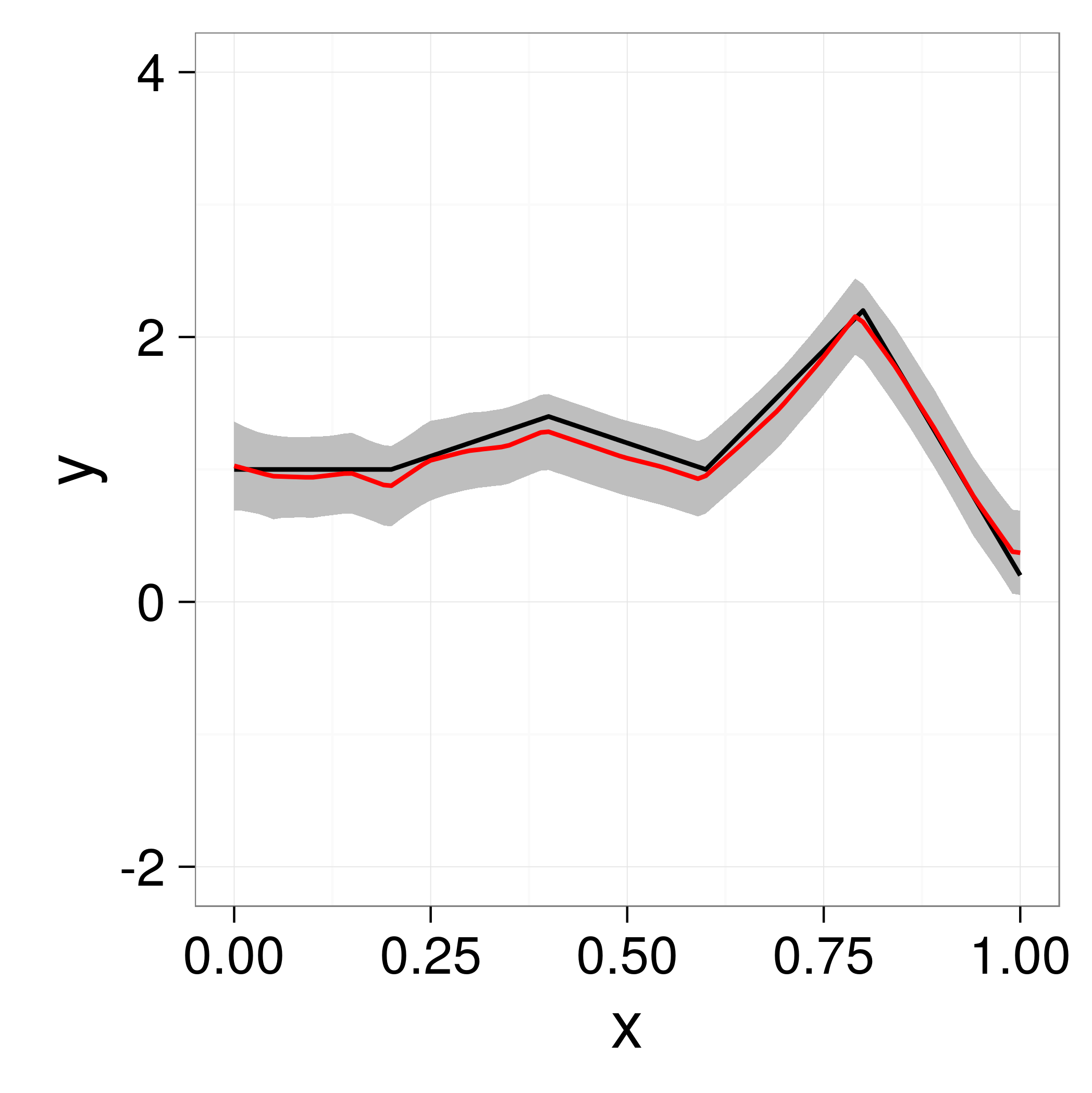}
    \caption{$\ell_2$ fit with mgcv \citep{wood2006generalized}}
  \end{subfigure}
  \caption[Simulation: marginal mean and credible bands from frequentist estimation]{Marginal mean and 95\% credible intervals from frequentist estimation: black is true marginal mean, red is estimated marginal mean}
  \label{sim_point}
\end{figure}

\begin{figure}[H]
\centering
  \begin{subfigure}{0.48\linewidth}
    \includegraphics[width=\textwidth]{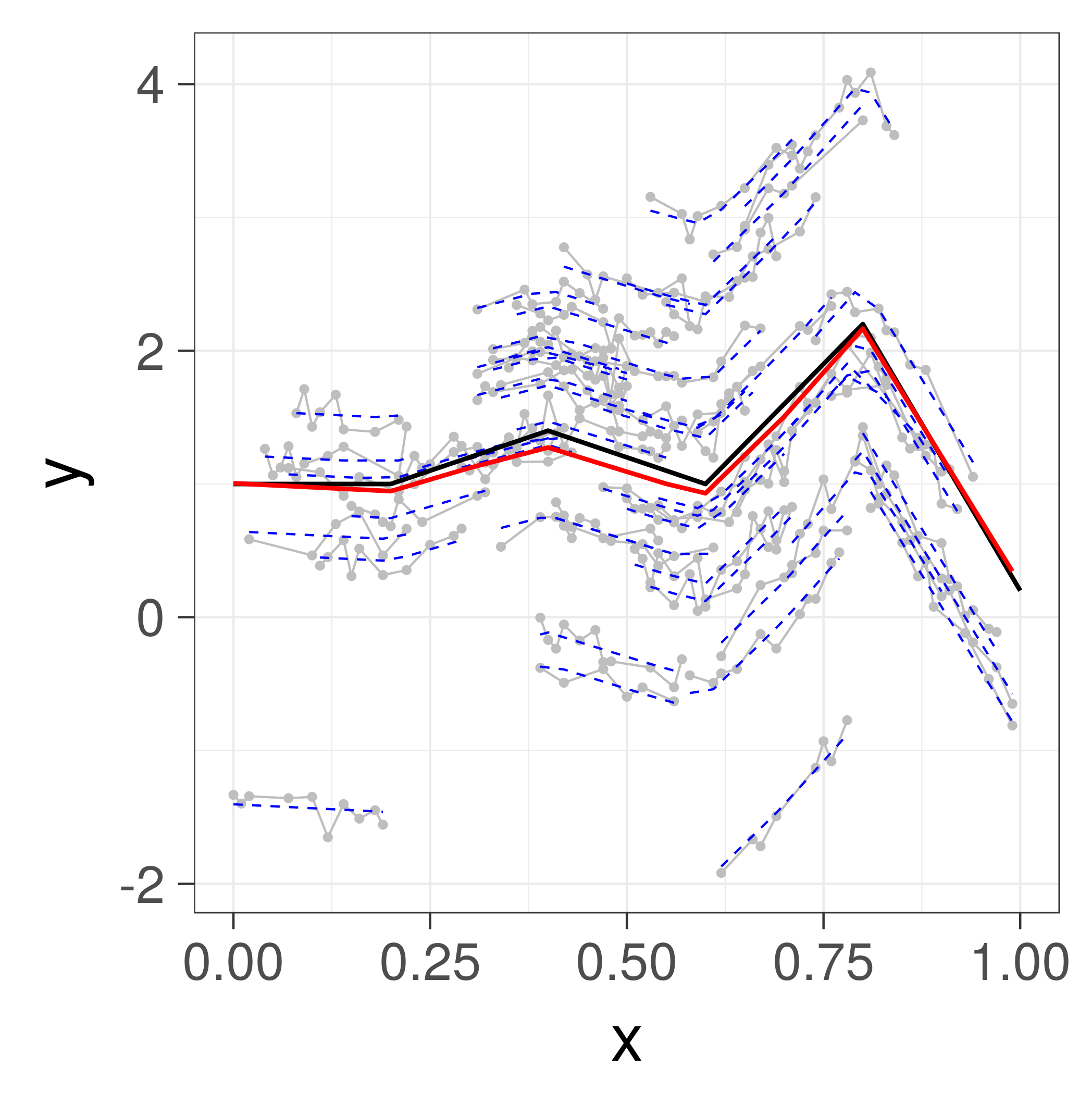}
    \caption{$\ell_1$ fit with ADMM and CV}
  \end{subfigure}
  \begin{subfigure}{0.48\linewidth}
    \includegraphics[width=\textwidth]{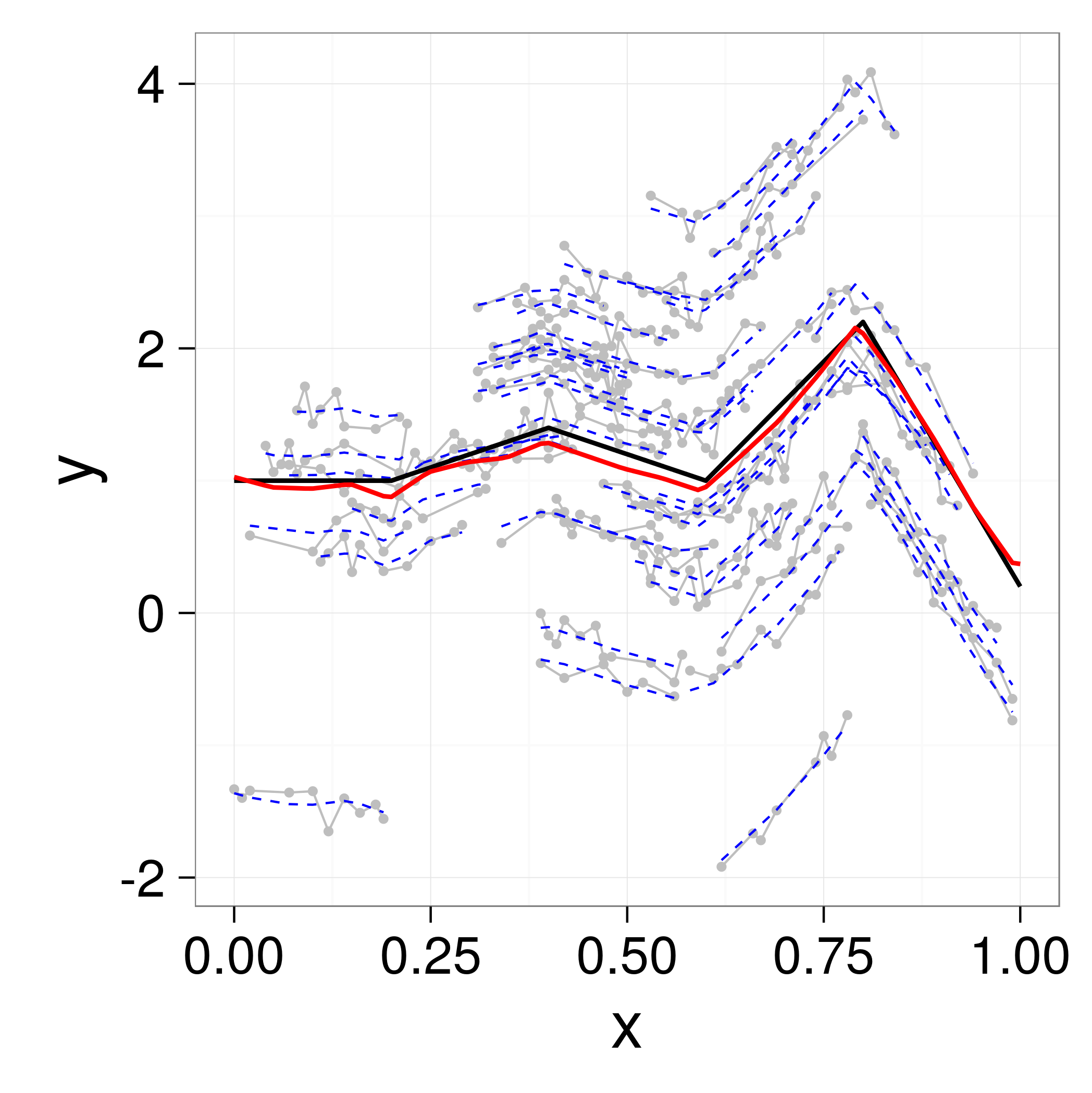}
    \caption{$\ell_2$ fit with mgcv \citep{wood2006generalized}}
  \end{subfigure}
  \caption[Simulation: subject-specific predicted curves from frequentist estimation]{Subject-specific predicted curves from frequentist estimation: black is true marginal mean, red is estimated marginal mean, blue is subject-specific curves}
  \label{sim_subj}
\end{figure}

As seen in Figures \ref{sim_point} and \ref{sim_subj}, the results from the $\ell_1$ and $\ell_2$ penalized models are very similar. However, the $\ell_1$ penalized model does slightly better at identifying the change points and the line segments. We explore this further in Section \ref{sec_change_point}.

Table \ref{tab_l1l2comp} compares the degrees of freedom and variance estimates from the $\ell_1$ penalized fit against those from the $\ell_2$ penalized fit. From Table \ref{tab_l1l2comp}, we see that the ridge degrees of freedom $\hat{\text{df}}^{\text{ridge}}$ appears reasonable, as it is near the estimate for the $\ell_2$ penalized model. The true degrees of freedom $\hat{\text{df}}$ also seems reasonable. Ideally, the degrees of freedom for the $\ell_1$ penalized fit should equal six, as there are four change points and we are using a second order difference penalty (see Section \ref{stab_fast_df}).

\begin{table}[H]
  \centering
  \caption[Simulation: estimated degrees of freedom and variance]{Estimated degrees of freedom for smooth $F$ and variance in $\ell_1$ and $\ell_2$ penalized models}
    \begin{tabular}{lccc}
    \hline \hline
    & \multicolumn{2}{c}{Penalty} \\
    \cline{2-3}
    Estimator & $\ell_1$ & $\ell_2$ & Truth\\
    \hline
    $\hat{\text{df}}^{\text{ridge}}$ & 17.7 & 19.0 & --\\
    $\hat{\text{df}}$ & 10 & -- & --\\
    $\hat{\sigma}^2_\epsilon$ & 0.0093 & 0.0106 & 0.01\\
    $\hat{\sigma}^2_b$ & 1.06 & 1.05 & 1
    \end{tabular}
  \label{tab_l1l2comp}
\end{table}

Table \ref{dfCompSim} compares the different estimates of degrees of freedom. In this simulation, the degrees of freedom based on the ridge approximation is larger than that from Stein's formula, and the approximations based on restricted derivatives are equal or near the estimate with Stein's formula.

\begin{table}[H]
  \centering
  \caption[Simulation: comparison of degrees of freedom estimates]{Comparison of degrees of freedom estimates for the $\ell_1$ penalized model}
  \begin{tabular}{lcccc}
  \hline \hline
  && \multicolumn{3}{c}{Smooth} \\
  \cline{3-5}
  Estimator & Description & Overall & $F$ & $Z$ \\
  \hline
$\hat{\text{df}}$ & Stein (\ref{dfAll}) and (\ref{dfj})  & 14.3 & 10.0 & 3.29 \\
$\tilde{\text{df}}$ & Restricted (\ref{dfResj}) and (\ref{dfRes}) & 14.6 & 10.0 & 3.63 \\
$\tilde{\text{df}}^{\text{ADMM}}$ & ADMM (\ref{dfADMMj}) and (\ref{dfADMM}) & 13.6 & 9.0 & 3.63 \\
$\hat{\text{df}}^{\text{ridge}}$ & Ridge (\ref{dfAllRidge}) and (\ref{dfRidgej}) & 22.1 & 17.7 & 3.31 \\
$\tilde{\text{df}}^{\text{ridge}}$ & Ridge restricted (\ref{dfRidgeResj}) and (\ref{dfRidgeRes}) & 22.4 & 17.8 & 3.63
  \end{tabular}
  \label{dfCompSim}
\end{table}

\subsection{Bayesian estimation}

We modeled the data as $\bm{y} | \bm{b} = \beta_0 \bm{1} + F \bm{\beta} + \bm{b} + \bm{\epsilon}$ where
\begin{align*}
\bm{\epsilon} &\sim N(\bm{0}, \sigma^2_\epsilon I) \\
\bm{b} &\sim N(0, \sigma^2_b I) \\
D^{(2)} \bm{\beta} &\sim \text{Laplace}(\bm{0}, \sigma^2_\lambda I) \\
p(\sigma_\epsilon) &\propto 1 \\
p(\sigma_b) &\propto 1 \\
p(\log(\sigma_\lambda)) &\propto 1.
\end{align*}
We also fit models with normal and diffuse priors for $D^{(2)} \bm{\beta}$.

We fit all models with \verb|rstan| \citep{rstan}, each with four chains of 2,000 iterations with the first 1,000 iterations of each chain used as warmup. The MCMC chains, not shown, appeared to be reasonably well mixing and stationary, and had $\hat{R}$ values under 1.1 \citep[see][]{gelman2014bayesian}.\footnote{As described by \citet[][pp. 284--285]{gelman2014bayesian}, for each scalar parameter, $\hat{R}$ is the square root of the ratio of the marginal posterior variance (a weighed average of between- and within-chain variances) to the mean within-chain variance. As the number of iterations in the MCMC chains goes to infinity, $\hat{R}$ converges to 1 from above. Consequently, $\hat{R}$ can be interpreted as a scale reduction factor, and \citet{gelman2014bayesian} recommend ensuring that $\hat{R}<1$ for all parameters.} Figure \ref{Bayes_CI} shows the marginal mean with 95\% credible intervals, and Figure \ref{Bayes_point} shows point estimates.

\begin{figure}[H]
\centering
\begin{subfigure}{0.32\textwidth}
  \includegraphics[width = \linewidth]{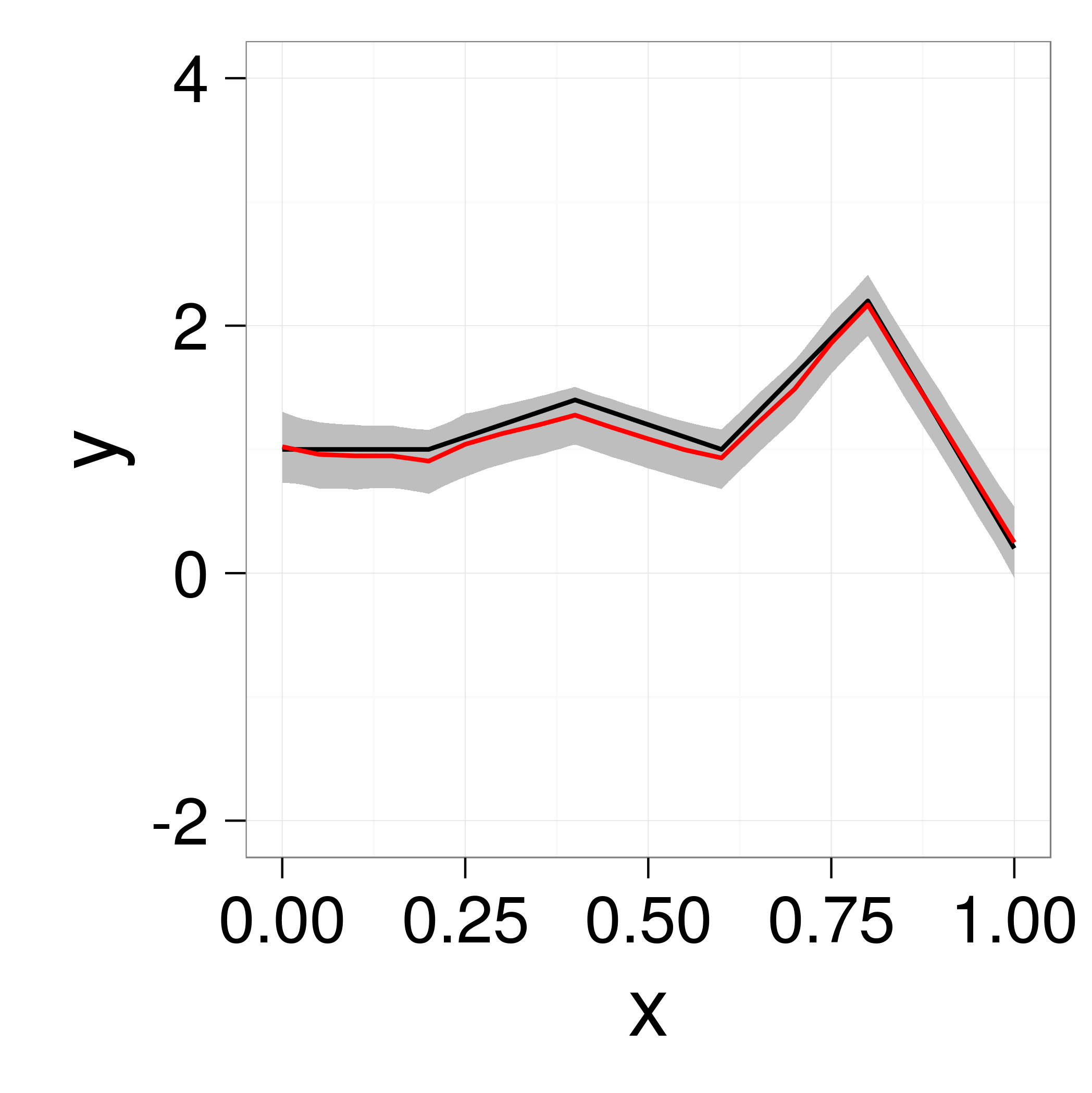}
  \caption{$D \bm{\beta} \sim \text{Laplace}(\bm{0}, \sigma^2_\lambda I)$}
\end{subfigure}
\begin{subfigure}{0.32\textwidth}
  \includegraphics[width = \linewidth]{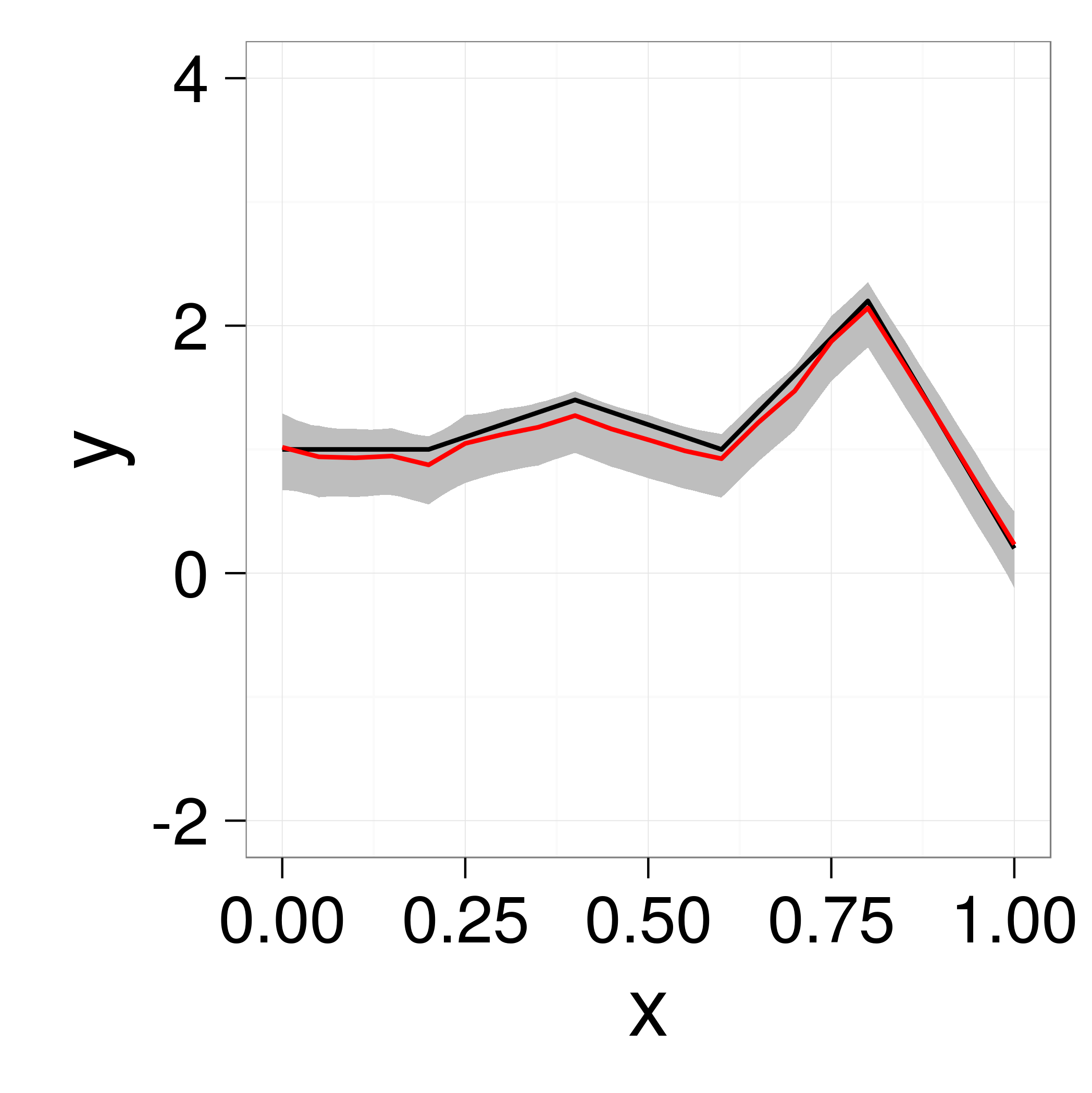}
  \caption{$D \bm{\beta} \sim N(\bm{0}, \sigma^2_\lambda I)$}
\end{subfigure}
\begin{subfigure}{0.32\textwidth}
  \includegraphics[width = \linewidth]{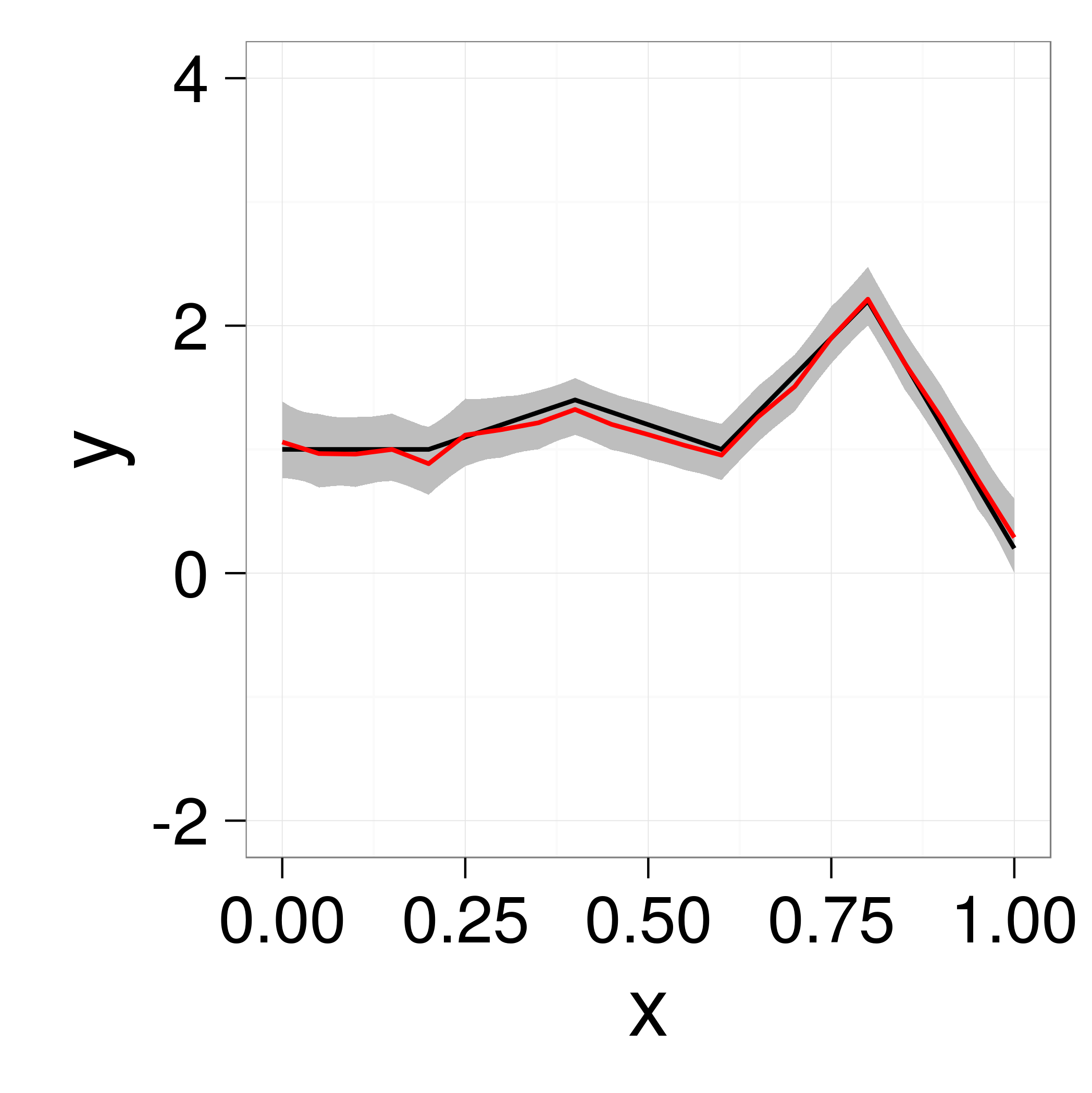}
  \caption{No prior on $D \bm{\beta}$}
\end{subfigure}
\caption[Simulation: marginal mean and credible bands from Bayesian estimation]{Credible bands for Bayesian models with order 2 (degree 1) B-splines. Black is true marginal mean, red dashed is estimated marginal mean, gray area is 95\% credible interval}
\label{Bayes_CI}
\end{figure}

\begin{figure}[H]
\centering
\begin{subfigure}{0.32\textwidth}
  \includegraphics[width = \linewidth]{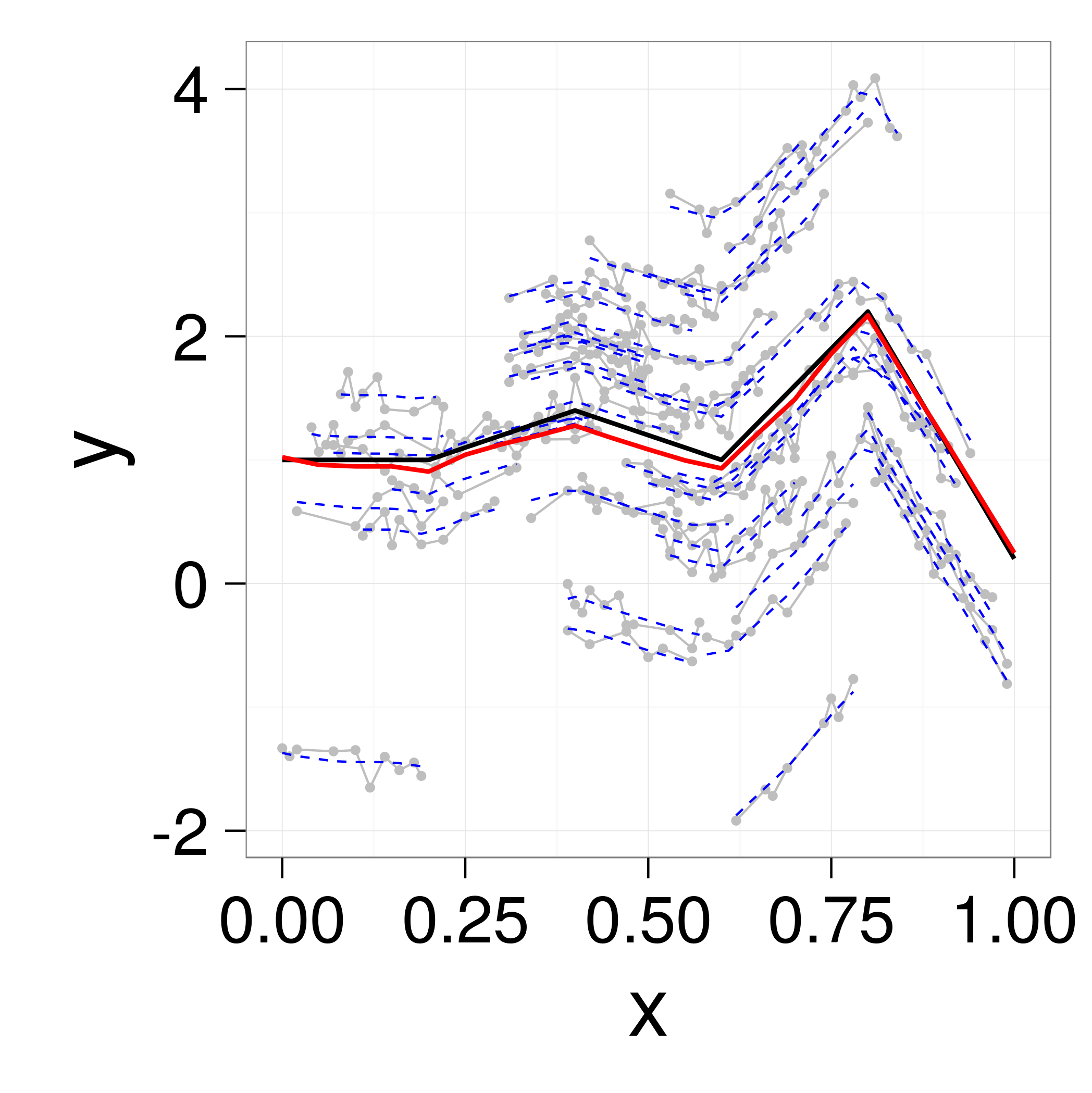}
  \caption{$D \bm{\beta} \sim \text{Laplace}(\bm{0}, \sigma^2_\lambda I)$}
\end{subfigure}
\begin{subfigure}{0.32\textwidth}
  \includegraphics[width = \linewidth]{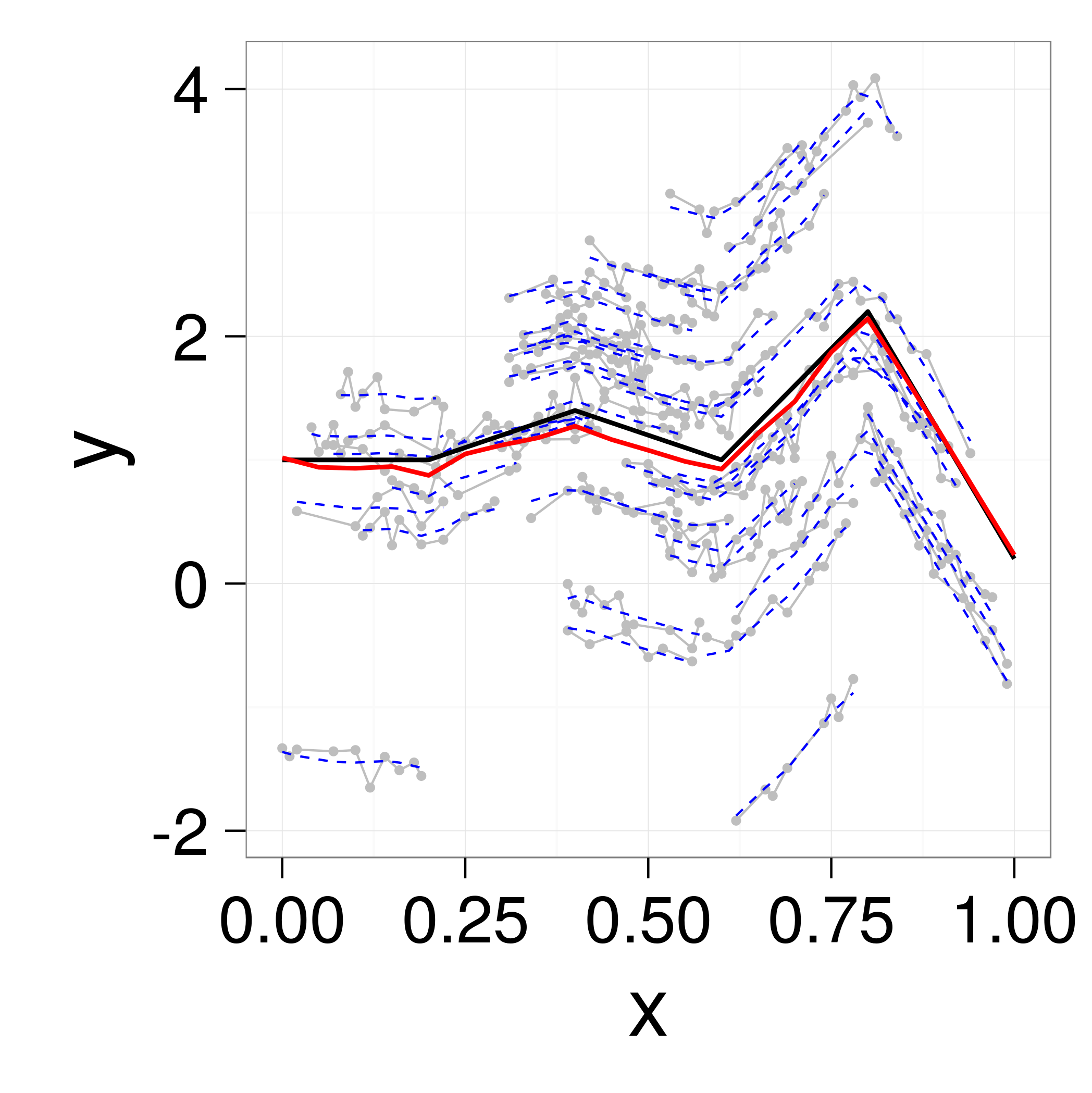}
  \caption{$D \bm{\beta} \sim N(\bm{0}, \sigma^2_\lambda I)$}
\end{subfigure}
\begin{subfigure}{0.32\textwidth}
  \includegraphics[width = \linewidth]{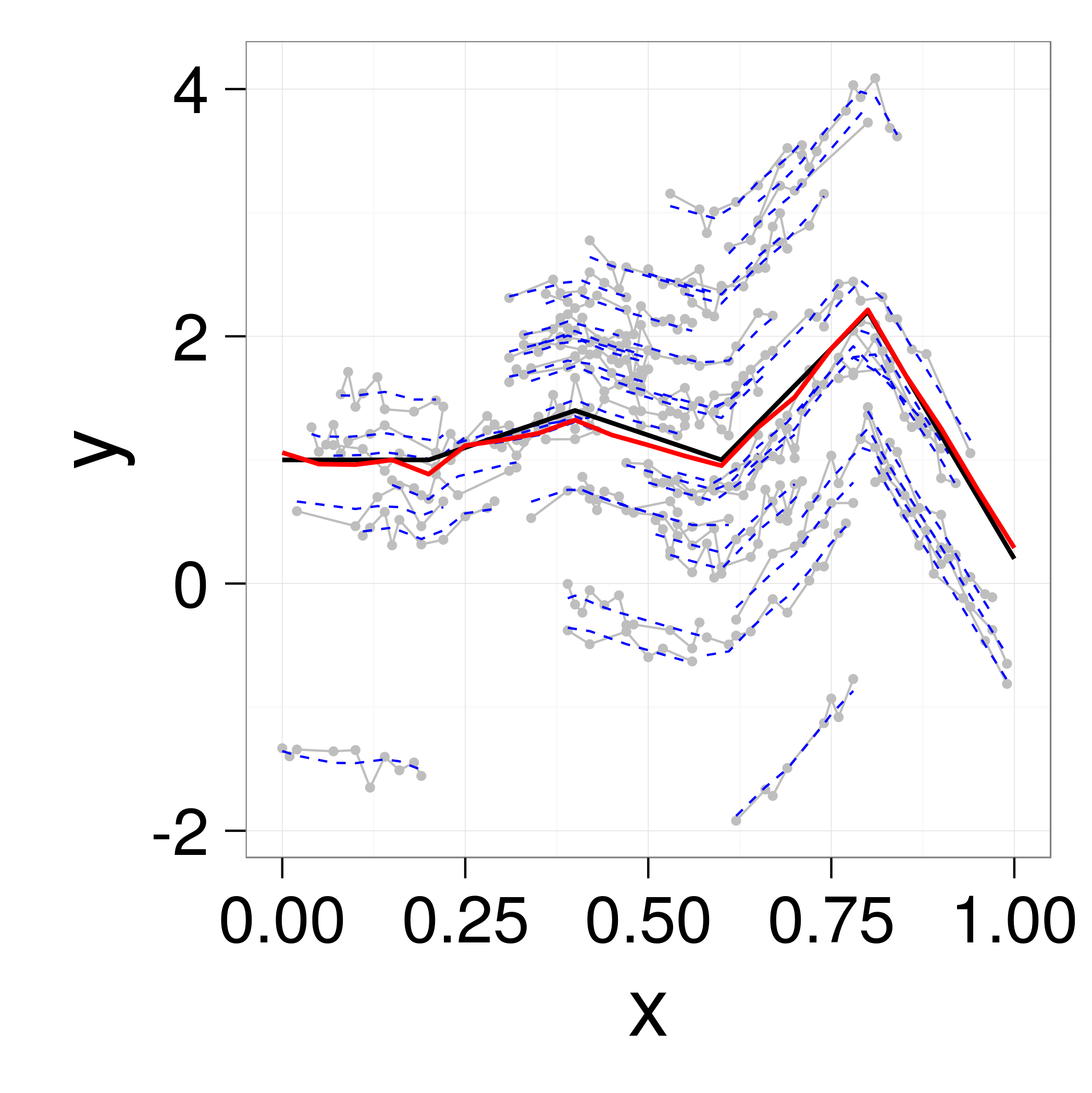}
  \caption{No prior on $D \bm{\beta}$}
\end{subfigure}
\caption[Simulation: subject-specific predicted curves from Bayesian estimation]{Subject-specific predicted curves from Bayesian models fit with order 2 (degree 1) B-splines. Gray is observed data, black is true marginal mean, red dashed is estimated marginal mean, and blue dashed is subject-specific predictions}
\label{Bayes_point}
\end{figure}

As seen in Figures \ref{Bayes_CI} and \ref{Bayes_point}, all models performed well and gave similar fits as above. Similar to before, the Laplace prior appears to better enforce a piece-wise linear fit, particularly around $x=0.2$.

\subsection{Change point detection \label{sec_change_point}}

We simulated 1,000 datasets with the same generating mechanism used to produce the data shown in Figure \ref{nonDiff_smallN} and measured the performance of the $\ell_1$ and $\ell_2$ penalized models on two criteria: 1) the number of inflection points found, and 2) the distance between the estimated inflection points and the closest true inflection point. To that end, let $\mathcal{T}=\{\tau_1,\ldots,\tau_4\}$ be the set of true inflection points, and $M = \max_{x \in \mathcal{X}}|\hat{f}''(x)|$ be the maximum absolute second derivative of the estimated function, where $\mathcal{X}=\{x_1,x_2,\ldots\}$ is the ordered set of unique simulated $x$ values. We approximate $\hat{f}''$ by
$$
\hat{f}''(x_i) \approx \frac{(\hat{f}(x_{i+1}) -\hat{f}(x_{i}))/(x_{i+1} - x_i) - (\hat{f}(x_{i}) -\hat{f}(x_{i-1}))/(x_{i} - x_{i-1})}{x_{i+1} - x_i}.
$$
Then let $\mathcal{I} = \{x \in \mathcal{X} : |\hat{f}''(x)| \ge c M \}$ be the set of estimated inflection points, where $c \in (0,1)$ is a cutoff value defining how large the second derivative must be to be counted as an inflection point. Also, let $n_{\mathcal{I}} = |\mathcal{I}|$ be the number of estimated inflection points, and $\bar{d} = n^{-1}_\mathcal{I} \sum_{x \in \mathcal{I}} \min_{\tau \in \mathcal{T}} |x - \tau|$ be the mean absolute deviance of the estimated inflection points.

Figure \ref{sim_cp} shows the results from 1,000 simulated datasets. The $\ell_1$ penalized model was better able to 1) find the correct number of inflection points, and 2) determine the location of the inflection points.

\begin{figure}[H]
  \begin{subfigure}{0.48\textwidth}
    \includegraphics[width = \linewidth]{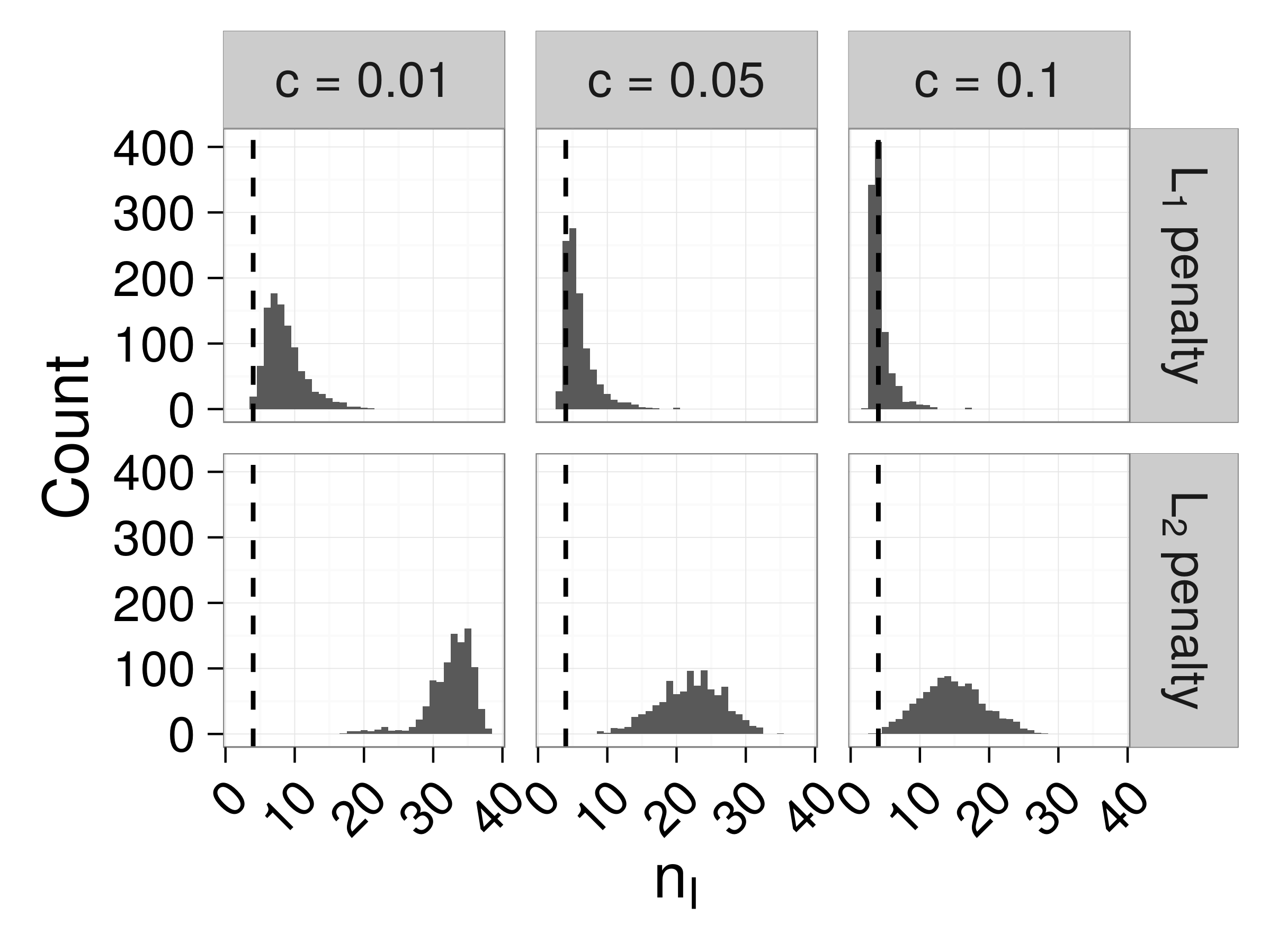}
    \caption{Number of estimated inflection points}
  \end{subfigure}
  \begin{subfigure}{0.48\textwidth}
    \includegraphics[width = \linewidth]{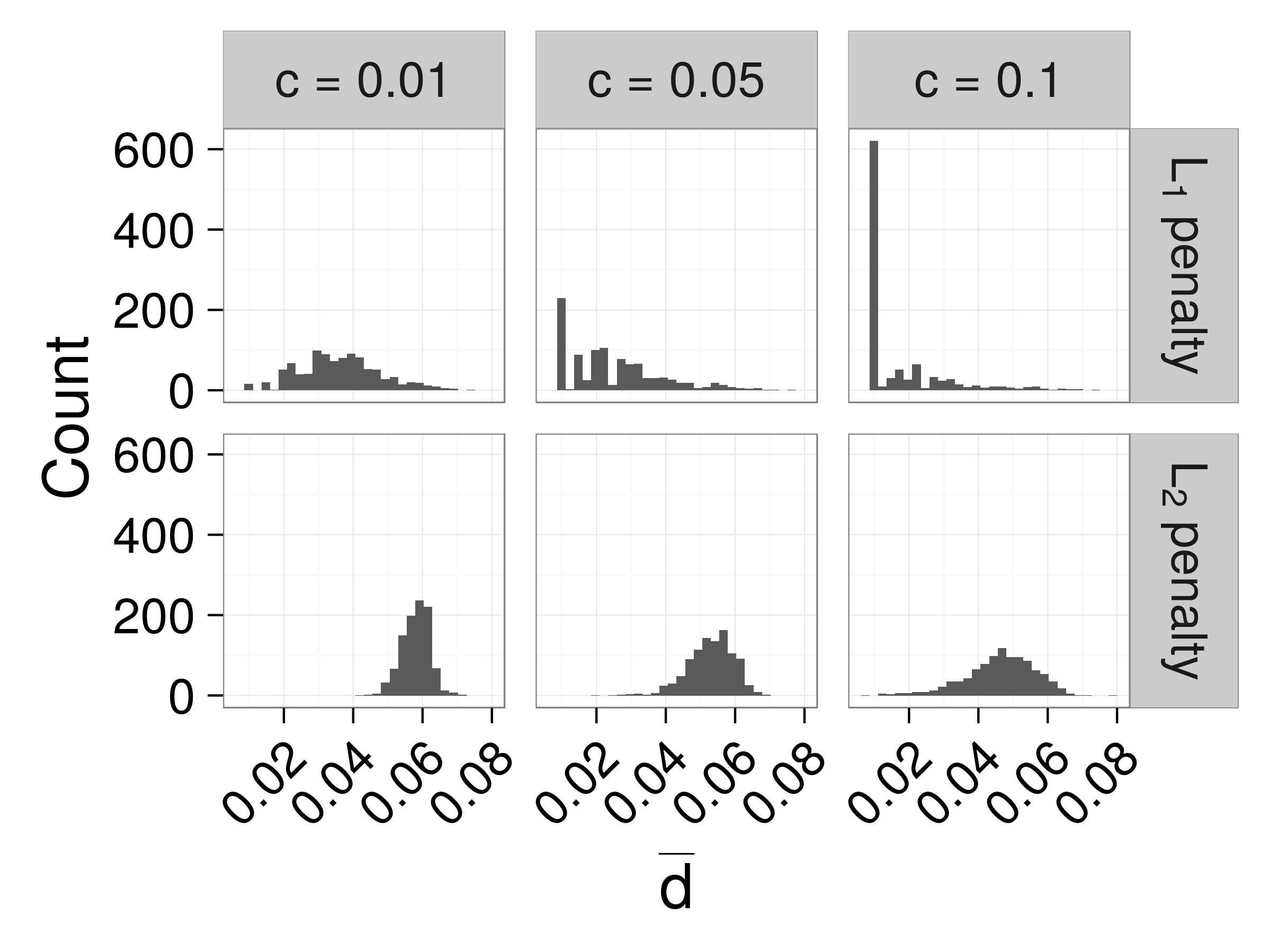}
    \caption{Mean absolute deviance}
  \end{subfigure}
  \caption[Simulation: change point detection results]{Results from 1,000 simulated datasets measuring ability of the models to detect inflection points}
  \label{sim_cp}
\end{figure}

\subsection{Coverage probability}

We simulated 1,000 datasets with the same generating mechanism used to produce the data shown in Figure \ref{nonDiff_smallN} and measured the coverage probability of the approximate Bayesian credible bands described in Section \ref{BayesCB} for the $\ell_1$ penalized model, and simultaneous Bayesian credible bands for the $\ell_2$ penalized model \citep{wood2006generalized}. Figure \ref{cp} shows the coverage probabilities for both approaches. As seen in Figure \ref{cp}, the confidence bands perform similarly and are near the nominal rate over most of the $x$ domain. Both approaches have difficulty maintaining nominal coverage at the edges of the $x$ domain.

\begin{figure}[H]
\centering
\includegraphics[scale = 0.48]{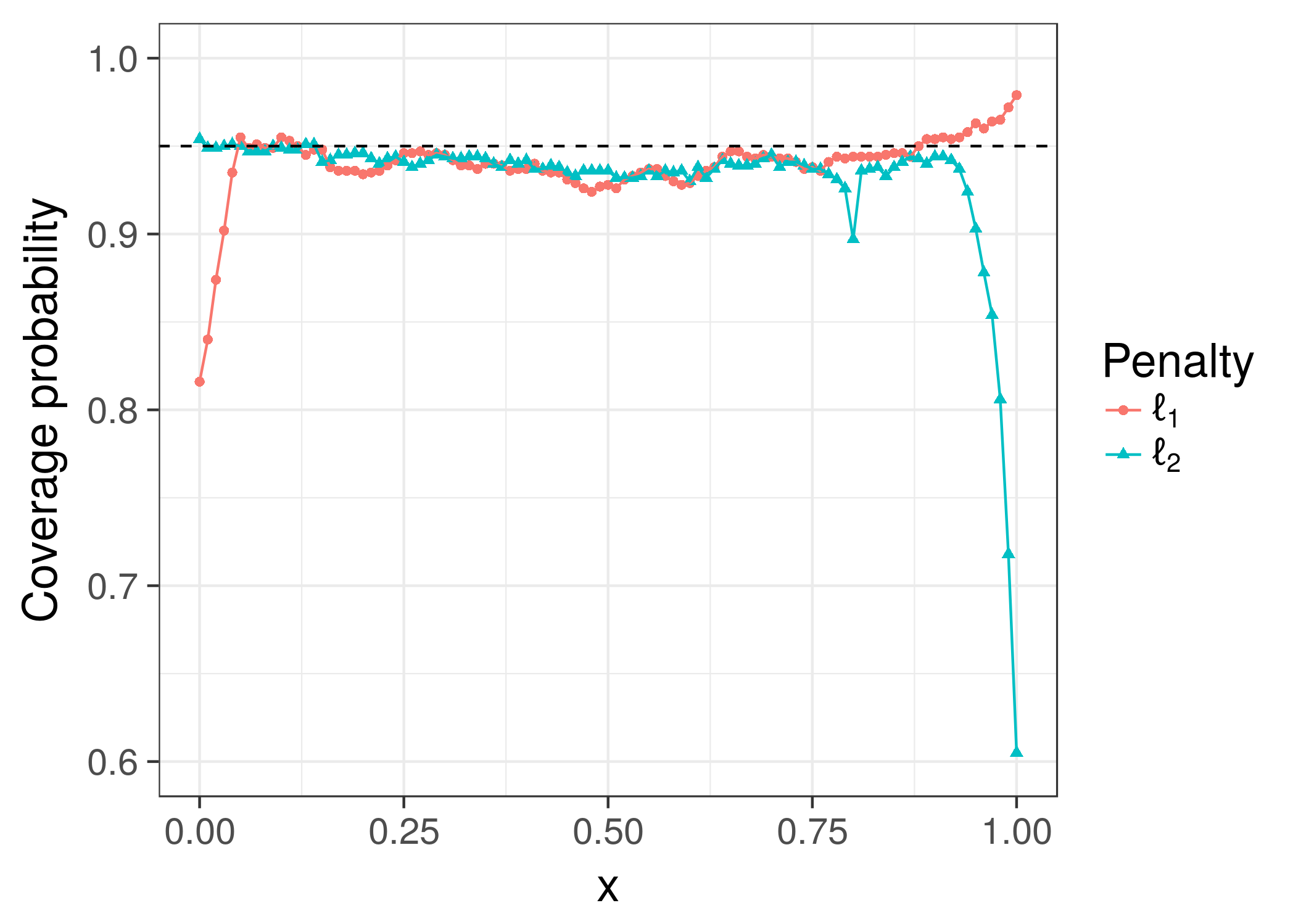}
\caption{Coverage probability from 1,000 simulated datasets using approximate Bayesian credible bands for the $\ell_1$ penalized model and simultaneous Bayesian credible bands for the $\ell_2$ penalized model.}
\label{cp}
\end{figure}

\section{Application \label{sec_app}}

\subsection{Data description and preparation}

In this section, we analyze electrodermal activity (EDA) data collected as part of a stress study. In brief, all subjects completed a written questionnaire prior to the study, which categorized the subjects as having either low vigilance or high vigilance personality types. During the study, all participants wore wristbands that measured EDA while undergoing stress-inducing activities, including giving a public speech and performing mental arithmetic in front of an audience. The scientific questions were: 1) Is EDA higher among high vigilance subjects, and 2) when did trends in stress levels change? In this section, we demonstrate how P-splines with an $\ell_1$ penalty can address both questions.

The raw EDA data are shown in Figure \ref{EDA_raw}. After excluding subjects who had EDA measurements of essentially zero throughout the entire study, we were left with ten high vigilance subjects and seven low vigilance subjects.

\begin{figure}[H]
\centering
  \includegraphics[scale = 0.42]{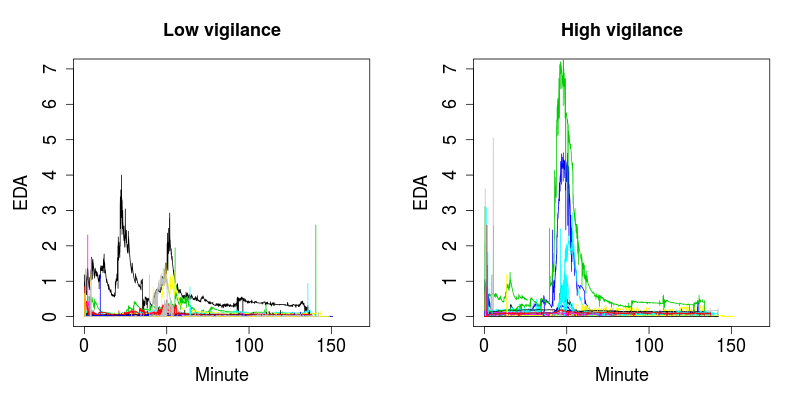}
\caption[Raw electrodermal activity data]{Raw electrodermal activity (EDA) data by experimental group}
\label{EDA_raw}
\end{figure}

To remove the extreme second-by-second fluctuations in EDA, which we believe are artifacts of the measurement device as opposed to real biological signals, we smoothed each curve separately with a Nadaraya--Watson kernel estimator using the \verb|ksmooth| function in \textsf{R}. We then thinned the data to reduce computational burden, taking 100 evenly spaced measurements from each subject. Figure \ref{thinned} shows the results of this process for a single subject, and Figure \ref{eda_data} shows the prepared data for all subjects. Because of the limited number of subjects, as well as issues of misalignment in the time series across individuals, the results presented here should be considered as illustrative rather than of full scientific validity.
\begin{figure}[H]
\centering
\includegraphics[scale = 0.42]{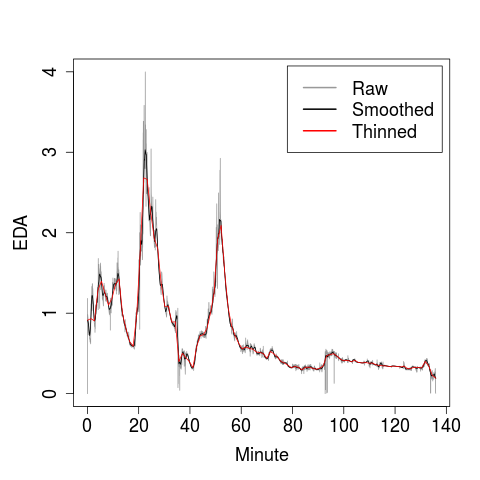}
\caption[Raw, smoothed, and thinned electrodermal activity data]{Raw, smoothed, and thinned electrodermal activity data for a single subject}
\label{thinned}
\end{figure}

\begin{figure}[H]
  \centering
  \includegraphics[scale=0.48]{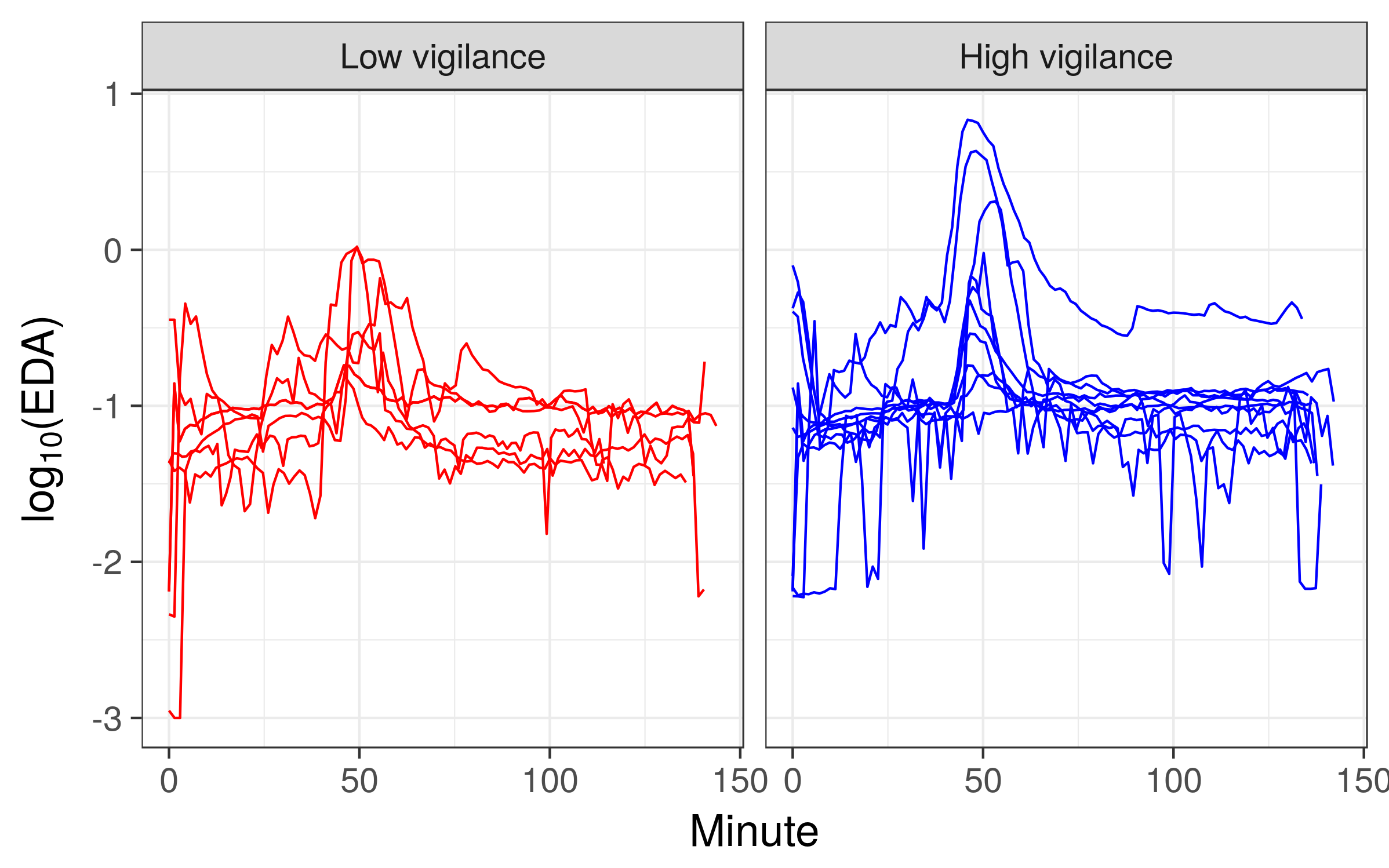}
  \caption[Electrodermal activity data used for analysis]{Electrodermal activity (EDA) data used in the analysis (seven low vigilance and ten high vigilance subjects). Note: subjects not aligned in time (x-axis).}
  \label{eda_data}
\end{figure}

\subsection{Models \label{sec:models}}

In all models, we fit the structure
$$
y_i(x) = \beta_0 + \beta_1(x) + \mathbbm{1}_{\text{high}}[i] \beta_2(x) + b_i(x) + \epsilon_{i}(x)
$$
where $x$ represents time in minutes, $\mathbbm{1}_{\text{high}}[i]=1$ if subject $i$ has high vigilance and $ \mathbbm{1}_{\text{high}}[i]=0$ if subject $i$ has low vigilance, $b_i(x)$ are random subject-specific curves, and $\epsilon_i(x) \sim N(0, \sigma^2_\epsilon)$. For $\beta_1(x)$, $\beta_2(x)$, and $b_i(x)$, we used a fourth order B-spline basis with 31 basis functions each and a second order difference penalty ($k=1$).

Written in matrix notation, the $\ell_1$ penalized model is
\begin{equation}
\min  \frac{1}{2} \| \bm{y} - \beta_0 \bm{1} - \sum_{j=1}^2 F_j \bm{\beta}_j - Z \bm{b} \|_2^2 + \sum_{j=1}^2 \lambda_j \|D^{(2)} \bm{\beta}_j \|_1 + \bm{b}^T S \bm{b}
\label{l1mod}
\end{equation}
where $\bm{y}$ is a stacked vector for subjects $i=1,\ldots,17$, $F_1$ is an $n \times p$ design matrix where $n=1,700$ and $p=31$, and $F_2 = \text{diag}( \mathbbm{1}_{\text{high}}[\bm{i}]) F_1$ where $\bm{i}$ is an $n \times 1$ vector of subject IDs. In other words, $F_2$ is equal to $F_1$, but with rows corresponding to low vigilance subjects zeroed out. We set
$$
Z = \begin{bmatrix}
Z_1 \\
& \ddots \\
&& Z_{17}
\end{bmatrix}
$$
where each $Z_i$ is an $n_i \times 31$ random effects design matrix of order 4 B-splines evaluated at the input points for subject $i$, and 
$$
S = \begin{bmatrix}
S_1 \\
& \ddots \\
&& S_{17}
\end{bmatrix}
$$
where $S_{i,jl} = \int \phi''_{ij}(t) \phi''_{il}(t) dt$ are smoothing spline penalty matrices. We also mean-centered $F_1$ as described in Section \ref{sec_model}, with the corresponding changes in dimensions.

To fit a comparable $\ell_2$ penalized model, in which $\lambda_j \|D^{(2)} \bm{\beta}_j \|_1$ in (\ref{l1mod}) is replaced with $(\lambda_j/2) \|D^{(2)} \bm{\beta}_j \|_2^2$, we rotated the random effect design and penalty matrices $Z$ and $S$ as described in Section \ref{sec_model}. To facilitate the use of existing software, we used a normal prior for the ``unpenalized" random effect coefficients, i.e. $\bm{\breve{b}}_f \sim N(\bm{0}, \sigma^2_f I)$.

We also fit a Bayesian model using the same rotations and equivalent penalties as above. In particular, we modeled the data as $\bm{y}|\bm{b} = \beta_0 \bm{1} + \sum_{j=1}^J F_j \bm{\beta}_j + \check{Z}_r \bm{\check{b}}_r + \breve{Z}_f \bm{\breve{b}}_f + \bm{\epsilon}$ where
\begin{align}
\bm{\check{b}}_r &\sim N(\bm{0}, \sigma^2_r I) \nonumber \\
\bm{\breve{b}}_f &\sim N(\bm{0}, \sigma^2_f I) \nonumber \\
\left( D_j \bm{\beta}_j \right)_l &\sim \text{Laplace}(0, a_j) \text{ for } a_j = \sigma^2_\epsilon / (2 \lambda_j), l=1,\ldots, p_j - k_j - 1, j=1,\ldots,J \label{lapPrior} \\
\bm{\epsilon} &\sim N \left( \bm{0}, \sigma^2_\epsilon I_n \right). \nonumber
\end{align}

\subsection{Results \label{sec_res}}

\subsubsection{Frequentist estimation}

We tried to use CV to estimate the smoothing parameters for the $\ell_1$ penalized model. However, with only 17 subjects split between two groups, we only did 3-fold CV. CV did not find a visually reasonable fit so we set the tuning parameters by hand.

Figure \ref{EDA_L1_smooth} shows the estimated marginal mean and 95\% credible bands for the $\ell_1$ penalized model, and Figure \ref{EDA_L1_subj} shows the subject-specific predicted curves for the $\ell_1$ penalized model. As seen in Figure \ref{EDA_L1_smooth1}, our model identified a few inflection points, particularly near minutes 40, 50, and 60. From Figure \ref{EDA_L1_smooth2} it appears that the difference in EDA between the low and high vigilance subjects was not statistically significant. Also, as seen in Figure \ref{EDA_L1_subj}, the subject-specific predicted curves are shrunk towards the mean, which is expected, because the predicted curves are analogous to BLUPs, although they are not linear smoothers.

\begin{figure}[H]
\centering
\begin{subfigure}{0.48\textwidth}
  \includegraphics[width = \linewidth]{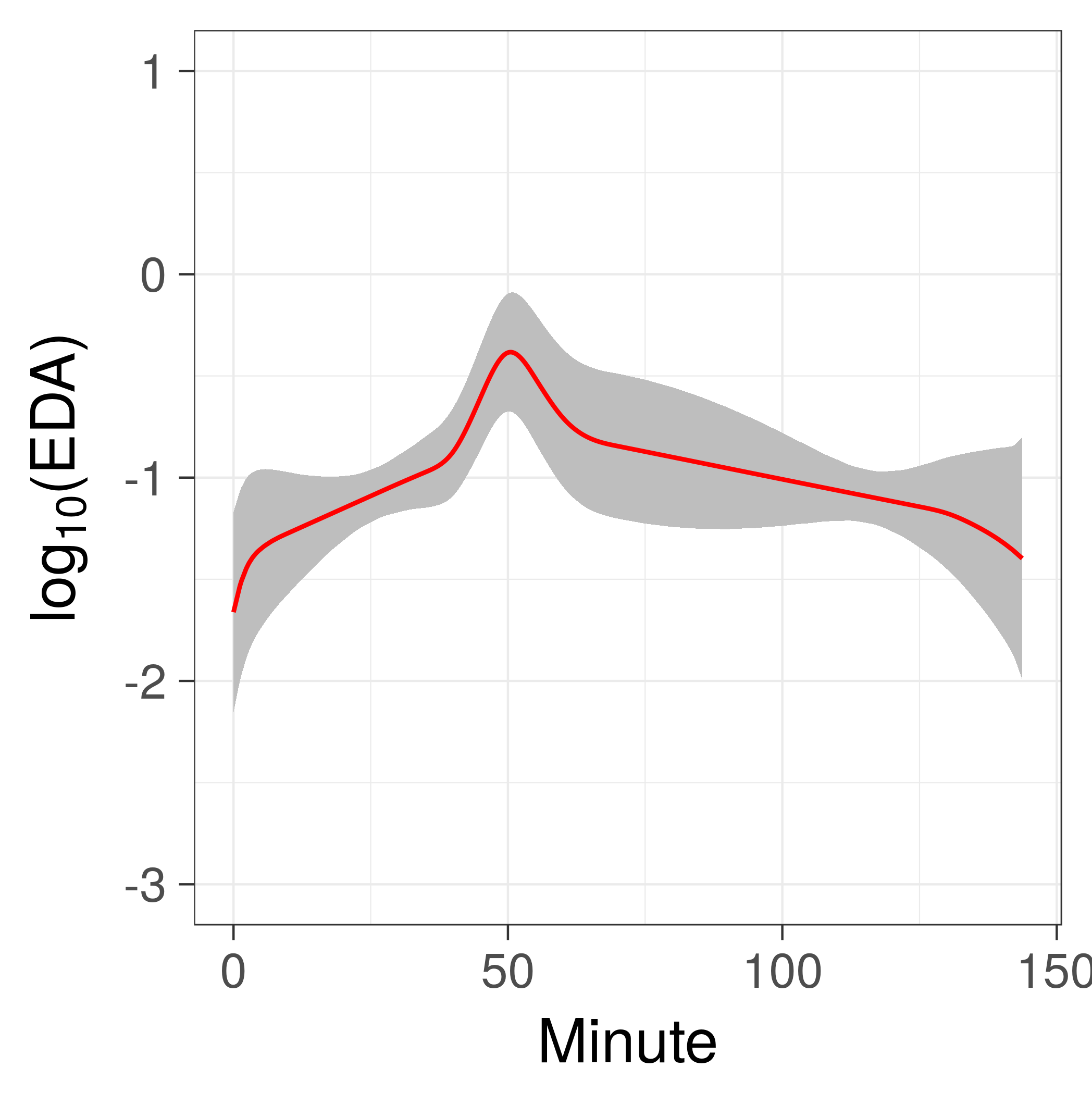}
  \caption{$\hat{\beta}_1(x)$ (low vigilance)}
  \label{EDA_L1_smooth1}
\end{subfigure}
\begin{subfigure}{0.48\textwidth}
  \includegraphics[width = \linewidth]{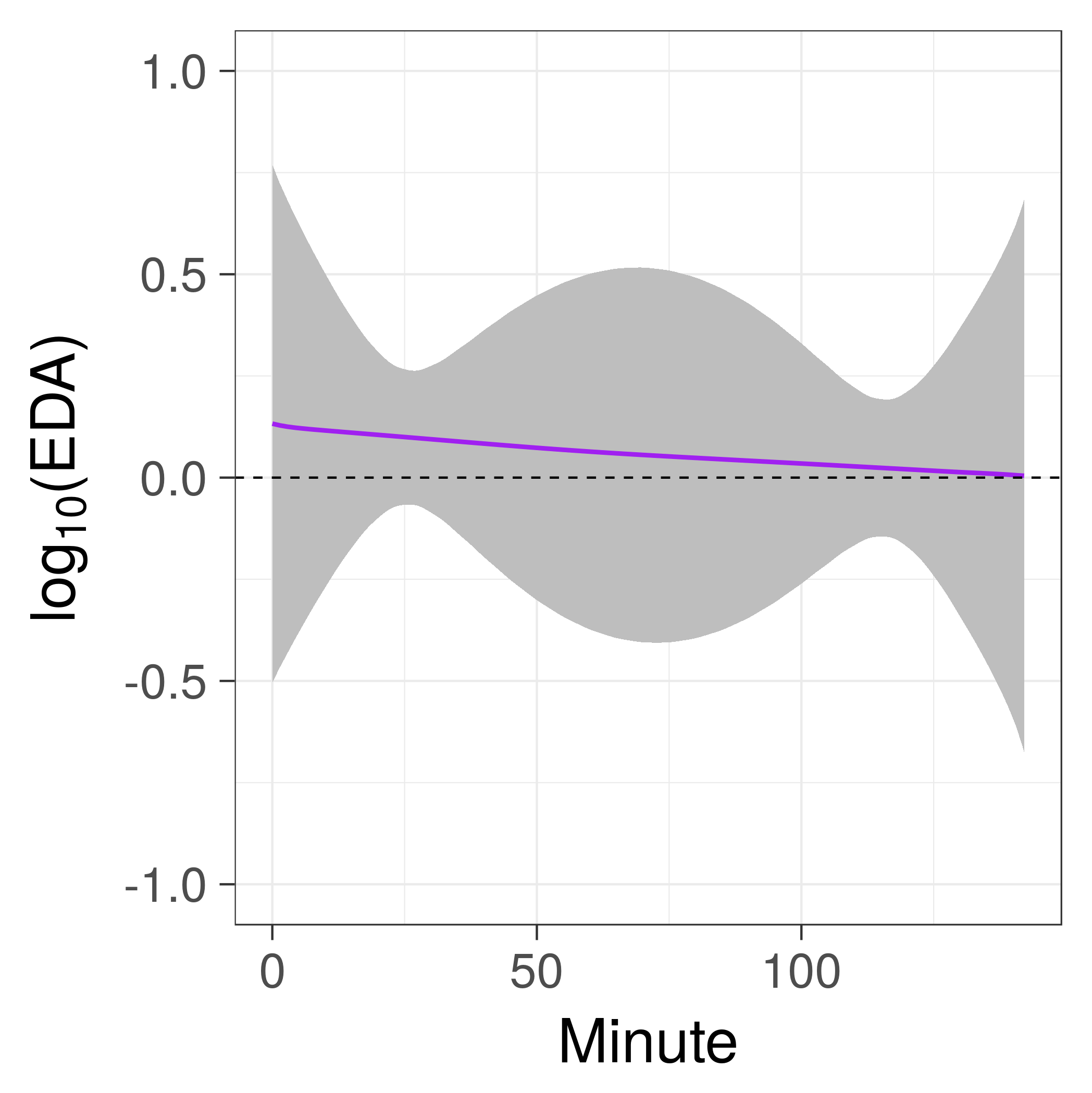}
  \caption{$\hat{\beta}_2(x)$ ($\text{high} - \text{low}$ vigilance)}
  \label{EDA_L1_smooth2}
\end{subfigure}
  \caption{$\ell_1$ penalized model: parameter estimates with 95\% confidence bands}
  \label{EDA_L1_smooth}
\end{figure}

\begin{figure}[H]
  \centering
  \includegraphics[scale = 0.48]{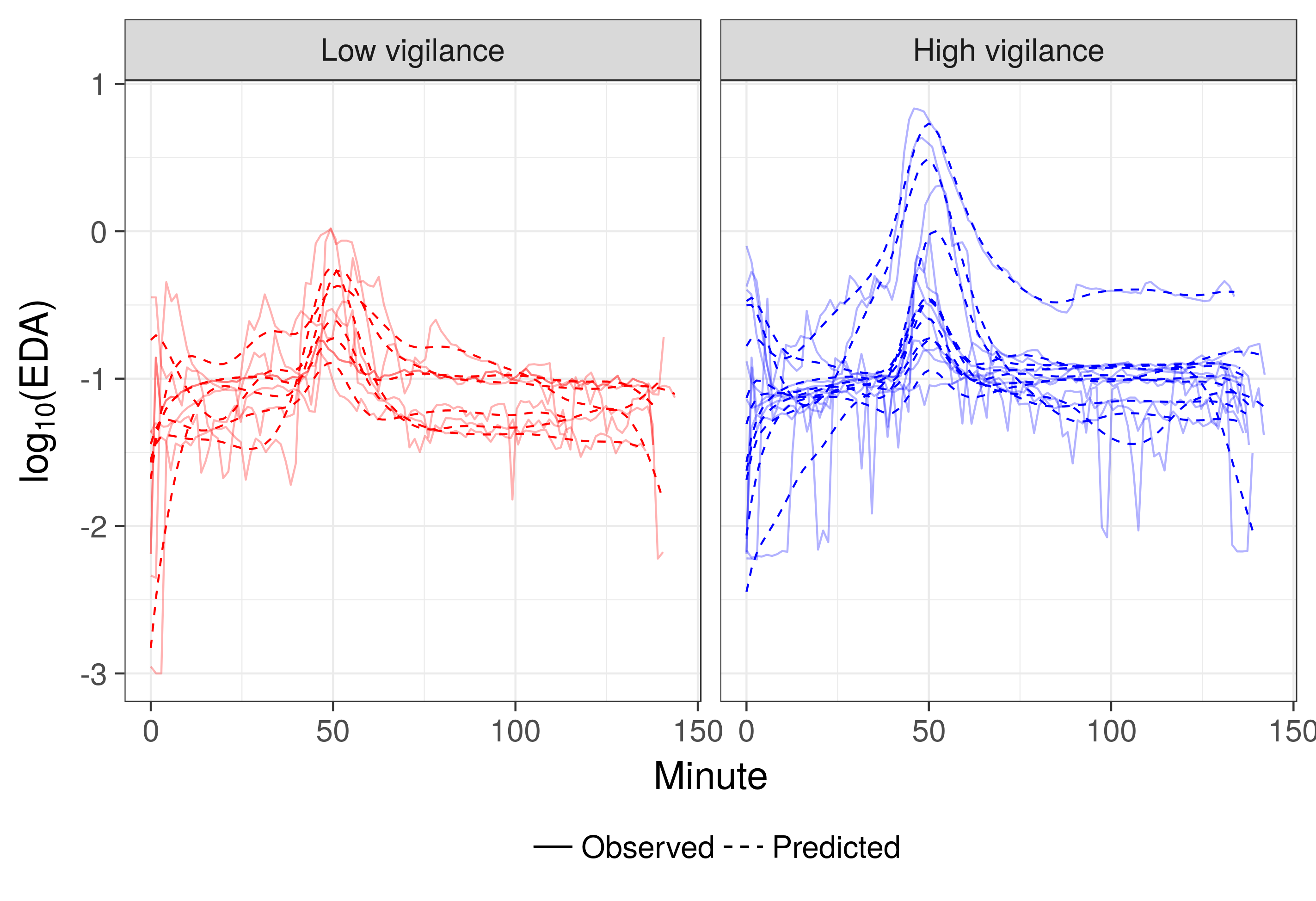}
  \caption{$\ell_1$ penalized model: subject-specific predicted curves}
  \label{EDA_L1_subj}
\end{figure}

Figure \ref{EDA_GAMM_smooths} shows the estimated marginal mean and 95\% credible bands for the $\ell_2$ penalized model, and Figure \ref{l2Pred} shows the subject-specific predicted curves for the $\ell_2$ penalized model. The estimate shown in Figure \ref{eda_l2_low_fit} is similar to that shown in Figure \ref{EDA_L1_smooth1}, though the inflection points are slightly less pronounced in Figure \ref{eda_l2_low_fit}. The results in Figure \ref{eda_l2_diff_fit} are for the most part substantively the same as those in Figure \ref{EDA_L1_smooth2}; the $\ell_2$ penalized model does not show a statistically significant difference between the low and high vigilance subjects, with the possible exception of minutes 45 to 66. As seen in Figure \ref{l2Pred}, the predicted subject-specific curves from the $\ell_2$ penalized model are also shrunk towards the mean.

\begin{figure}[H]
\centering
\begin{subfigure}{0.48\textwidth}
  \includegraphics[width = \linewidth]{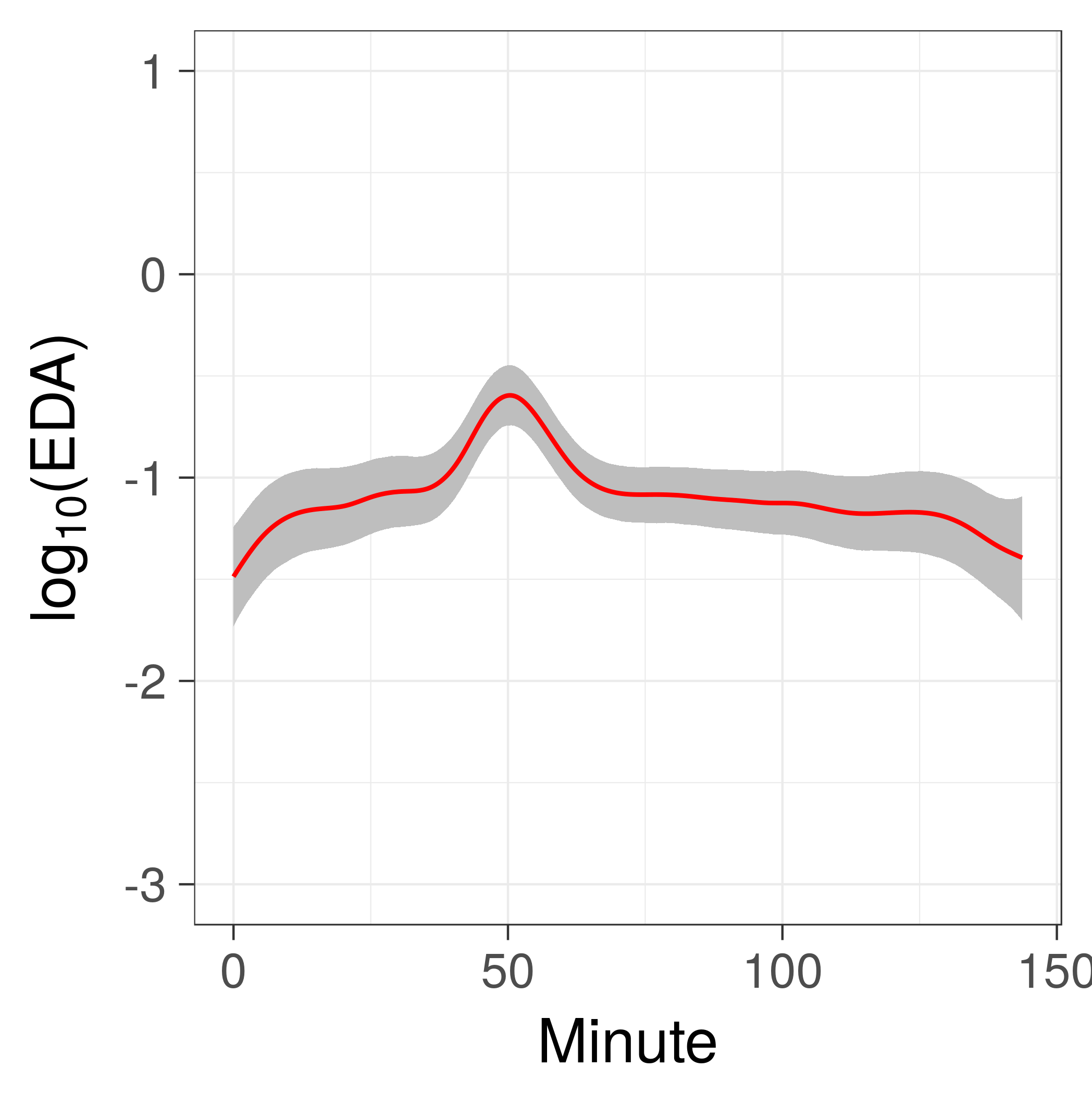}
  \caption{$\hat{\beta}_1(x)$ (low vigilance)}
  \label{eda_l2_low_fit}
\end{subfigure}
\begin{subfigure}{0.48\textwidth}
  \includegraphics[width = \linewidth]{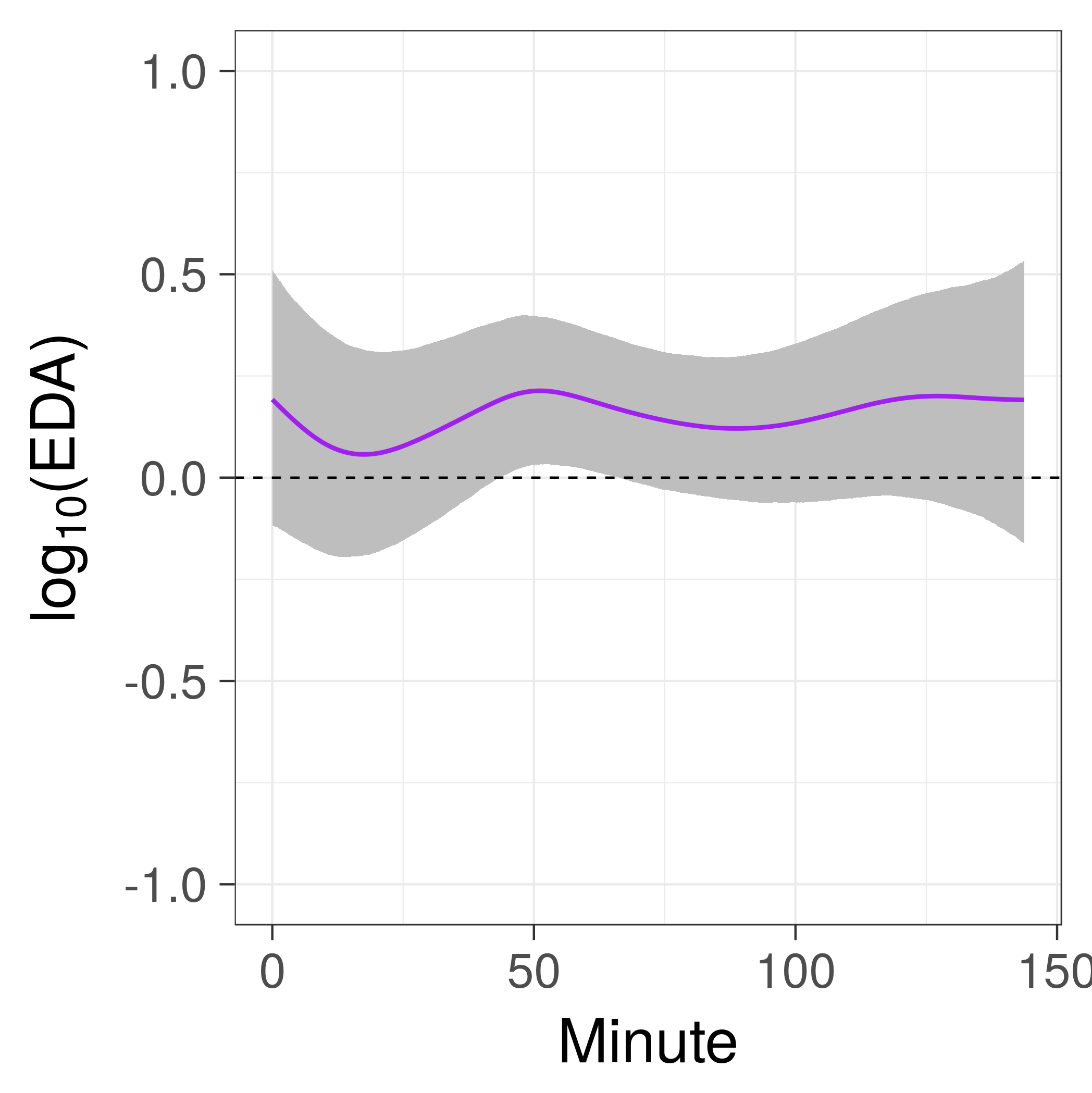}
  \caption{$\hat{\beta}_2(x)$ ($\text{high} - \text{low}$ vigilance)}
  \label{eda_l2_diff_fit}
\end{subfigure}
\caption{$\ell_2$ penalized model: parameter estimates with 95\% confidence bands}
\label{EDA_GAMM_smooths}
\end{figure}

\begin{figure}[H]
  \centering
  \includegraphics[scale = 0.48]{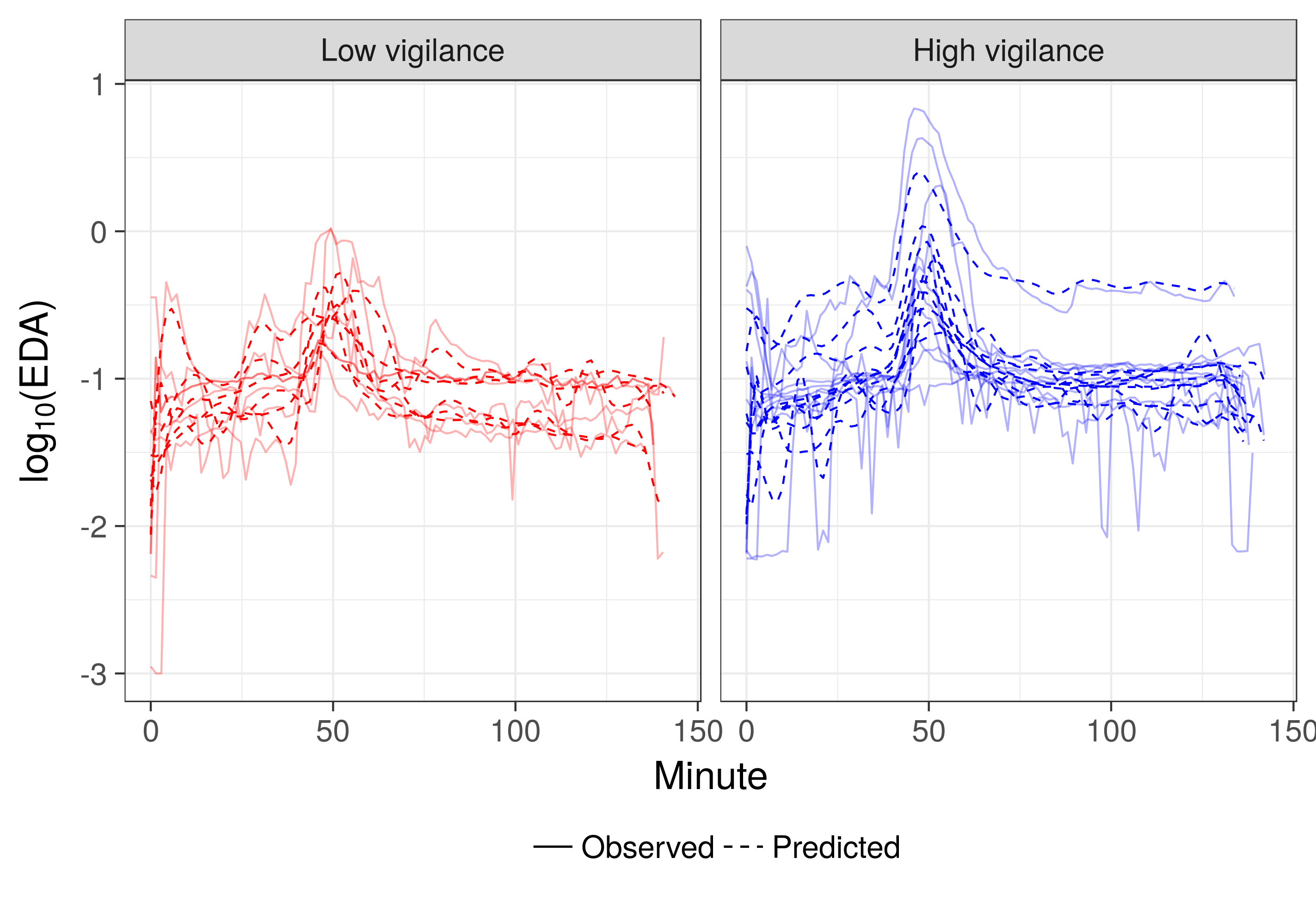}
  \caption{$\ell_2$ penalized model: subject-specific predicted curves}
  \label{l2Pred}
\end{figure}

Table \ref{dfCompApp} shows the estimated degrees of freedom for the $\ell_1$ penalized model. Stein's method $\hat{\text{df}}$ ((\ref{dfAll}) and (\ref{dfj})) and the ridge approximation $\hat{\text{df}}^{\text{ridge}}$ ((\ref{dfAllRidge}) and (\ref{dfRidgej})) were numerically instable ($A^T A + \Omega$ and $U^T U + \Omega^{\text{ridge}}$ were computationally singular). Therefore we used the restricted derivative approximation $\tilde{\text{df}}$ to estimate the variance, as described in Section \ref{sec_var}. In the $\ell_2$ penalized model, smooth $F_1$ had 14.2 degrees of freedom, and smooth $F_2$ had 6.96 degrees of freedom.

\begin{table}[H]
  \centering
  \caption{Comparison of degrees of freedom estimates for the $\ell_1$ penalized model}
  \begin{tabular}{cccccc}
  \hline \hline
  && \multicolumn{4}{c}{Smooth} \\
  \cline{3-6}
  Estimator & Description & Overall & $F_1$ & $F_2$ & $Z$ \\
  \hline
$\hat{\text{df}}$ & Stein (\ref{dfAll}) and (\ref{dfj})  & -- & -- & -- & -- \\
$\tilde{\text{df}}$ & Restricted (\ref{dfResj}) and (\ref{dfRes}) & 194 & 10.0 & 2.00 & 181 \\
$\tilde{\text{df}}^{\text{ADMM}}$ & ADMM (\ref{dfADMMj}) and (\ref{dfADMM}) & 193 & 9.00 & 2.00 & 181 \\
$\hat{\text{df}}^{\text{ridge}}$ & Ridge (\ref{dfAllRidge}) and (\ref{dfRidgej}) & -- & -- & -- & -- \\
$\tilde{\text{df}}^{\text{ridge}}$ & Ridge restricted (\ref{dfRidgeResj}) and (\ref{dfRidgeRes}) & 216 & 21.1 & 13.50 & 181
  \end{tabular}
  \label{dfCompApp}
\end{table}

\subsubsection{Bayesian estimation}

We fit the model described in Section \ref{sec:models} with an element-wise Laplace prior on $D \bm{\beta}$ given by (\ref{lapPrior}). To fit the model, we used \verb|rstan| \citep{rstan} with four chains of 5,000 iterations each, with the first 2,500 iterations of each chain used as warmup. The MCMC chains, not shown, appeared to be reasonably well mixing and stationary with $\hat{R}$ values under 1.1 \citep[see][]{gelman2014bayesian}. Figure \ref{EDA_Bayes_CI} shows the marginal means with 95\% credible bands, and Figure \ref{EDA_Bayes_subj_curves} shows the subject-specific curves. Similar to the $\ell_2$ penalized model, the Bayesian model found a slightly statistically significant difference between low and high vigilance between minutes 42 and 65.

\begin{figure}[H]
\centering
\begin{subfigure}{0.48\textwidth}
  \includegraphics[width = \linewidth]{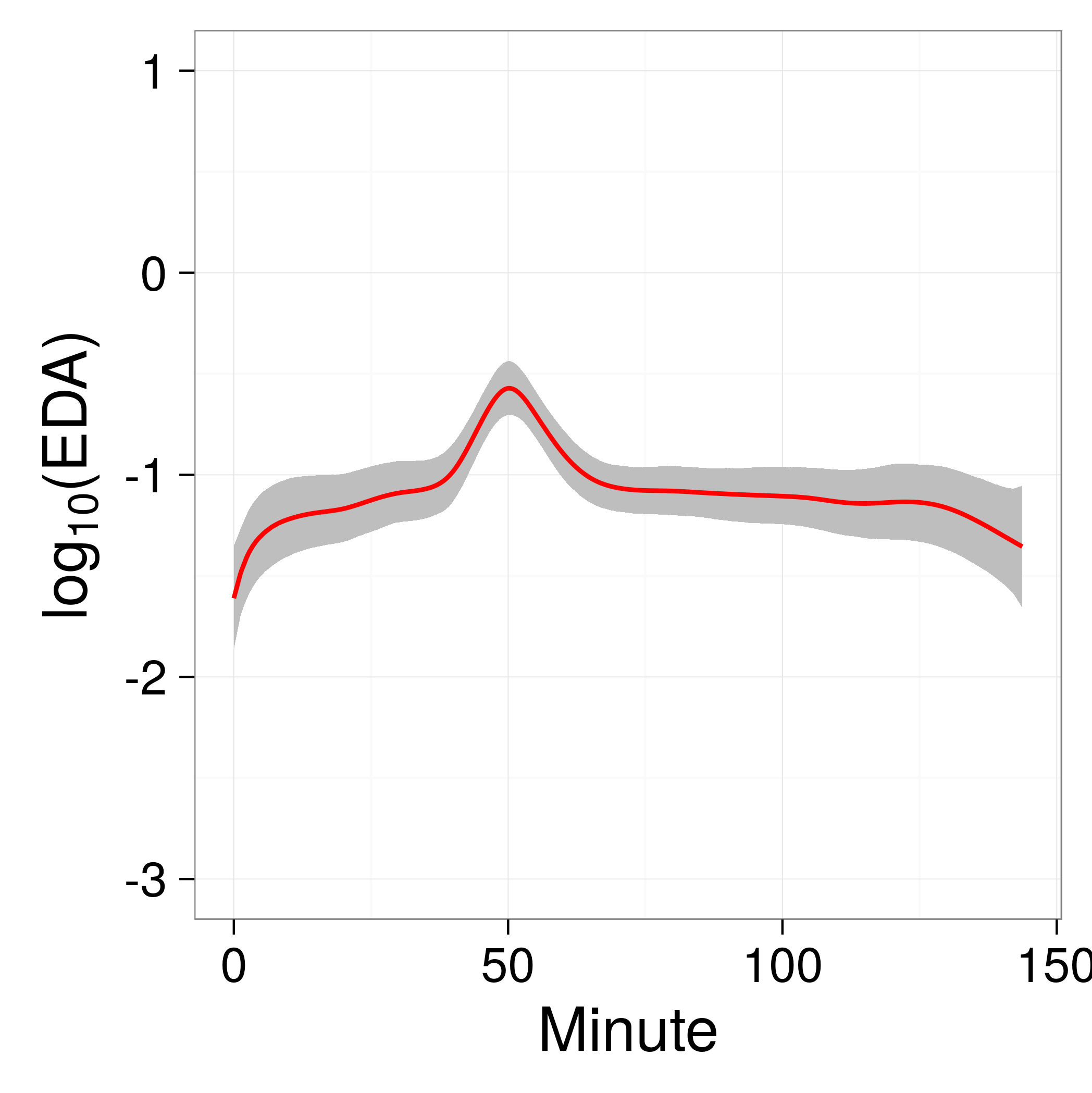}
  \caption{$\hat{\beta}_1(x)$ (low vigilance subjects)}
  \label{EDA_Bayes_CI_low}
\end{subfigure}
\begin{subfigure}{0.48\textwidth}
  \includegraphics[width = \linewidth]{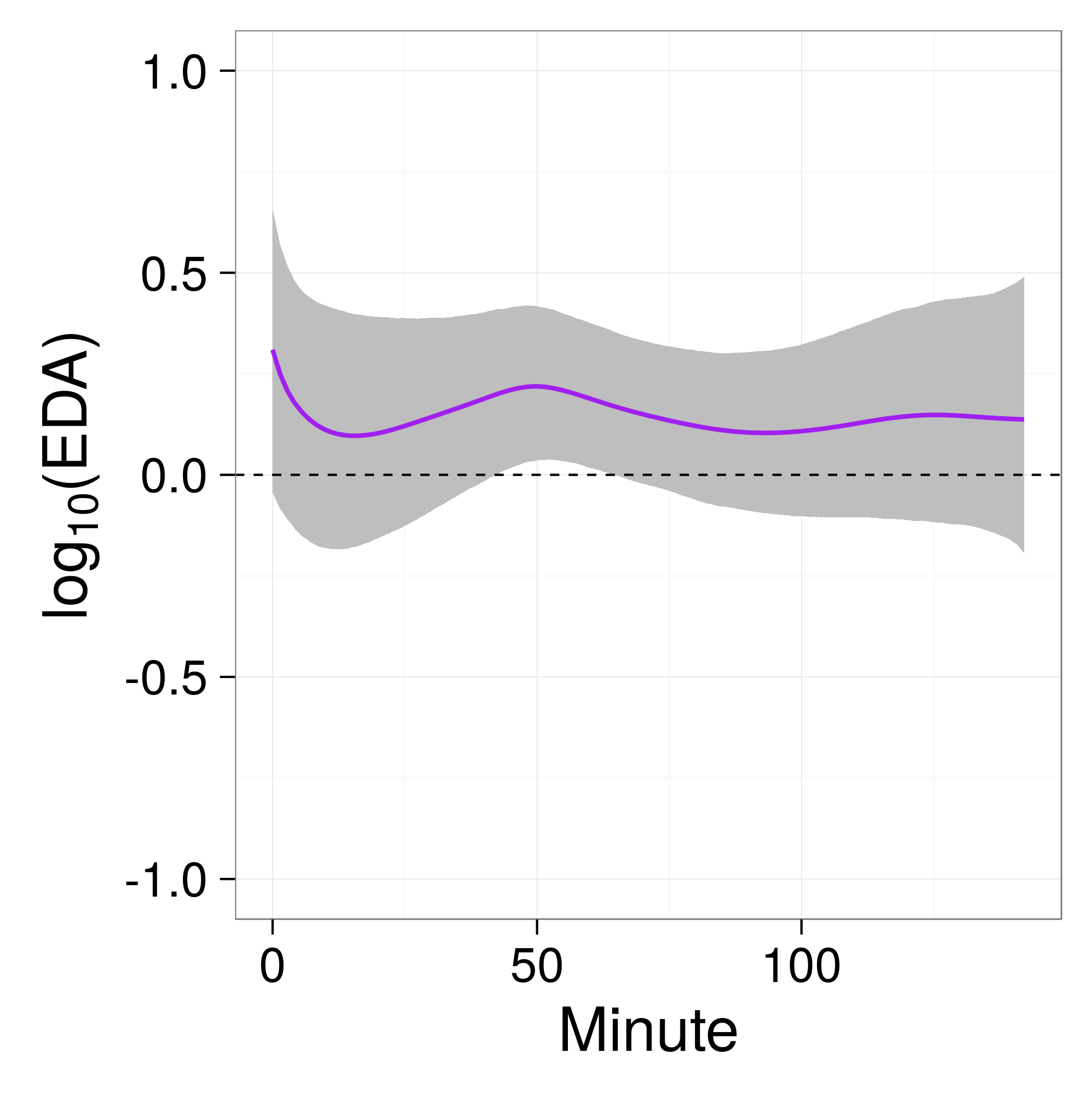}
  \caption{$\hat{\beta}_2(x)$ ($\text{high} - \text{low}$ vigilance subjects)}
  \label{EDA_Bayes_CI_diff}
\end{subfigure}
\caption{Bayesian model: parameter estimates with 95\% confidence bands}
\label{EDA_Bayes_CI}
\end{figure}

\begin{figure}[H]
  \centering
  \includegraphics[scale = 0.48]{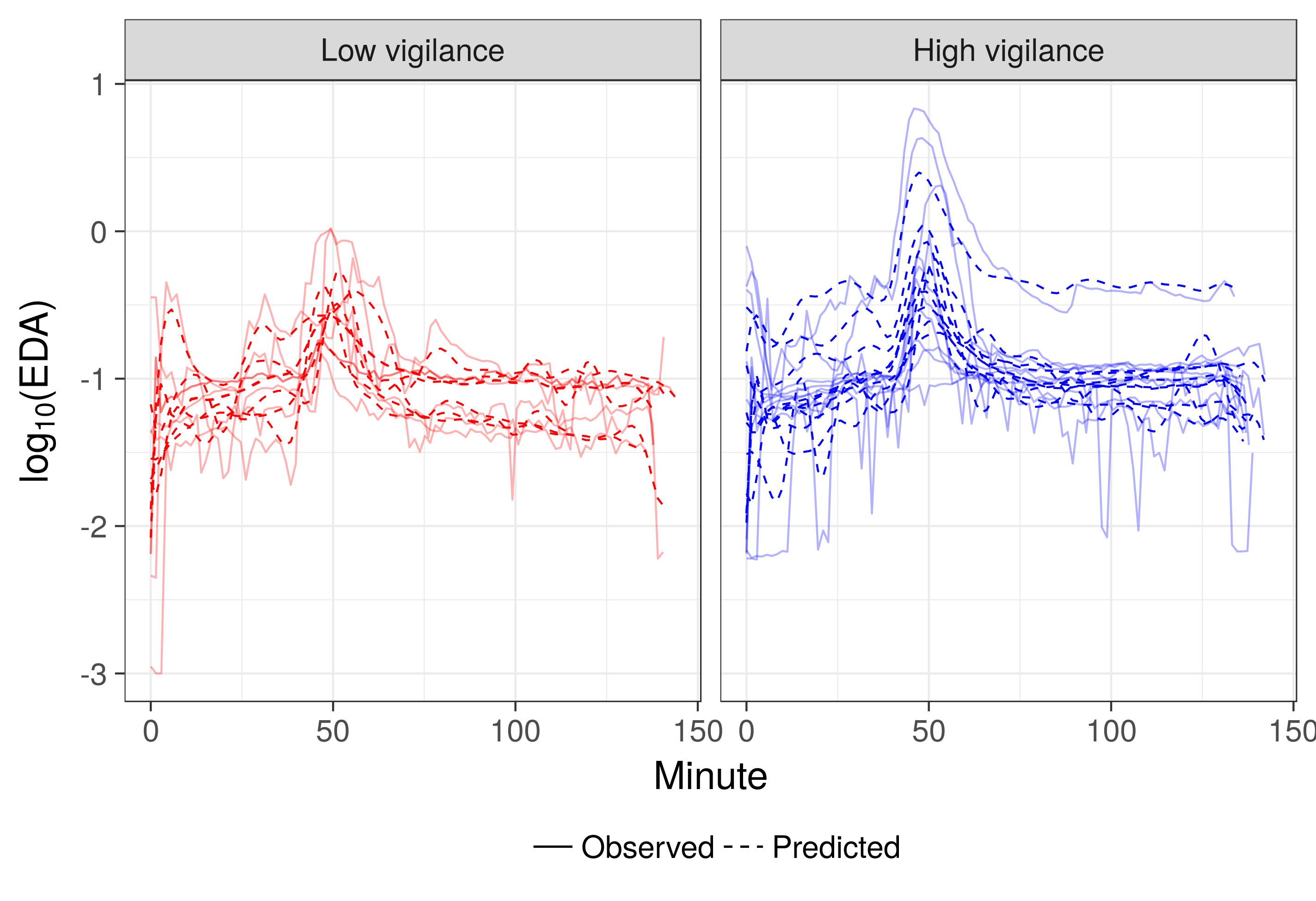}
  \caption{Bayesian model: subject-specific predicted curves}
  \label{EDA_Bayes_subj_curves}
\end{figure}

\subsection{Alternative correlation structure \label{cor_alt}}

For comparison, we also fit $\ell_1$ and $\ell_2$ penalized models with alternative correlation structures similar to that recommended by \citet[][p. 192]{ruppert2003semiparametric}.

For the $\ell_1$ penalized model, in place of the correlation structure implied by the penalty matrix $S$ described above, we set the penalty matrix to $S:= I_q$. While this is a simplification of the correlation structure recommended by \citet[][p. 192]{ruppert2003semiparametric}, we think it offers a similar amount of flexibility.

Figure \ref{EDA_L1_smooth_alt} shows the estimated marginal mean and 95\% credible bands, and Figure \ref{EDA_L1_subj_alt} shows the subject-specific predicted curves. The point estimates shown in Figure \ref{EDA_L1_smooth_alt} are similar to that shown in Figure \ref{EDA_L1_smooth}. However, the confidence intervals in Figure \ref{EDA_L1_smooth_alt} appear more reasonable. The subject-specific predicted curves shown in \ref{EDA_L1_subj_alt} are not shrunk towards the mean as much as in Figure \ref{EDA_L1_subj}.

\begin{figure}[H]
\centering
\begin{subfigure}{0.48\textwidth}
  \includegraphics[width = \linewidth]{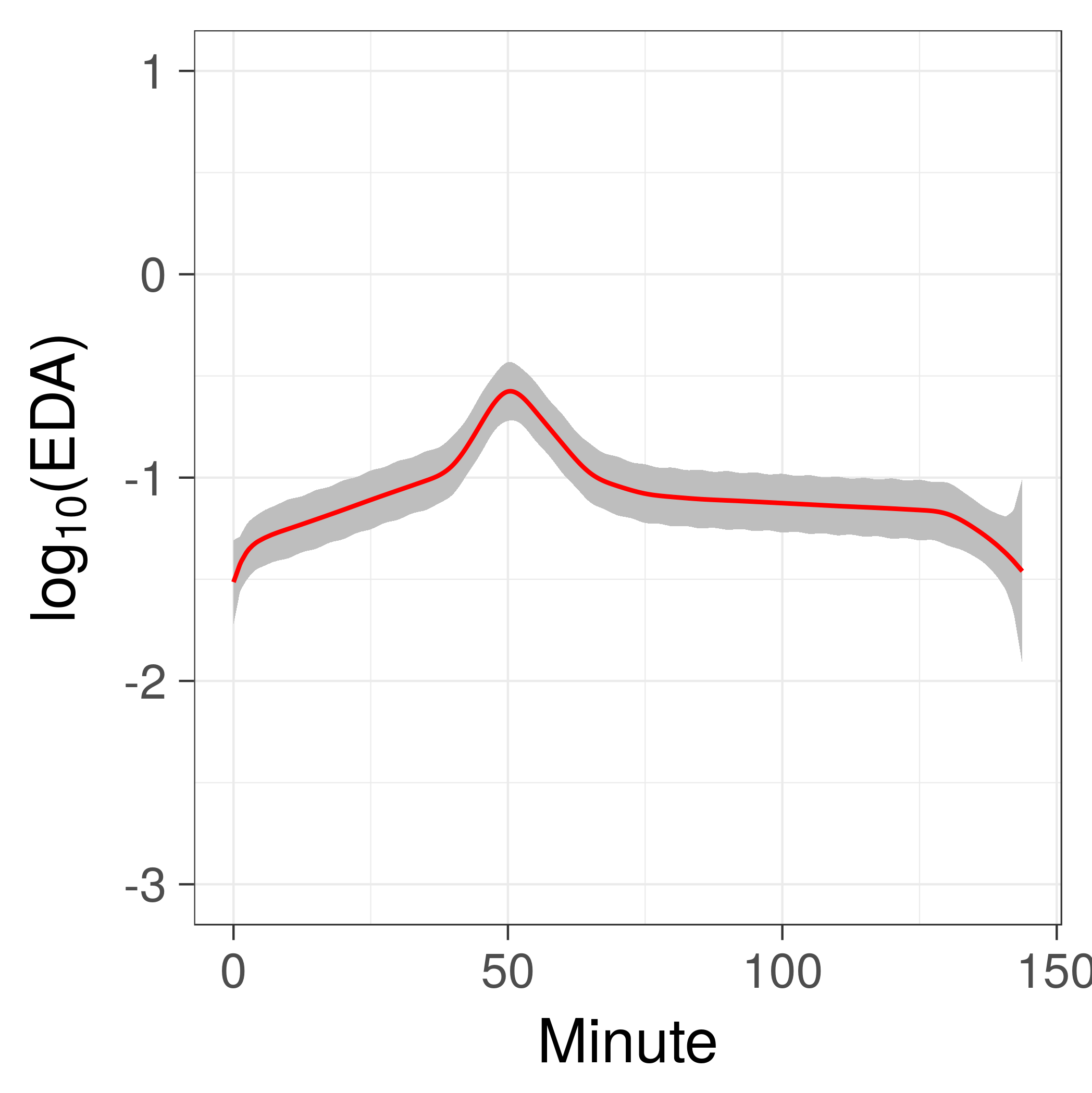}
  \caption{$\hat{\beta}_1(x)$ (low vigilance)}
  \label{EDA_L1_smooth1_poster_alt}
\end{subfigure}
\begin{subfigure}{0.48\textwidth}
  \includegraphics[width = \linewidth]{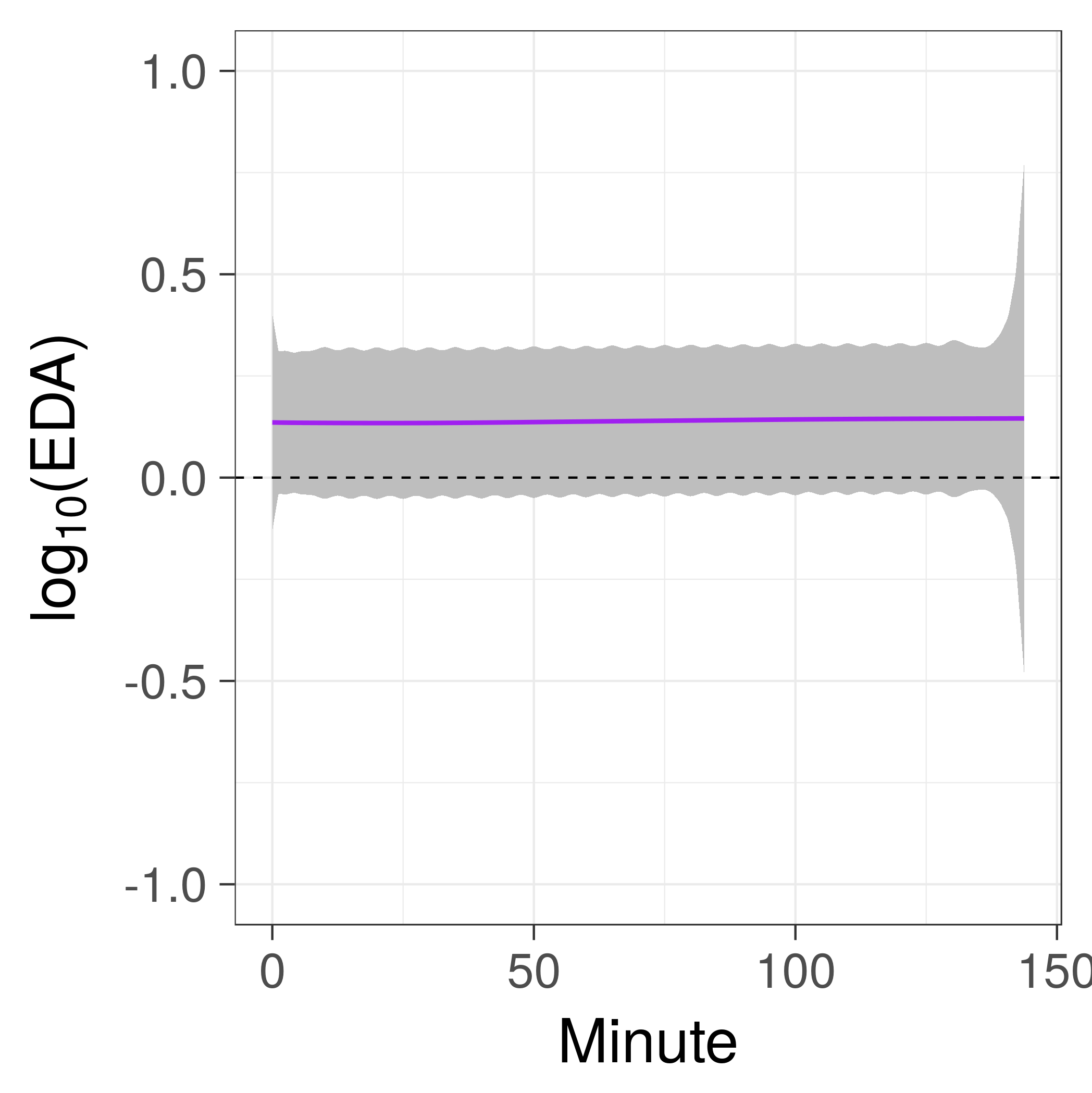}
  \caption{$\hat{\beta}_2(x)$ ($\text{high} - \text{low}$ vigilance)}
  \label{EDA_L1_smooth2_poster_alt}
\end{subfigure}
  \caption{$\ell_1$ penalized model with alternative correlation structure: parameter estimates with 95\% confidence bands}
  \label{EDA_L1_smooth_alt}
\end{figure}

\begin{figure}[H]
  \centering
  \includegraphics[scale = 0.48]{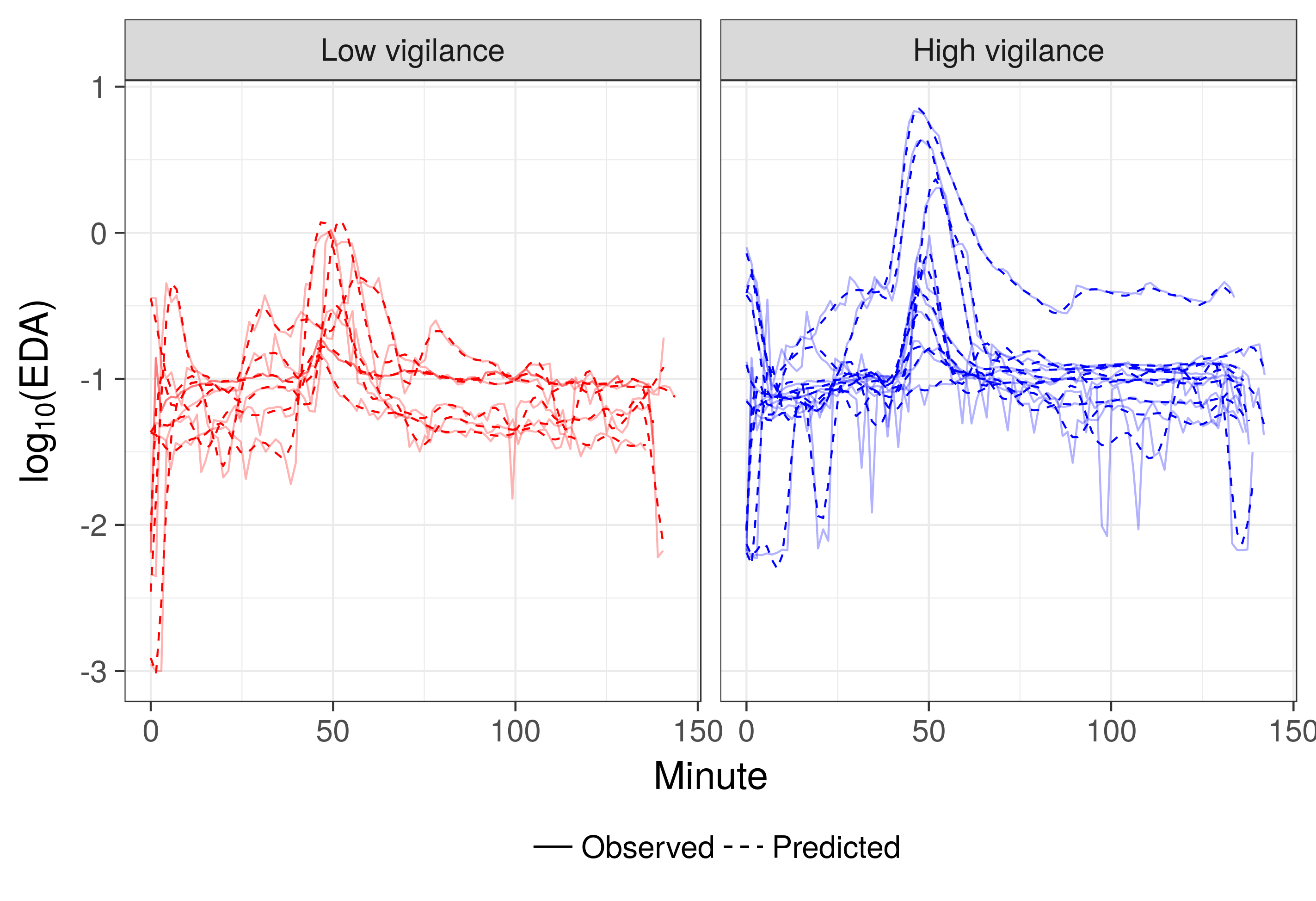}
  \caption{$\ell_1$ penalized model with alternative correlation structure: subject-specific predicted curves}
  \label{EDA_L1_subj_alt}
\end{figure}

For the $\ell_2$ Penalized model, in place of the correlation structure implied by the penalty matrix $S$ described above, we augmented each $Z_i$ matrix on the left with the columns $[\bm{1}, \bm{x}_i]$, where $\bm{x}_i$ is an $n_i \times 1$ vector of measurement times for subject $i$. We then replaced $Z_i \bm{b}_i$ with $[\bm{1}, \bm{x}_i, Z_i] (\bm{u}_i^T, \bm{b}_i^T)^T$, and assumed $(\bm{u}_i^T, \bm{b}_i^T)^T \sim N(\bm{0}, \Sigma_i)$ where 
$$
\Sigma_i = \begin{bmatrix}
\Sigma' \\
& \sigma_b^2 I_{q_i}
\end{bmatrix} 
$$
and $\Sigma'$ is a common $2 \times 2$ unstructured positive definite matrix. To model the within-subject correlations, we used a continuous autoregressive process of order 1. In particular, $\text{Cor}(y_i(x_{ij}), y_i(x_{ij'})) = \zeta^{|x_{ij} - x_{ij'}|}$ for a common parameter $\zeta>0$.

Figure \ref{EDA_GAMM_smooths_alt} shows the estimated marginal mean and 95\% credible bands, and Figure \ref{l2Pred} shows the subject-specific predicted curves. The estimates shown in Figure \ref{EDA_GAMM_smooths_alt} are similar to that shown in Figure \ref{EDA_GAMM_smooths}. While estimates of the difference between low and high vigilance subjects differs between this model and the $\ell_2$ penalized model in Section \ref{sec_res}, the more notable difference is in the subject-specific predicted curves. As seen in Figure \ref{l2Pred_alt}, the predicted subject-specific curves are not shrunk towards the mean as much as in Figure \ref{l2Pred}.

\begin{figure}[H]
\centering
\begin{subfigure}{0.48\textwidth}
  \includegraphics[width = \linewidth]{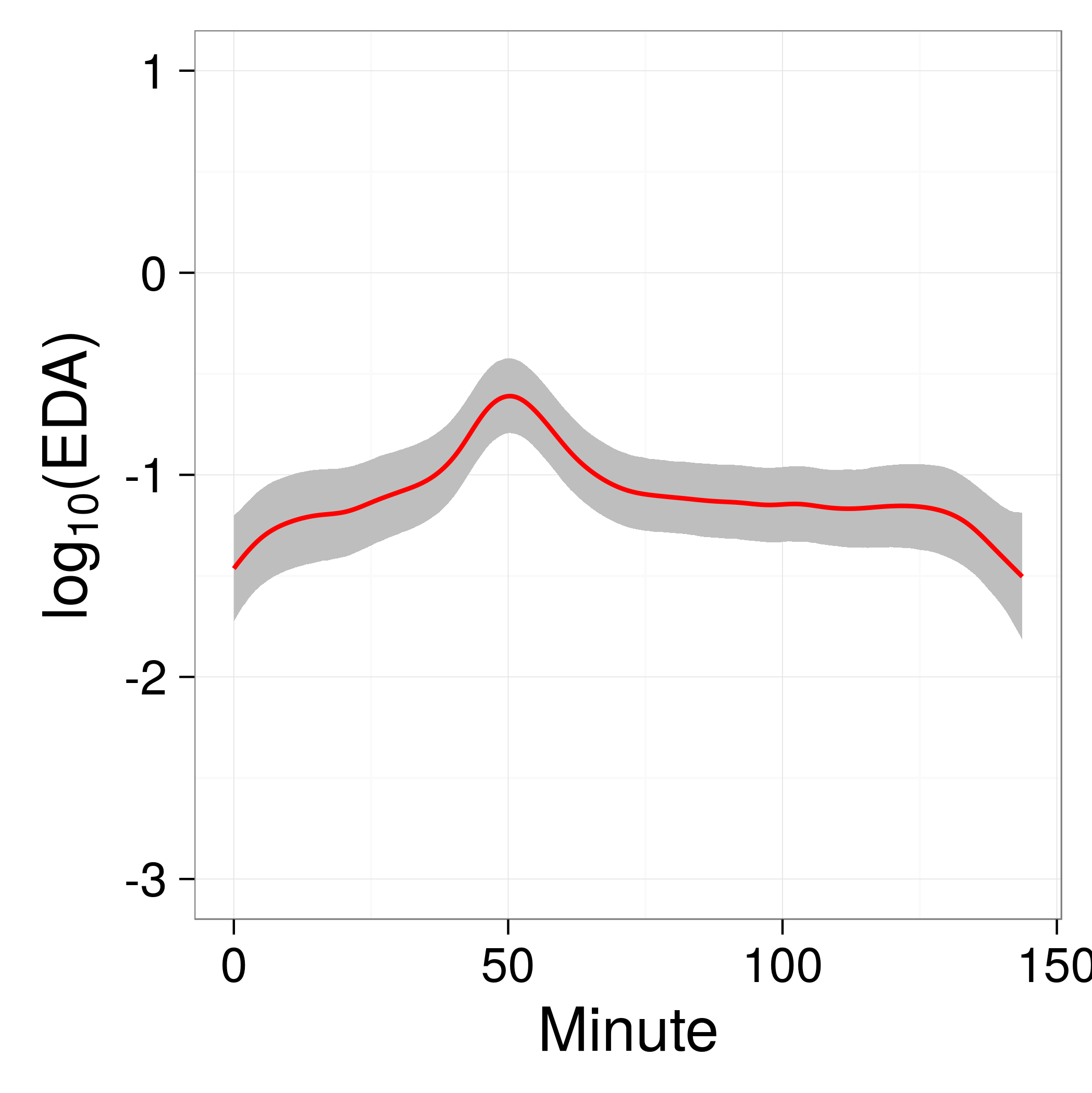}
  \caption{$\hat{\beta}_1(x)$ (low vigilance subjects)}
  \label{eda_l2_low_fit_alt}
\end{subfigure}
\begin{subfigure}{0.48\textwidth}
  \includegraphics[width = \linewidth]{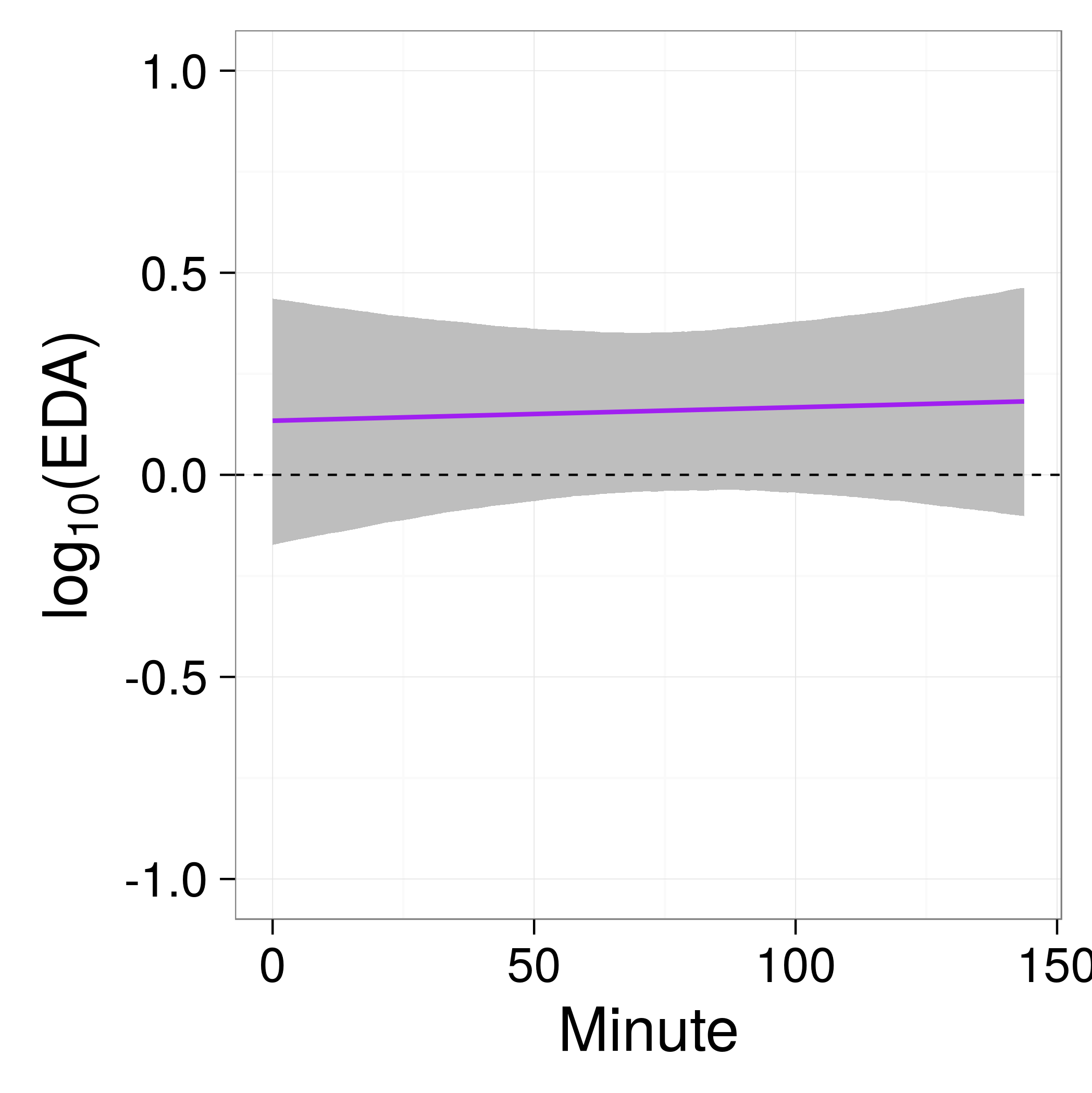}
  \caption{$\hat{\beta}_2(x)$ ($\text{high} - \text{low}$ vigilance subjects)}
  \label{eda_l2_diff_fit_alt}
\end{subfigure}
\caption{$\ell_2$ penalized model with alternative correlation structure: parameter estimates with 95\% confidence bands}
\label{EDA_GAMM_smooths_alt}
\end{figure}

\begin{figure}[H]
  \centering
  \includegraphics[scale = 0.48]{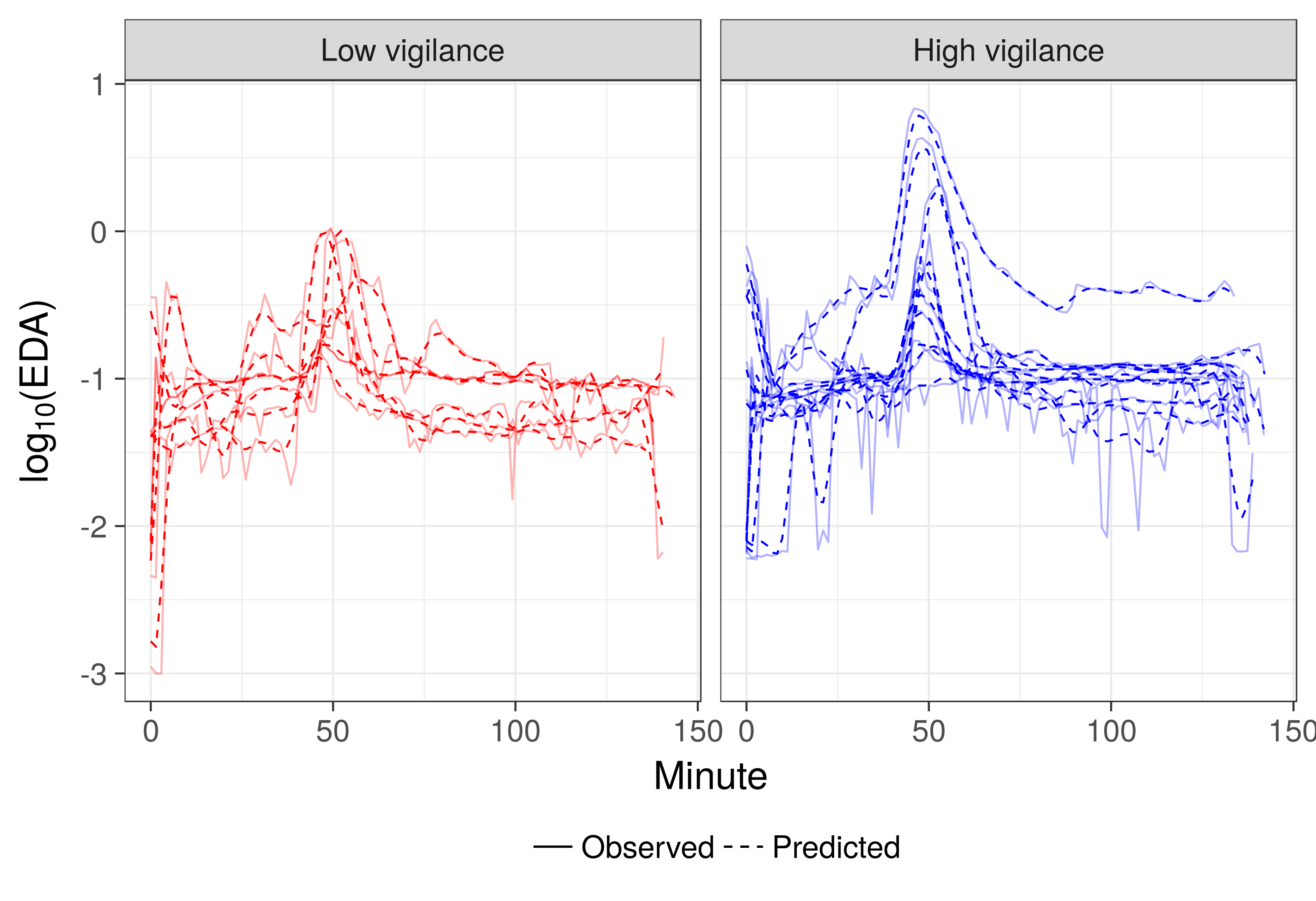}
  \caption{$\ell_2$ penalized model with alternative correlation structure: subject-specific predicted curves}
  \label{l2Pred_alt}
\end{figure}

Table \ref{corr_perf} shows the mean squared error (MSE) and computing time for the $\ell_1$ penalized and $\ell_2$ penalized models. In Table \ref{corr_perf}, computing time for the $\ell_1$ penalized model does not include cross-validation, because the parameters were hand-tuned (with only 17 subjects and a complex random effects structure, cross-validation did not find reasonable parameter values). As can be seen in Table \ref{corr_perf}, the alternative correlation structure led to smaller MSE for both the $\ell_1$ and $\ell_2$ penalized models, and less computing time for the $\ell_2$ penalized model.

\begin{table}[H]
\centering
\caption{MSE and computing time for different random curve correlation structures using the $\ell_1$ penalized and $\ell_2$ penalized models. ``Smoothing'' refers to smoothing splines used in Section \ref{sec_res} and ``Alternative'' refers to the correlation structures described in Section \ref{cor_alt}.}
\begin{tabular}{ccccc}
\hline \hline
 & \multicolumn{2}{c}{$\ell_1$ penalty} & \multicolumn{2}{c}{$\ell_2$ penalty} \\
 \cline{2-3} \cline{4-5}
 & Smoothing & Alternative & Smoothing & Alternative \\
MSE & 0.0195 & 0.00649 & 0.0348 & 0.00764 \\
Computing time (seconds) & $10.9^*$ & $11.2^*$ & 166 & 56.9\\
\hline
\multicolumn{5}{l}{$^*$Does not include cross-validation (parameters hand-tuned)}\\
\multicolumn{5}{l}{Note: Models fit on a laptop with Intel i7 quad CPUs at 2.67 GHz and 8 GB memory}
\end{tabular}
\label{corr_perf}
\end{table}

\section{Discussion and potential extensions \label{discussSec}}

As demonstrated in this article, P-splines with an $\ell_1$ penalty can be useful for analyzing repeated measures data. Compared to related work with $\ell_1$ penalties, our model is ambitious in that we allow for multiple smoothing parameters and propose approximate inferential procedures that do not require Bayesian estimation. However, these are also the two aspects of our proposed approach that require additional future work. For P-splines with an $\ell_2$ penalty, in most cases the knot placement is not critical so long as the number of knots is large enough \citep{ruppert2002, eilers2015}. We believe this also holds for P-splines with an $\ell_1$ penalty, though further experimentation is needed to support this assumption. In practice, we recommend fitting models with a few different knot placements and widths to determine whether the model is sensitive to those choices for the data at hand.

Regarding estimation, our current approach of using ADMM and CV appears to work reasonably well for random intercepts, but is not yet reliable for random curves. In the future, we plan to develop more robust estimation techniques, particularly for smoothing parameters. As one possibility, we have done preliminary work to minimize quantities similar to GCV and AIC instead of the more computationally intensive CV, though these approaches do not seem as promising as their $\ell_2$ counterparts. It may also be helpful to set the degrees of freedom prior to fitting the model. When possible, Bayesian estimation may be the most reliable way to currently fit these models. Bayesian estimation also opens the possibility of using other sparsity inducing priors, such as spike and slab models \citep{ishwaran2005}.

Regarding inference, in future work it may be possible to use the $\bm{\delta}$ quantity to bound difference between $\ell_1$ and $\ell_2$ penalized fits under certain assumptions on the data. It may also be helpful to investigate the use of post-selection inference methods to develop confidence bands for linear combinations of the active set, and to further investigate through simulations the performance of our proposed approximations of degrees of freedom. However, we note that our primary use of the degrees of freedom estimate $\hat{\text{df}}$ is to obtain the residual degrees of freedom $\hat{\text{df}}_{\text{resid}} = n - \hat{\text{df}}$, which we then use to estimate the variances $\hat{\sigma}^2_\epsilon = \| \bm{r} \|_2^2 / \hat{\text{df}}_{\text{resid}}$. Therefore, when $n \gg \hat{\text{df}}$, $\hat{\sigma}^2_\epsilon$ is not very sensitive to $\hat{\text{df}}$, in which case it is not critical for our purposes to obtain an exact estimate of degrees of freedom.

As for P-splines with an $\ell_2$ penalty, users must select both the order $M$ of the B-splines and the order $k+1$ of the finite differences. These choices will depend on the scientific problem and analytical goals. Using $k=1$ ($2^{nd}$ order differences) is likely an appropriate starting point for most applications, and larger $k$ could be used to increase the amount of smoothness.

For P-splines with an $\ell_2$ penalty, in most cases the knot placement is not critical so long as the number of knots is large enough \citep{ruppert2002, eilers2015}. We believe this also holds for P-splines with an $\ell_1$ penalty, though further experimentation is needed to support this assumption. In practice, we recommend fitting models with a few different knot placements and widths to determine whether the model is sensitive to those choices for the data at hand.

Regarding the rate of convergence, from Observation \ref{equiv} and the work of \citet{tibshirani2014adaptive}, we know that for equally spaced data and $F = I_{n}$, P-splines with an $\ell_1$ penalty achieve the minimax rate of convergence for the class of weakly differentiable functions of bounded variation. When there are less knots than data points, we do not think it is possible to achieve the minimax rate of convergence. However, if the knots are selected well, it may be possible to achieve the same performance in practice.

It could also be useful to extend these results to generalized additive models to allow for non-normal responses, and to extend the approach of \citet{sadhanala2017} to include random effects and multiple smoothing parameters.

\section{Supplementary material}

We have implemented our method in the \text{R} package \verb|psplinesl1| available at \url{https://github.com/bdsegal/psplinesl1}. All code for the simulations and analyses in this paper are available at \url{https://github.com/bdsegal/code-for-psplinesl1-paper}.

\section*{Acknowlegements}

We thank Margaret Hicken for sharing the data from the stress study.

\begin{appendices}

\section{Simulated demonstration with two smooths}

In this appendix, we simulated data similar to that in Section \ref{sec_sim}, but with an additional varying-coefficient smooth. In particular, we simulated data for two groups with 50 subjects in each group and between 4 and 14 measurements per subject (900 total observations). The data for subject $i$ at time $t$ was generated as $y_{it} = \beta_0 + f_1(x_{it}) + f_2(x_{it}) \mathbbm{1}[\text{subject } i \text{ in Group 2}] + b_i + \epsilon_{it}$ where $b_i \sim N(0, \sigma_b^2)$ and $\epsilon_{it} \sim N(0, \sigma^2_{\epsilon})$ for $\sigma_b^2 = 1$ and $\sigma_\epsilon^2 = 0.01$. The true group means $f_1(x)$ and $f_2(x)$ are shown in Figure \ref{nonDiff_2groups_truth} and the simulated data are shown in Figure \ref{nonDiff_2groups_facet}.

\begin{figure}[H]
\centering
\includegraphics[scale = 0.48]{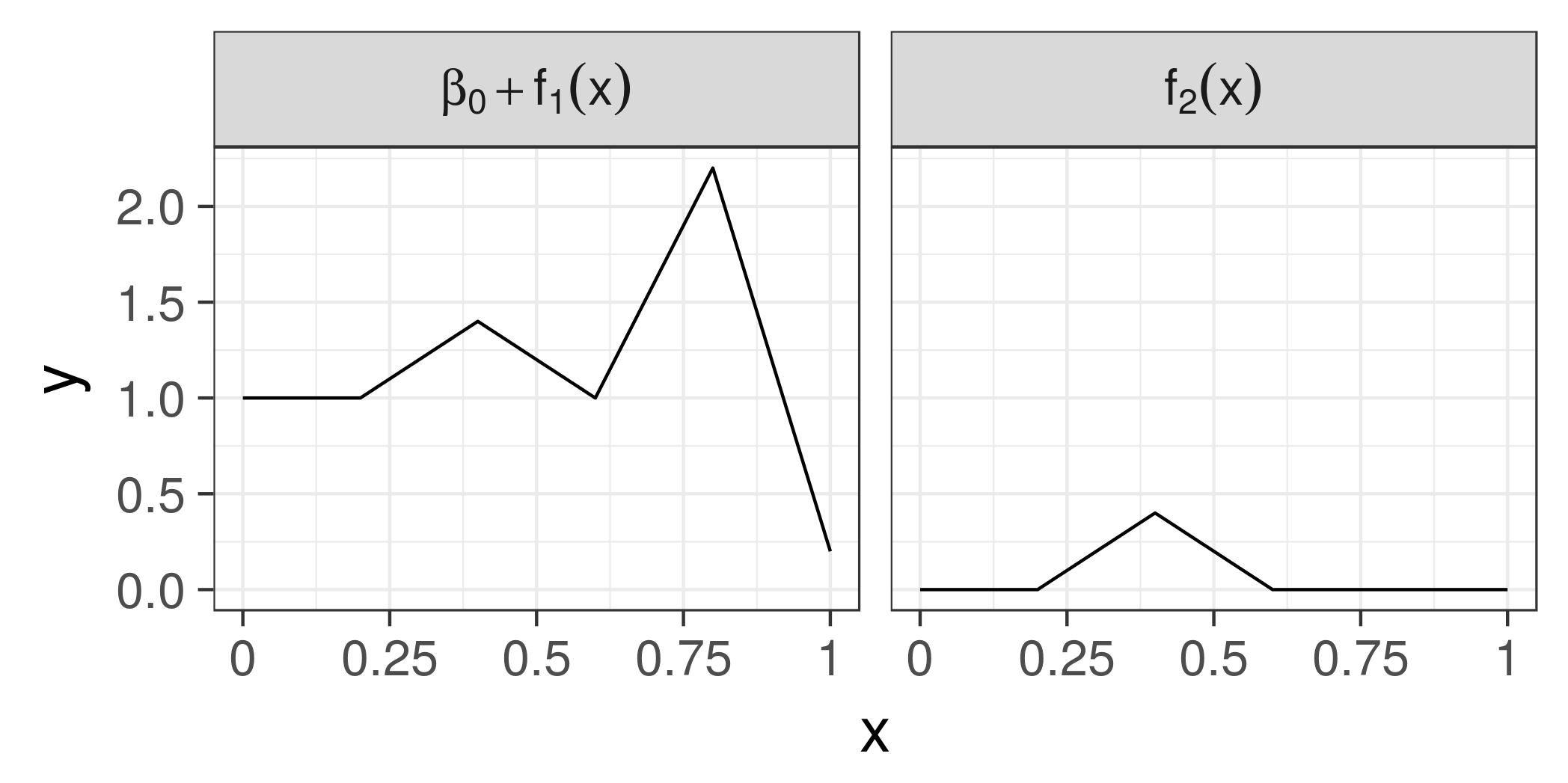}
\caption{True means}
\label{nonDiff_2groups_truth}
\end{figure}

\begin{figure}[H]
\centering
\includegraphics[scale = 0.48]{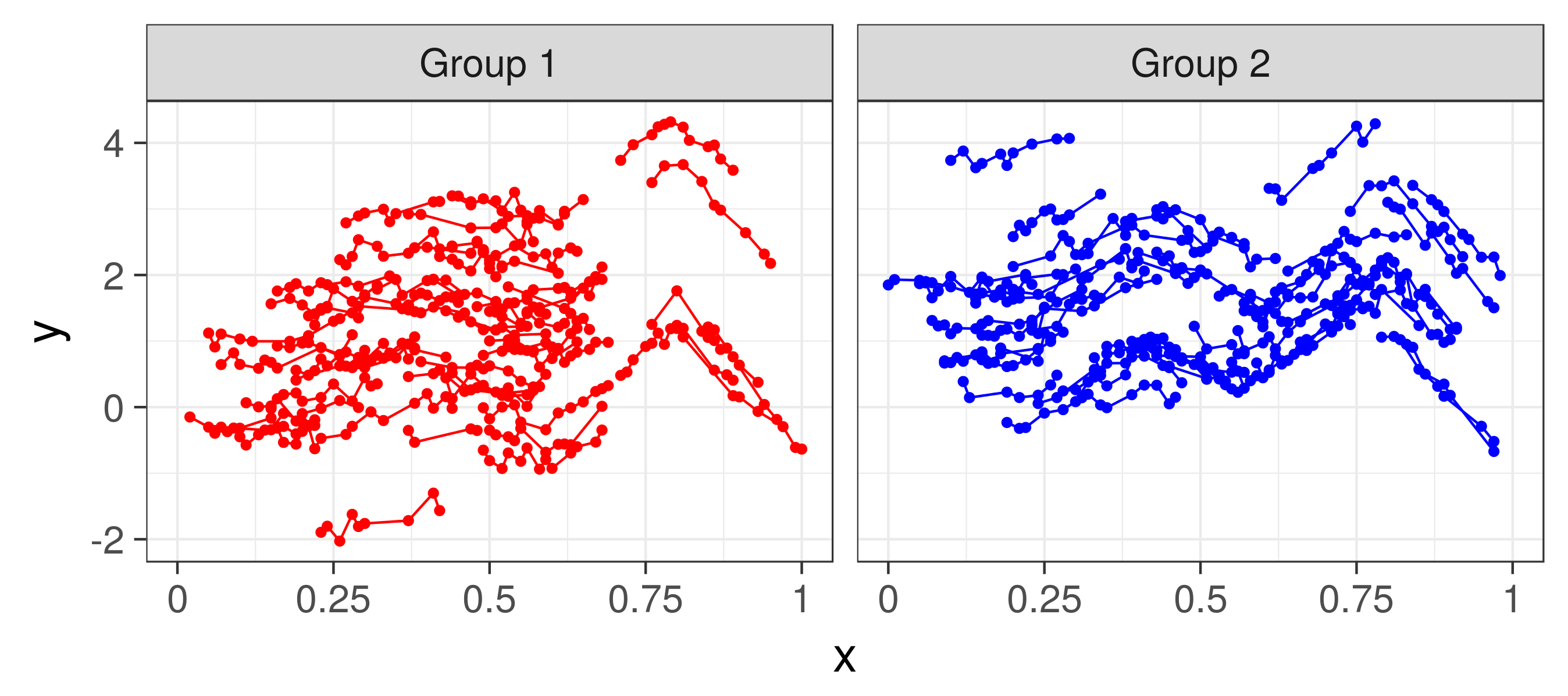}
\caption{Simulated data}
\label{nonDiff_2groups_facet}
\end{figure}

We fit a varying-coefficient model with $J=2$ smooths to the data. In particular, we used ADMM and 5-fold CV to minimize
\begin{align}
\underset{\beta_0 \in \mathbb{R}, \bm{\beta}_1 \in \mathbb{R}^{p-1}, \bm{\beta}_2 \in \mathbb{R}^{p}, \bm{b} \in \mathbb{R}^{N}}{\text{minimize}} & \frac{1}{2} \| \bm{y} - \beta_0 \bm{1} - F_1 \bm{\beta}_1 - F_2 \bm{\beta}_2 - Z \bm{b} \|_2^2 \nonumber \\
& + \lambda_1 \|D^{(2)} \bm{\beta}_1 \|_1 + \lambda_2 \|D^{(2)} \bm{\beta}_2 \|_1 + \tau \bm{b}^T \bm{b}. \label{sim_obj_2}
\end{align}
where $F_1, F_2 \in \mathbb{R}^{n \times p}$ were formed with second order (first degree) B-splines and $p=21$ basis functions, $F_2 = \text{diag}(\bm{u}) F_1$ where $u_i = \mathbbm{1}[\text{subject } i \text{ in Group 2}]$, and $Z_{il} = 1$ if observation $i$ belongs to subject $l$ and zero otherwise. We also fit an equivalent model with an $\ell_2$ penalty using the \verb|mgcv| package \citep{wood2006generalized}, i.e. with $(\lambda_j / 2) \|D^{(2)} \bm{\beta}_j \|_2^2$ in place of $\lambda_j \|D^{(2)} \bm{\beta}_j \|_1$ in (\ref{sim_obj_2}), $j = 1, 2$.

The estimated curves are shown in Figure \ref{l1_2groups_group} for the $\ell_1$ penalized model and in Figure \ref{l2_2groups_group} for the $\ell_2$ penalized model. We used 5-fold CV to estimate the smoothing parameters $\lambda_1, \lambda_2$ and $\tau$ in the $\ell_1$ penalized model, and LME updates to estimate $\sigma^2_b$ and $\bm{b}$ in the final model. As seen in Figures \ref{l1_2groups_group} and \ref{l2_2groups_group}, the fits are similar, but the results with the $\ell_1$ penalized model are slightly closer to the truth.

\begin{figure}[H]
\centering
  \begin{subfigure}{0.48\linewidth}
    \includegraphics[width=\textwidth]{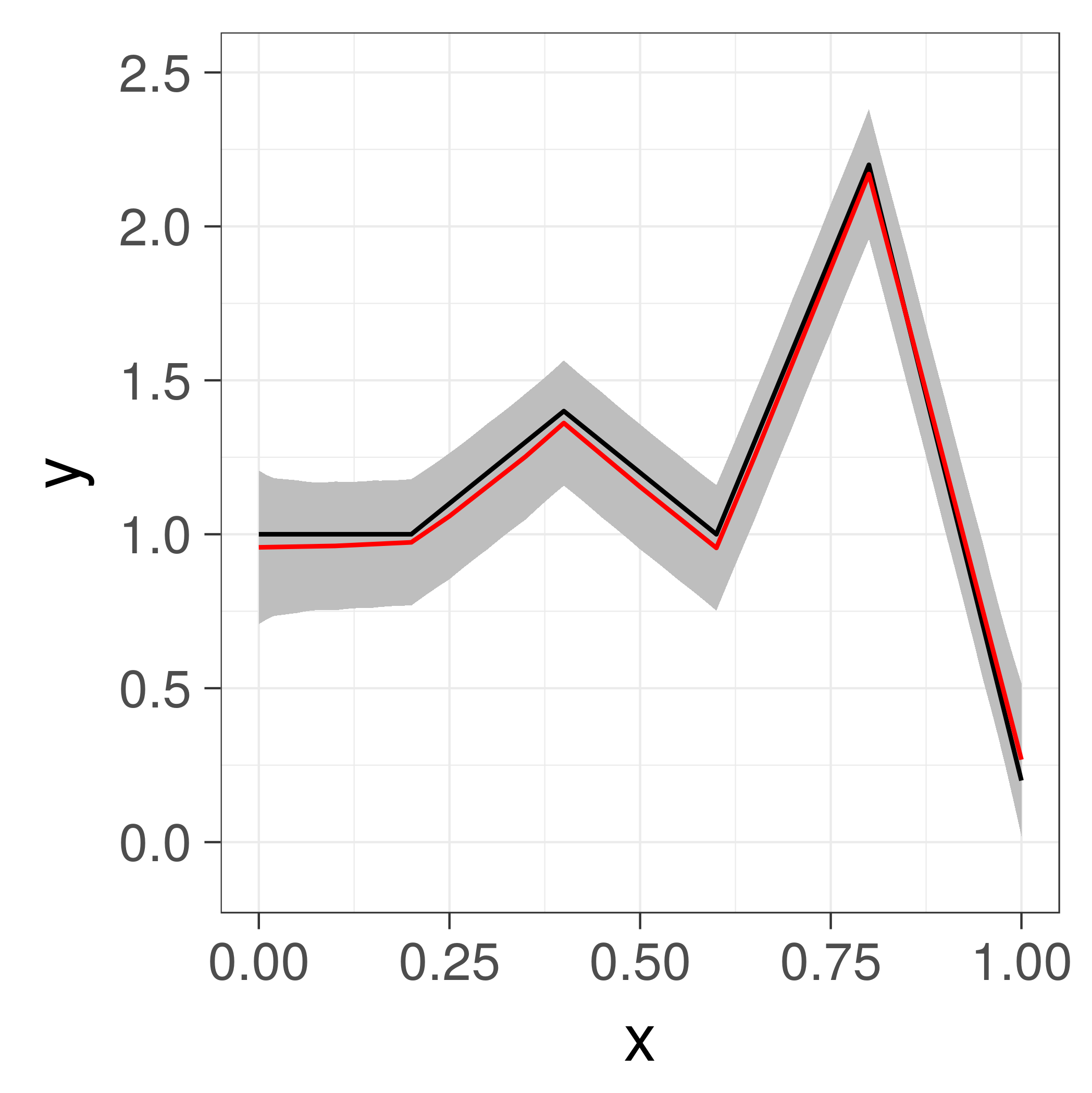}
    \caption{$\hat{\beta}_0 + \hat{f}_1(x)$}
  \end{subfigure}
  \begin{subfigure}{0.48\linewidth}
    \includegraphics[width=\textwidth]{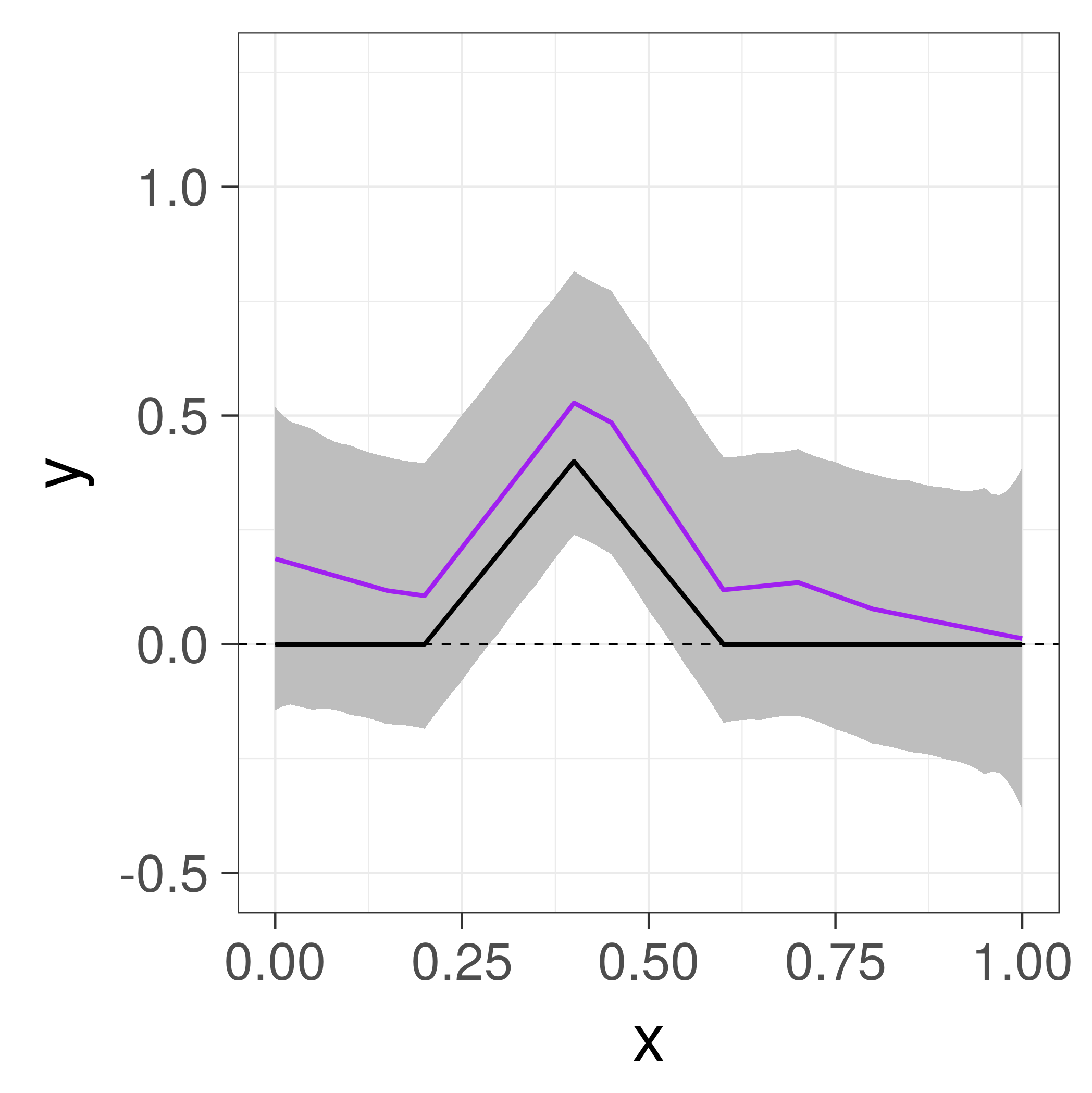}
    \caption{$\hat{f}_2(x)$}
  \end{subfigure}
  \caption{Marginal mean and 95\% credible intervals from $\ell_1$ penalized model fit with ADMM and CV: black is true marginal mean, red is estimated marginal mean}
  \label{l1_2groups_group}
\end{figure}

\begin{figure}[H]
\centering
  \begin{subfigure}{0.48\linewidth}
    \includegraphics[width=\textwidth]{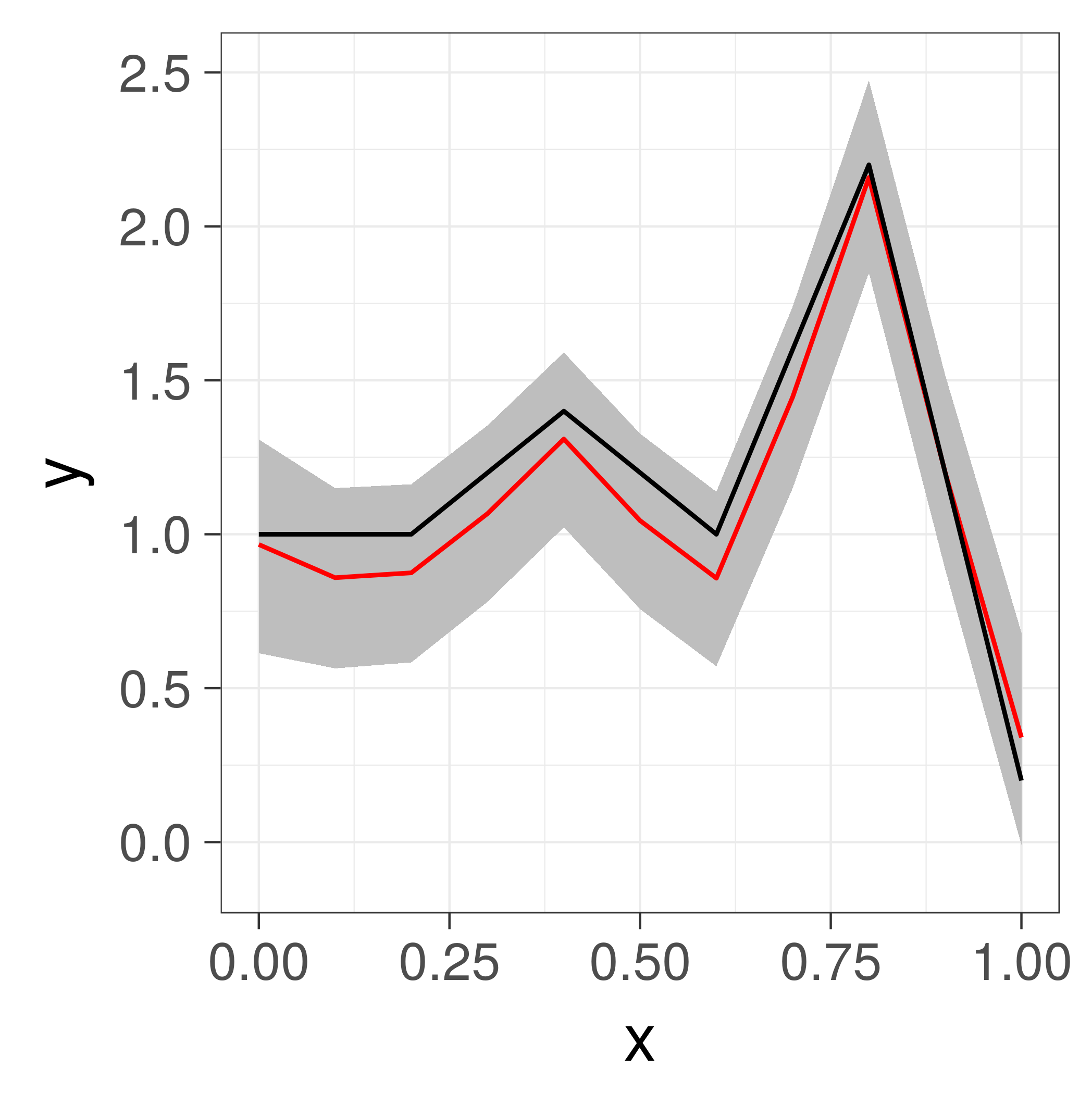}
    \caption{$\hat{\beta}_0 + \hat{f}_1(x)$}
  \end{subfigure}
  \begin{subfigure}{0.48\linewidth}
    \includegraphics[width=\textwidth]{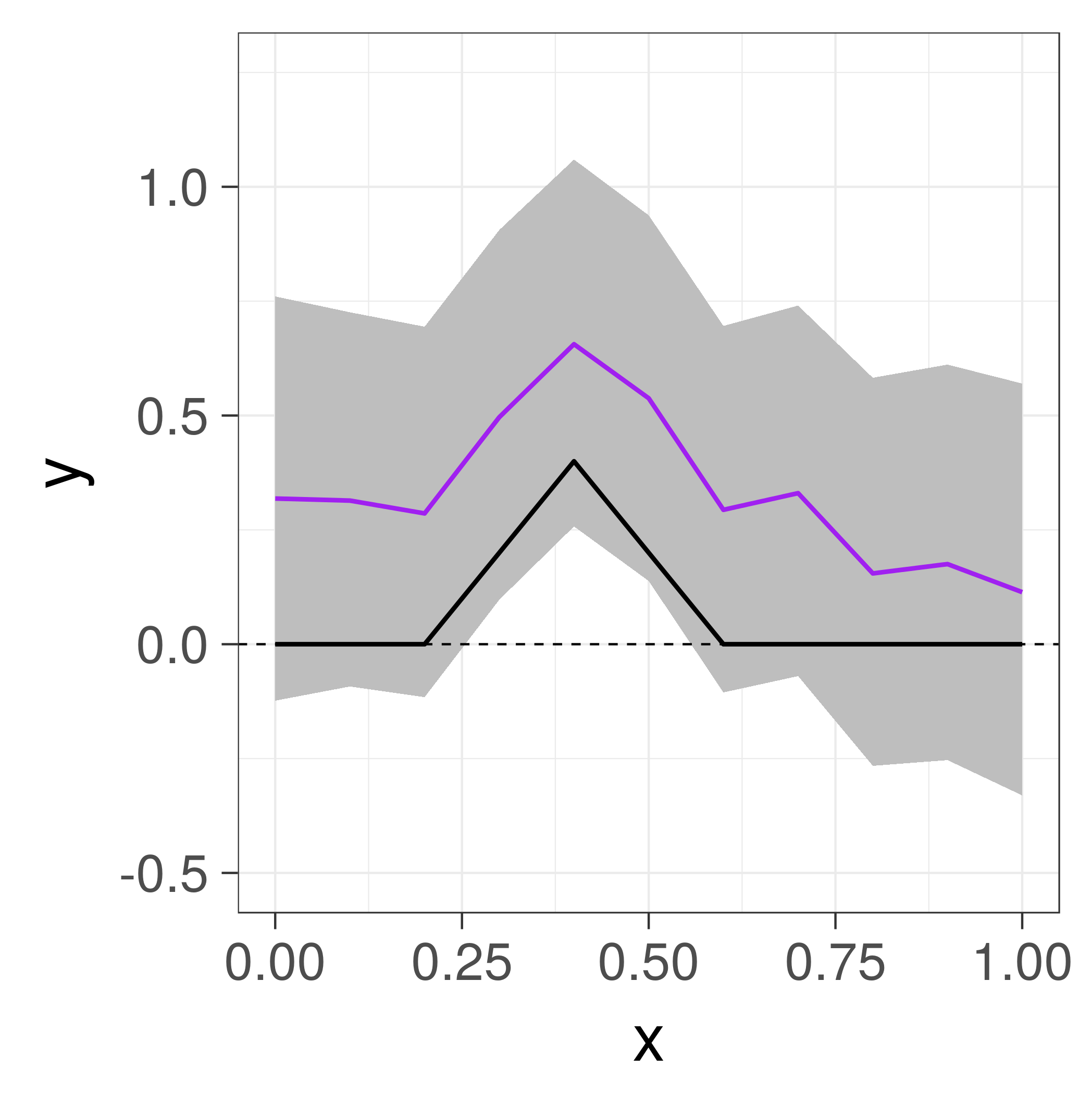}
    \caption{$\hat{f}_2(x)$}
  \end{subfigure}
  \caption{Marginal mean and 95\% credible intervals from $\ell_2$ penalized model fit with mgcv \citep{wood2006generalized}: black is true marginal mean, red is estimated marginal mean}
  \label{l2_2groups_group}
\end{figure}

Table \ref{2group_table} shows the degrees of freedom and variance estimates with the $\ell_1$ penalized and $\ell_2$ penalized models. As seen in Table \ref{2group_table}, variance estimates from both the $\ell_1$ and $\ell_2$ penalized models are very near the true values.

\begin{table}[H]
\centering
  \caption{Degrees of freedom and variance in $\ell_1$ and $\ell_2$ penalized models}
  \begin{tabular}{cccccc}
    \hline \hline
    & \multicolumn{4}{c}{Penalty} \\
    \cline{2-5}
    & \multicolumn{2}{c}{$\ell_1$} & \multicolumn{2}{c}{$\ell_2$} & \\
    \cline{2-3} \cline{4-5}
    & $j = 1$ & $j = 2$ & $j = 1$ & $j = 2$ & Truth\\
    \hline
    df (ridge) & 17.7 & 17.8 & 19.3 & 13.8 & --\\
    df (Stein) & 12 & 9 & -- & -- & -- \\
    $\hat{\sigma}^2_\epsilon$ & \multicolumn{2}{c}{0.0090} & \multicolumn{2}{c}{0.010} & 0.01\\
    $\hat{\sigma}^2_b$ & \multicolumn{2}{c}{1.04} & \multicolumn{2}{c}{1.02} & 1
  \end{tabular}
  \label{2group_table}
\end{table}

\section{Details for $\lambda_j^{\max}$ \label{lambdaMaxDetails}}

Letting $\bm{r}_j = \bm{y} - \beta_0 \bm{1} - \sum_{\ell \ne j} F_{\ell} \bm{\beta}_{\ell}  - Z \bm{b}$ be the $j^{th}$ partial residuals, we can write the terms in (\ref{l1PsplineAMM}) that involve $\bm{\beta}_j$ as $(1/2) \|\bm{r}_j - F_j \bm{\beta}_j\|_2^2 + \lambda_j \|D_j \bm{\beta}_j \|_1$. Then taking the sub-differential of (\ref{l1PsplineAMM}) with respect to $\bm{\beta}_j$, we have
\begin{align}
\bm{0} &= -F_j^T(\bm{r}_j - F_j \bm{\hat{\beta}}_j) + D_j^T \lambda_j \bm{s}_j
\label{subDiff}
\end{align}
for some $\bm{s}_j = (s_{j, 1},\ldots,s_{j, p_j - k_j - 1})^T$ where
\begin{equation*}
s_{j,\ell} \in 
\begin{cases}
\{1\} & \text{ if } (D \bm{\hat{\beta}}_j)_\ell > 1 \\
\{-1\} & \text{ if } (D \bm{\hat{\beta}}_j)_\ell < 1 \\
[-1, 1] & \text{ if } (D \bm{\hat{\beta}}_j)_\ell = 0.
\end{cases}
\end{equation*}
Solving (\ref{subDiff}) for $\bm{\hat{\beta}}_j$, we have $\bm{\hat{\beta}}_j = (F_j^T F_j)^{-1} F_j^T \bm{r}_j - D_j^T \lambda_j \bm{s}_j$. Multiplying through by $D_j$ and noting that $D_j D_j^T$ is full rank and thus invertible, we have
\begin{equation}
(D_j D_j^T)^{-1} D_j \bm{\hat{\beta}}_j = (D_j D_j^T)^{-1} D_j (F_j^T F_j)^{-1} F_j^T \bm{r}_j - \lambda_j \bm{s}_j.
\label{subDiffSolved}
\end{equation}
Setting $D_j \bm{\hat{\beta}}_j = \bm{0}$ in (\ref{subDiffSolved}), we get that $(D_j D_j^T)^{-1} D_j (F_j^T F_j)^{-1} F_j^T \bm{r}_j = \lambda_j \bm{s}_j$ where $s_{j, \ell} \in [-1, 1]$ for all $\ell$. This can only hold if $\lambda_j = \|(D_j D_j^T)^{-1} D_j (F_j^T F_j)^{-1} F_j^T \bm{r}_j\|_\infty$, which gives us $\lambda_j^{\max}$.

\section{Controlling total variation with the $\ell_1$ penalty \label{controlSmooth}}

Let $f(x) = \sum_{j=1}^p \beta_j \phi_j^M(x)$. Suppose the knots are equally spaced, and let $h_{M-k-1} = (M-k-1) / (t_{j+M-k-1} - t_{j})$ for all $j$ and $0 \le k < M - 1$. Then on the interval $[t_M = x_{\min}, t_{p+1} = x_{\max}]$, from \cite[][p. 117]{de1978practical} we have
\begin{align}
\frac{d^{k+1}}{dx^{k+1}}f(x) &= h_{M-1} \cdots h_{M-k-1} \sum_{j=k+2}^{p} \nabla^{k+1} \beta_j \phi_j^{M-k-1}(x)
\label{deriv}
\end{align}
where $\nabla^{k+1}$ is the $(k+1)^{th}$ order backwards difference.

Let $a^M_{k+1} = \max_{j \in \{k + 2,\ldots p \} } \int_{x_{\min}}^{x_{\max}} \phi_j^{M-k-1}(x) dx$. We note that $a^M_{k+1}$ is finite and positive for all $0 \le k < M-1$. Then from (\ref{deriv}), we have
\begin{align}
\frac{1}{h_{M - 1} \cdots h_{M-k-1}} & \int_{x_{\min}}^{x_{\max}} \left| \frac{d^{k+1}}{dx^{k+1}} f(x) \right| dx \nonumber \\
&=  \int_{x_{\min}}^{x_{\max}} \left| \sum_{j=k+2}^p \nabla^{k+1} \beta_j \phi_j^{M-k-1}(x) \right|  dx \nonumber \\
&= \int_{x_{\min}}^{x_{\max}} \left| \sum_{j=k+2}^p (D^{(k+1)} \bm{\beta})_{j-k-1} \phi_j^{M-k-1}(x) \right| dx \nonumber \\
&\le \int_{x_{\min}}^{x_{\max}} \sum_{j=k+2}^p \left| (D^{(k+1)} \bm{\beta})_{j-k-1} \phi_j^{M-k-1}(x) \right| dx \nonumber \\
&= \sum_{j=k+2}^p \int_{x_{\min}}^{x_{\max}} \left| (D^{(k+1)} \bm{\beta})_{j-k-1} \phi_j^{M-k-1}(x) \right| dx \nonumber \\
&= \sum_{j=k+2}^p \left| (D^{(k+1)} \bm{\beta})_{j-k-1} \right| \int_{x_{\min}}^{x_{\max}} \phi_j^{M-k-1}(x) dx \label{aline} \\
&\le a^M_{k+1} \sum_{j=k+2}^p \left| (D^{(k+1)} \bm{\beta})_{j-k-1} \right| \nonumber \\
&= a^M_{k+1} \| D^{(k+1)} \bm{\beta} \|_1 \label{result}
\end{align}
where (\ref{aline}) follows because $\phi_j^{M-k-1}(x) \ge 0 \; \forall x \in \mathbb{R}$.

Rewriting (\ref{result}), for $0 \le k < M-1$ we have
\begin{equation*}
\int_{x_{\min}}^{x_{\max}} \left| \frac{d^{k+1}}{dx^{k+1}} f(x) \right| dx \le C_{M,k+1} \|D^{(k+1)} \bm{\beta} \|_1
\end{equation*}
where $C_{M, k+1} = a^M_{k+1} h_{M-1} \cdots h_{M-k-1}$ is a constant. This shows that controlling the $\ell_1$ norm of the $(k+1)^{th}$ order finite differences in coefficients also controls the total variation of the $k^{th}$ derivative of the function.
\end{appendices}

\bibliographystyle{apalike}
\bibliography{refsBib}

\end{document}